%% file: main.tex
\newif\ifdouble

\doubletrue
\ifdouble
\documentclass[journal,a4paper]{IEEEtran}
\else
\documentclass[12pt, draftclsnofoot, onecolumn]{IEEEtran}
\fi
\newif\ifleakageD
\leakageDtrue

\usepackage{amsthm}
\usepackage{latexsym}
\usepackage{amssymb}
\usepackage{graphicx}
\usepackage{subfigure}
\usepackage{amsmath}
\usepackage{color}
\usepackage{cite}
\usepackage{cleveref}
\usepackage[english]{babel}
\usepackage{multicol}
\usepackage{nicefrac}
\usepackage{comment}
\usepackage{algpseudocode}
\usepackage{algorithm}
\usepackage{tikz}
\usetikzlibrary{math}
\usepackage{ifthen}
\usepackage{footmisc} 

\theoremstyle{plain}
\newtheorem{theorem}{Theorem}

\newtheorem{prop}{Proposition}
\newtheorem{corollary}{Corollary}
\theoremstyle{definition}
\newtheorem{definition}{Definition}
\theoremstyle{remark}

\newcommand{\off}[1]{}
\newcommand{\ar}[1]{\textcolor[rgb]{0.00,0.00,0.00}{#1}}

\crefformat{footnote}{#2\footnotemark[#1]#3}
\newcommand{\rev}[1]{\textcolor[rgb]{0.00,0.00,0.00}{#1}}
\newcommand{\jr}[1]{\textcolor[rgb]{0.00,0.00,0.00}{#1}}
\newcommand{\revGui}[1]{\textcolor[rgb]{0.00,0.00,1.00}{#1}}
\newcommand{\jrev}[1]{\textcolor[rgb]{0.00,0.00,0.00}{#1}}
\newcommand{\jrevm}[1]{\textcolor[rgb]{0.00,0.00,0.00}{#1}} 
\ifCLASSINFOpdf
\else
\fi

\begin{document}
\ifdouble
\title{Adaptive Causal Network Coding with Feedback\\ for Multipath Multi-hop Communications}
\markboth{}{}
 \author{%
   \IEEEauthorblockN{Alejandro Cohen\IEEEauthorrefmark{1},
                     Guillaume Thiran\IEEEauthorrefmark{2},
                     Vered Bar Bracha\IEEEauthorrefmark{3},
                     and Muriel M\'edard\IEEEauthorrefmark{1}}\\
   \IEEEauthorblockA{\IEEEauthorrefmark{1}%
                     Research Laboratory of Electronics, MIT, Cambridge, MA, USA, \{cohenale, medard\}@mit.edu}\\
   \IEEEauthorblockA{\IEEEauthorrefmark{2}%
                     UCLouvain, Belgium,  guillaume.thiran@student.uclouvain.be}\\
   \IEEEauthorblockA{\IEEEauthorrefmark{3}%
                     Intel Corporation, vered.bar.bracha@intel.com}\thanks{This research was supported in part by the Intel Corporation and by DARPA: DFARS 252.235-7010. Patent application submitted: no. 62/853,090. Part of this work is accepted for publication in IEEE International Conference on Communications (ICC), 2020. Accepted for publication in IEEE Transactions on Communications.}\\
 \vspace{-8mm}}
 \else
 \title{Adaptive Causal Network Coding with Feedback\\ for Multipath Multi-hop Communications}
\markboth{}{}
 \author{%
   \IEEEauthorblockN{Alejandro Cohen\IEEEauthorrefmark{1},
                     Guillaume Thiran\IEEEauthorrefmark{2},
                     Vered Bar Bracha\IEEEauthorrefmark{3},
                     and Muriel M\'edard\IEEEauthorrefmark{1}}\\
   \IEEEauthorblockA{\IEEEauthorrefmark{1}%
                     RLE, MIT, Cambridge, MA, USA, \{cohenale, medard\}@mit.edu}\\
   \IEEEauthorblockA{\IEEEauthorrefmark{2}%
                     UCLouvain, Belgium,  guillaume.thiran@student.uclouvain.be}\\
   \IEEEauthorblockA{\IEEEauthorrefmark{3}%
                     Intel Corporation, vered.bar.bracha@intel.com}\thanks{This research was supported in part by the Intel Corporation and by DARPA: DFARS 252.235-7010. Patent application submitted: no. 62/853,090. Part of this work is accepted for publication in IEEE International Conference on Communications (ICC), 2020.}\\
 \vspace{-10mm}}
 \fi
\maketitle


\begin{abstract}
We propose a novel multipath multi-hop adaptive and causal random linear network coding (AC-RLNC) algorithm with forward error correction. This algorithm generalizes our joint optimization coding solution for point-to-point communication with delayed feedback. AC-RLNC is adaptive to the  estimated channel condition, and is causal, as the coding adjusts the retransmission rates using a priori and posteriori algorithms. In the multipath network, to achieve the desired throughput and delay, we propose to incorporate an adaptive packet allocation algorithm for retransmission, across the available resources of the paths. This approach is based on a discrete water filling algorithm, i.e., bit-filling, \rev{but, with two \emph{desired objectives}}, maximize throughput and minimize the delay. In the multipath multi-hop setting, we propose a new decentralized balancing optimization algorithm. This balancing algorithm minimizes the throughput degradation, caused by the variations in the channel quality of the paths at each hop. \off{called bottleneck effects.This balancing minimizes the throughput degradation, which is due to the bottlenecks caused by the variations in the channel quality of the paths at each hop.} Furthermore, to increase the efficiency, in terms of the \rev{desired} objectives, we propose a new selective recoding method at the intermediate nodes. We derive bounds on the throughput and the mean and maximum in-order delivery delay of AC-RLNC, both in the multipath and multipath multi-hop case. \jrev{In the multipath case, we prove that in the non-asymptotic regime, the suggested code may achieve more than 90\% of the channel capacity with zero error probability under mean and maximum in-order delay constraints, namely a mean delay smaller than three times the optimal genie-aided one and a maximum delay within eight times the optimum.} In the multipath multi-hop case, the balancing procedure is proven to be optimal with regards to the achieved rate. Through simulations, we demonstrate that the performance of our adaptive and causal approach, compared to selective repeat (SR)-ARQ protocol, is capable of gains up to a factor two in throughput and a factor of more than \jrev{three in mean delay and eight in maximum delay. The improvements on the throughput delay trade-off are also shown to be significant with regards to the previously developed singlepath AC-RLNC solution.}
\end{abstract}
\section{Introduction}\label{intro}
The increasing demand for network connectivity and high data rates \jrevm{necessitates} efficient utilization of all possible resources. In recent years, the connectivity moved forward from point-to-point schemes (i.e. single path, SP) to heterogeneous multipath (MP) multi-hop (MH) networks in which intermediate nodes can cooperate and share the common medium for efficient communications. \jrev{Such advanced networks appear for instance when considering the backhaul of smart cities, which can act as a bottleneck in case of a massive data traffic \cite{saadat2018multipath, noauthor_snob-5g_nodate}. Taking advantage of the connectivity to various resources such as public WiFi and wireless links with several base stations, MP-MH networks can} provide robustness and reliable communications with high data rates by simultaneously sharing the physical layer resources over each path and hop. However, it is critical to ensure that efficiency is not compromised and that in-order delivery delay is managed \cite{zeng2012joint}, particularly when the links are unreliable and retransmissions are required; when there are high round-trip-time (RTT) delays between the sender and the receiver; or when the state of each link/path is not fully determined. \Cref{MH_MP_setting} illustrates the considered \jrevm{communication settings}.

\ifdouble
\begin{figure}
    \centering
    \includegraphics[trim=0cm 0.7cm 0cm 0cm,width = 1 \columnwidth]{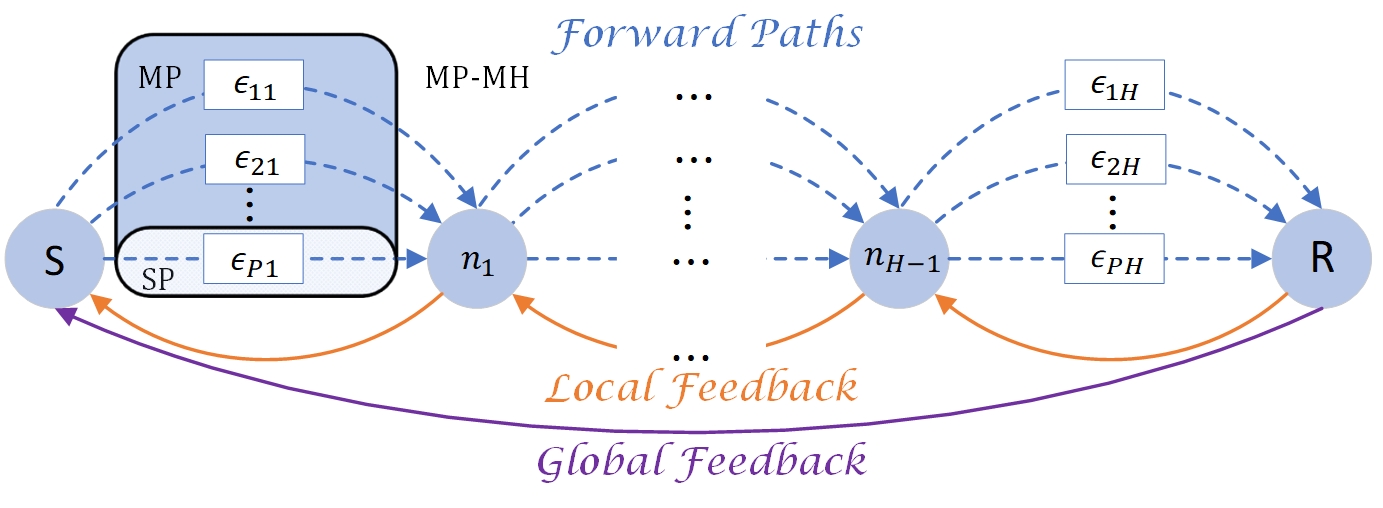}
    \caption{Multipath multi-hop network with delayed feedback. The single-path, multipath, and multipath multi-hop communications are denoted by SP, MP, and MP-MH, respectively.}
    \label{MH_MP_setting}
\end{figure}
\fi

To achieve the min-cut max-flow capacity in the \jrevm{MP-MH} networks, \rev{solutions based on information-theory} with a very large blocklength regime have been considered \cite{ho2006random}. However, in streaming communications, which are utilized, for example, in audio/video transmissions, automotive, smart-city, \jrevm{and control applications,} there are strict real-time constraints that demand low in-order delivery delays while the high data rates require all the available resources of the MP-MH networks. Traditional information-theory solutions with large blocklength are not able to reach this desired trade-off.

In point-to-point channels with erasures and delayed feedback, different packet-level techniques have been considered to manage and reduce the effects of this throughput-delay trade-off. Using forward error correction (FEC) according to the feedback acknowledgments, the in-order delivery delay can be reduced \cite{KarLei2014}, and the performance of SR-ARQ (Selective Repeat - Automatic Repeat reQuest) protocols can be boosted \cite{cloud2015coded}. Moreover, coding  solutions are considered as well in \ar{\cite{luby2002lt,shokrollahi2006raptor,joshi2012playback,joshi2014effect,malak2019tiny}} for SP communication, and in \cite{du2015network,cloud2016multi,gabriel2018multipath,ferlin2018mptcp,MSME19} for MP and MH networks. However, although part of those solutions are reactive to the feedback acknowledgments (i.e. causal), none of those solutions are tracking the varying channel condition and the rate (i.e. not adaptive). \jrev{In \cite{yang2014deadline}, an adaptive coding algorithm for blocks is presented, able to choose optimally the size of the next coding block for deadline-aware applications while in \cite{shi2015adaptive}, a multi-hop protocol optimizing the transmission of files is designed. Taking advantage of a per-packet acknowledgment, for SP communications, we have recently provided  in \cite{cohen2019adaptive} a novel adaptive and causal random linear network coding (AC-RLNC) solution with FEC. This solution is adaptive not only to the mean packet loss probability as the above works but also to the specific erasure pattern of each block. } To ensure low in-order delivery delays, in this solution the sender tracks the channel variations (i.e. adaptive) using a priori and posteriori algorithms (i.e. causal), and adapts the code rate according to the erasure realizations.

\ifdouble\else
\begin{figure}
    \centering
    \includegraphics[trim=0cm 0.7cm 0cm 0cm,width = 0.7 \columnwidth]{system3.jpg}
    \caption{Multipath multi-hop network with delayed feedback. The single-path, multipath, and multipath multi-hop communications are denoted by SP, MP, and MP-MH, respectively.}
    \label{MH_MP_setting}
\end{figure}
\fi

Although joint optimization of coding and scheduling in SP communications has been considered \jrev{in our previous work}  \cite{cohen2019adaptive}, joint coding and scheduling in MP and MH networks for optimizing the trade-off between throughput and in-order delivery delay  remains a challenging open problem. To generalize the AC-RLNC solution to MP networks, the main \rev{difficulty} lies in the discrete allocation of the new coded packets of information and FEC coded packets, over the available paths. \rev{On} one hand, in order to achieve a high throughput, the sender needs to send as many packets of information as possible. \rev{On} the other hand, it has to balance the failed transmissions by retransmitting several times the same packets. \rev{Furthermore, the} earlier the retransmissions are sent, \rev{the} lower is the delay. To obtain the desired throughput-delay trade-off, the raised open problem is how to allocate transmitted packets to the paths, at each time slot. Moreover, in a MH network, the rate is limited by the link with the smallest rate, therefore causing bottleneck effects if the channel rate varies from hop to hop. In this paper, we propose a MP and MH adaptive causal network coding solution that can learn the erasure pattern of the channels, and adaptively adjust the allocation of its retransmitted packets across the available paths. \rev{Furthermore}, in the MH setting, we propose a decentralized adaptive solution which may completely \rev{avoid} the bottleneck effect by reorganizing the order of global paths at each intermediate node. \jrev{The discrete allocation of packets on the paths as well as the decentralized solution able to decrease the bottleneck effect in MP and MH networks constitute the key improvements with regards to our previous SP solution \cite{cohen2019adaptive}.}  This novel approach closes the mean and max in-order-delay gap and boosts the throughput. 

\subsection*{Contributions}
We propose a novel adaptive causal coding solution with FEC for MP and MH communications with delayed feedback. The proposed solution generalizes the SP setting, in which the AC-RLNC algorithm can track the network condition, and adaptively adjust its retransmission rates a priori and \jrevm{a posteriori} over all the available resources in the network. Specifically, to balance the allocation of packets across different paths, the adaptive algorithm we utilize is based on \rev{a} discrete water-filling approach. Furthermore, in the MH setting, to reduce the bottleneck effect due to the channel variations at each hop, we propose a new decentralized optimization balancing algorithm. This adaptive algorithm can be applied independently at each intermediate node according to the feedback acknowledgments, while tracking the erasure patterns of the income and outcome channels.

\jrevm{For both MP and MH networks}, we provide bounds on the throughput and in-order delivery delay of our AC-RLNC solution. Specifically, we prove that for the MP network in the non-asymptotic regime, the suggested code may achieve more than 90\% of the channel capacity with zero error probability \jrev{under mean and maximum delay constraints}. \jrevm{The provided bounds on the in-order delivery delay provide further insights on the advantage of our MP protocol with regards to the SP one.} Moreover, in the MP-MH network, we prove that the decentralized balancing protocol is optimal with regards to the achieved rate.

\jrevm{We contrast the performance of the proposed approach with that of SR-ARQ.} We demonstrate that the proposed approach\off{ based on AC-RLNC for MP and MH networks} can, compared to SR-ARQ, achieve a throughput up to two times better and reduce the delay by a factor of more than three. \jrev{\jrevm{Furthermore, we compare our MP-MH solution} to the SP AC-RLNC protocol presented in \cite{cohen2019adaptive}, proving that the multipath and multihop specificities of our algorithm outperform the SP one, both in throughput and in in-order delay. \jrevm{While \cite{cohen2019adaptive} presented the improvements brought by the SP AC-RLNC protocol over traditional SP coding schemes, here we present a MP-MH solution that also provides better trade-offs than such traditional schemes.}}

The structure of this work is as follows. In \Cref{sys}, we formally describe the system model and the metrics \rev{in use}. In \Cref{rel}, we provide a background on the SP solution and on the MH transmission protocol. In \Cref{MP}, we present the MP solution with theoretical guarantees and  simulation results. In \Cref{MH}, \rev{we generalize the MP solution and analyses to the MP-MH solution}. Finally, we conclude the paper in \Cref{Conc}.

\ifdouble
\begin{table}
\normalsize
    \centering
    \begin{tabular}{|l|l|}
        \hline
        \textbf{Param} & \textbf{Definition}\\\hline
        $P$, $H$& number of paths and number of hops \\\hline
        $\epsilon_{p,h}$ &  erasure probability of the $p$th path and $h$th hop \\\hline
        $r_{p,h}$& $1-\epsilon_{p,h}$, rate of the $p$th path and $h$th hop \\\hline
        $\bar{\epsilon}$ & $\nicefrac{(\epsilon_1+...+\epsilon_P)}{P}$, mean erasure rate \\\hline
        t& time slot index\\\hline
        M & number of information packets\\\hline
        $x_i$ & information packets, $i\in [1,M]$\\\hline
        $\mu_i \in \mathbb{F}_z$ & random coefficients \\\hline
        $c_{t,p}$ & RLNC to transmit at time slot $t$ on path $p$\\\hline
        $t^{-}$ & $t-RTT$, time slot of the delayed feedback \\\hline
        $D_{mean}$& mean in-order delivery delay\\\hline
        $D_{max}$ & max in-order delivery delay\\\hline
        $\eta$ & throughput\\\hline
        $k$ & number of packets in window, $RTT-1$ \\\hline
        $EW$ & end window of $k$ new packets \\\hline
        $\bar{o}$ & maximum window size \\\hline
        $m_p$ & number of FEC's for the $pth$ path\\\hline
        $ad$ & number of added DoF's (global denoted by g)\\\hline
        $md$ & number of missing DoF's (global denoted by g)\\\hline
        $d$ & $md/ad$, DoF rate \\\hline
        $th$ & retransmission parameter \\\hline
        $\Delta$ & $P\cdot(d-1-th)$, DoF rate gap \\\hline
        $\mathcal{C}$ &  set of all sent RLNC's\\\hline
        $\mathcal{C}^{r}, \mathcal{C}^{n}$ & set of repeated and new RLNC's\\\hline
        $\mathcal{A}$& set of RLNC's with ACK feedback\\\hline
        $\mathcal{N}$& set of RLNC's with NACK feedback\\\hline
        $\mathcal{F}$&  set of RLNC's that do not have a feedback yet\\\hline
        $\mathcal{U}$& set of RLNC's depending on undecoded packets\\\hline
        $\mathcal{P}_p$& set of RLNC's sent on path $p$\\\hline
       \jrevm{$l(.,.)$} & \jrevm{Bhattacharya distance}\\\hline
       \jrevm{$r_p(t)$} & \jrevm{rate of the $p$-th path at time $t$}\\\hline
       \jrevm{$r_{p,up}$} & \jrevm{upper bound on the rate of the $p$-th path}\\\hline
       \jrevm{$n_{p}^{EW}$} & \jrevm{number of useless packets sent at the EW}\\
       & \jrevm{on the $p$-th path}\\\hline
       \jrevm{$n_{p}^{w}$} & \jrevm{number of packets sent on the $p$-th path}\\
       & \jrevm{ in one window}\\\hline
       \jrevm{$D_{\text{mean}}^s$} & \jrevm{$D_{\text{mean}}$ in case of feedback state $s$, i.e. either}\\
      &  \jrevm{no feedback, ACK feedback or NACK feedback} \\\hline 
       \jrevm{$\lambda$} &  \jrevm{fraction of time  without feedback}\\
       &\jrevm{compared to the total time of transmission}\\\hline
       \jrevm{$T_{\max}$} & \jrevm{maximum number of transmissions}\\\hline
    \end{tabular}
    \vspace{0.1cm}
    \caption{Symbol definition}
    \label{table : definition}
    \vspace{-0.9cm}
\end{table}
\fi

\section{System Model and Problem Formulation}\label{sys}
We consider a real-time slotted communication model with feedback in two different settings: first a multipath (MP) channel with $P$ paths and then a  multipath and multi-hop (MH) channel with $H$ hops and $P$ paths per hop, as represented in \Cref{MH_MP_setting}. \rev{Symbol definitions are provided in \ifdouble\Cref{table : definition}\else Table I\fi}.

\paragraph*{Multipath setting}
 At each time slot $t$, the sender transmits over each path $p\in \{1,\ldots,P\}$ a coded packet $c_{t,p}$.\footnote{A negligible size header containing transmission information may be sent with the coded packets.} \jrev{The considered coding scheme is part of RLNC. The sender is thus assumed to be able to generate random coefficients in a given Galois field in order to obtain random linear combinations of raw packets of information. \jrevm{A formal definition of such linear combinations is provided in section III with \eqref{eq:coding_process}.} For the decoding process, the receiver retrieves the raw packets of information \jrevm{by performing a Gaussian elimination on a linear system, built by considering the underlying equation for each coded packet.} In \cite{ho2006random}, this coding scheme is proven to generate asymptotically (i.e. for large field sizes) linearly independent coded combinations. Hence, it is assumed in the following that if $k$ coded packets depend only on $k$ raw packets of information, decoding is possible.
 } Forward paths between the sender and the receiver are modeled as independent binary erasure channels (BEC)\footnote{\jrev{Channel models including burst of erasures such as Gilbert-Elliott channels\cite{gilbert1960capacity} can also be considered. This has been done for our singlepath AC-RLNC protocol in \cite{cohen2019adaptive} and constitutes an interesting direction for future work.}}. Namely, erasure events are i.i.d. with probability $\epsilon_p$ for each $p$-th path. According to the erasure realizations,  the receiver sends at each time slot either an acknowledgment (ACK) or a negative-acknowledgment (NACK) message to the sender, using the same paths. For simplicity, feedback messages are assumed to be reliable, i.e., without errors in the feedback\footnote{Analysis that incorporates errors in the feedback \jrev{as well as a fluctuating RTT} can be conducted, e.g., by using the techniques provided in \cite{malak2019tiny}. However, we leave this interesting extension as future work.}.

The delay between the transmission of a coded packet and the reception of the corresponding feedback is called round trip time (RTT). Defining $\rho_p$ as the rate of the forward path $p$, in bits/second and $|c_{t,p}|$ as the size of coded packet $c_{t,p}$, in bits, the maximum duration of a transmission is $t_d = \rev{\underset{t,p}{\max}}\nicefrac{|c_{t,p}|}{\rho_p}$. \jrev{Letting $t_{\text{prop}}$ be the maximum propagation time of the paths between sender and receiver in seconds and assuming the size of feedback messages is negligible compared to the one of coded packets, the RTT is given by
\ifdouble
\begin{equation*}
\textstyle    RTT = t_d + 2t_{\text{prop}}.
\end{equation*}
\else
$RTT = t_d + 2t_{\text{prop}}$.
\fi
}Hence, for each transmitted coded packet $c_{t,p}$, the sender receives a ACK$(t,p)$ or NACK$(t,p)$ after RTT seconds. \jrev{As we consider slotted communications, the time-dependent quantities can also be defined in terms of number of slots. In the following, the RTT will hence have units of slots\footnote{\jrev{In WiFi communications with $12$Mbs speed, 8 slots correspond approximately to 1ms.}}.}
\paragraph*{Multipath Multi-hop setting} In this setting, the sender and receiver behave exactly as in the single-hop case. For simplicity, we assume that there are $P$ paths in each hop \rev{$h\in\{1,\ldots,H\}$}, each with i.i.d erasure probabilities $\epsilon_{p,h}$. At each time slot, each intermediate node $n_h$, \rev{$h\in\{1,\ldots,H-1\}$}, receives from the $h$-th hop (and therefore either from the sender for the first node or from the previous node for the others) $P$ coded packets from the independent paths. The node then sends $P$ (possibly different) coded packets on the $(h+1)$-th hop (towards next node or the receiver for the last intermediate node). For feedback acknowledgments, either a local hop-by-hop mechanism (from node to node) or a global one (directly from the receiver to the sender) can be used. \jrev{Letting $t_{\text{prop},h}$ be the maximum propagation delay of one hop, in seconds, and assuming all hops have the same propagation delay, the propagation time $t_{\text{prop}}$ becomes $t_{\text{prop}} = H t_{\text{prop},h}$.}

\ifdouble\else
\begin{figure}
    \centering
    \includegraphics[trim=0cm 0.7cm 0cm 0cm,width = 0.5 \columnwidth]{table2}
\end{figure}
\fi

Our goal for both of these settings, with parameters $\rho$-th and RTT, is to maximize the throughput, $\eta$, while minimizing the in-order delay, $D$, as given in the following definitions.
\begin{definition}{Throughput, $\eta$.}\label{dife:throu}
This is the total amount of information (in bits/second) delivered, in order\rev{,} at the receiver in $n$ transmissions over the forward channel. The normalized throughput is the total amount of information delivered, in order at the receiver divided by $n$ and the size of the packets.
\end{definition}

\begin{definition}{In-order delivery delay, $D$.}\label{dife:delay} This is the difference between the time an information packet is first transmitted in a coded packet by the sender and the time that the same information packet is decoded, in order at the receiver and successfully acknowledged \cite{zeng2012joint}. \jrev{The time needed for the transmission of the acknowledgment message is not counted in $D$}.
\end{definition}

 More precisely, we consider the mean and max in-order delay $D_{mean}$ and $D_{max}$. \rev{While} $D_{mean}$ reduces the overall delay, $D_{max}$ is critical for real-time applications that need a low inter-arrival time between packets.

\section{Background}\label{rel}
Before we consider the MP and MH protocols in details, we present the SP AC-RLNC algorithm suggested in \cite{cohen2019adaptive} and the capacity achieving protocol described in \cite{LunMedKoeEff2008}.

\ifdouble
\begin{figure}
\centering
\includegraphics[trim=0cm 0.8cm 0cm 0cm,width=1\columnwidth]{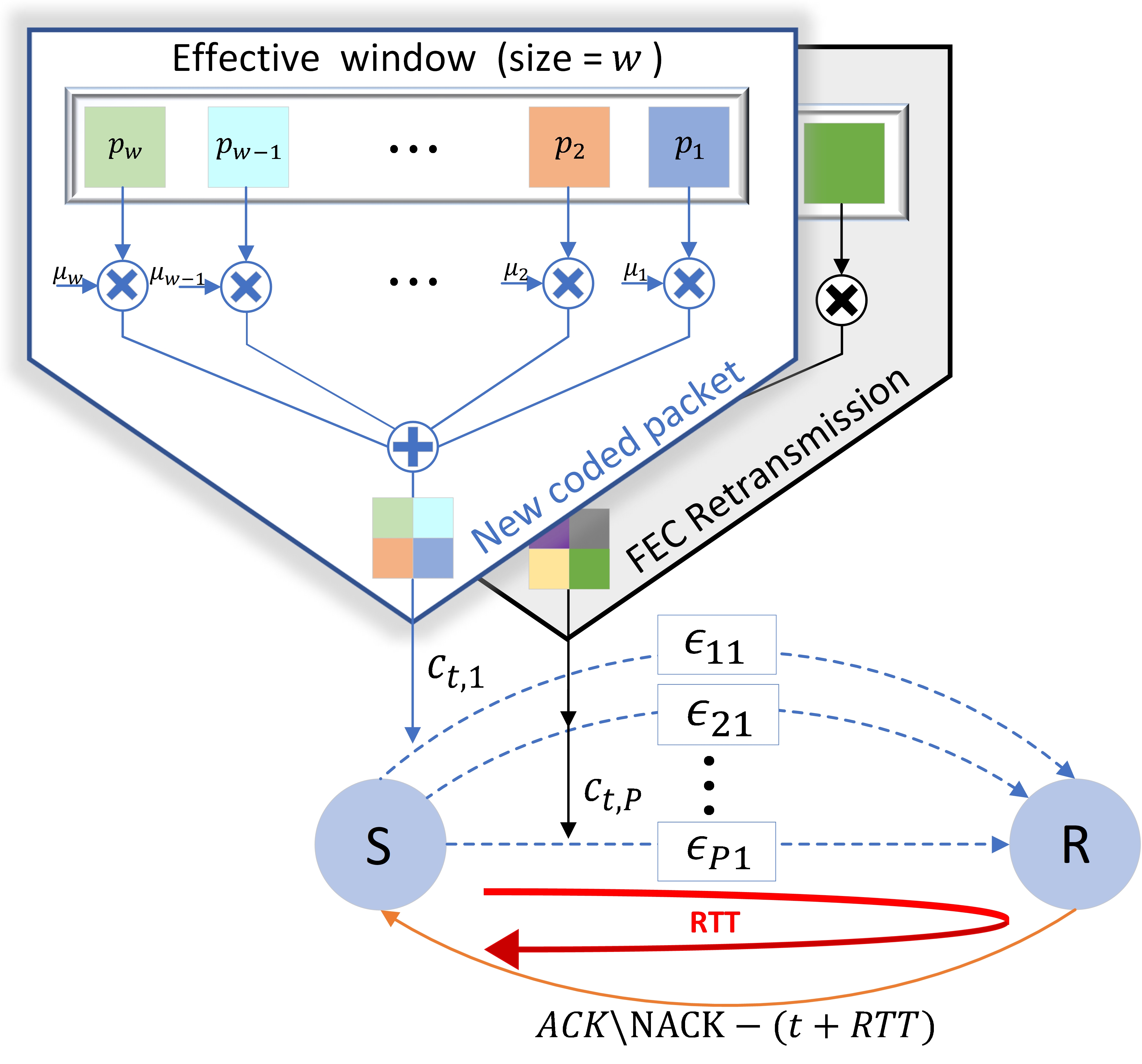}
     \caption{System model and encoding process of the coded RLNC combinations in MP network. In this example, for simplicity of notation, $w_{min}=1$.}
    \label{fig:encoding_pros}
\end{figure}
\else
\begin{figure}
\centering
\subfigure[]{\includegraphics[trim=0cm 0.8cm 0cm 0cm,width=0.4\columnwidth]{window_coding2.jpg}}
\subfigure[]{\includegraphics[trim=0cm 0.7cm 0cm 0cm,width=0.4\columnwidth]{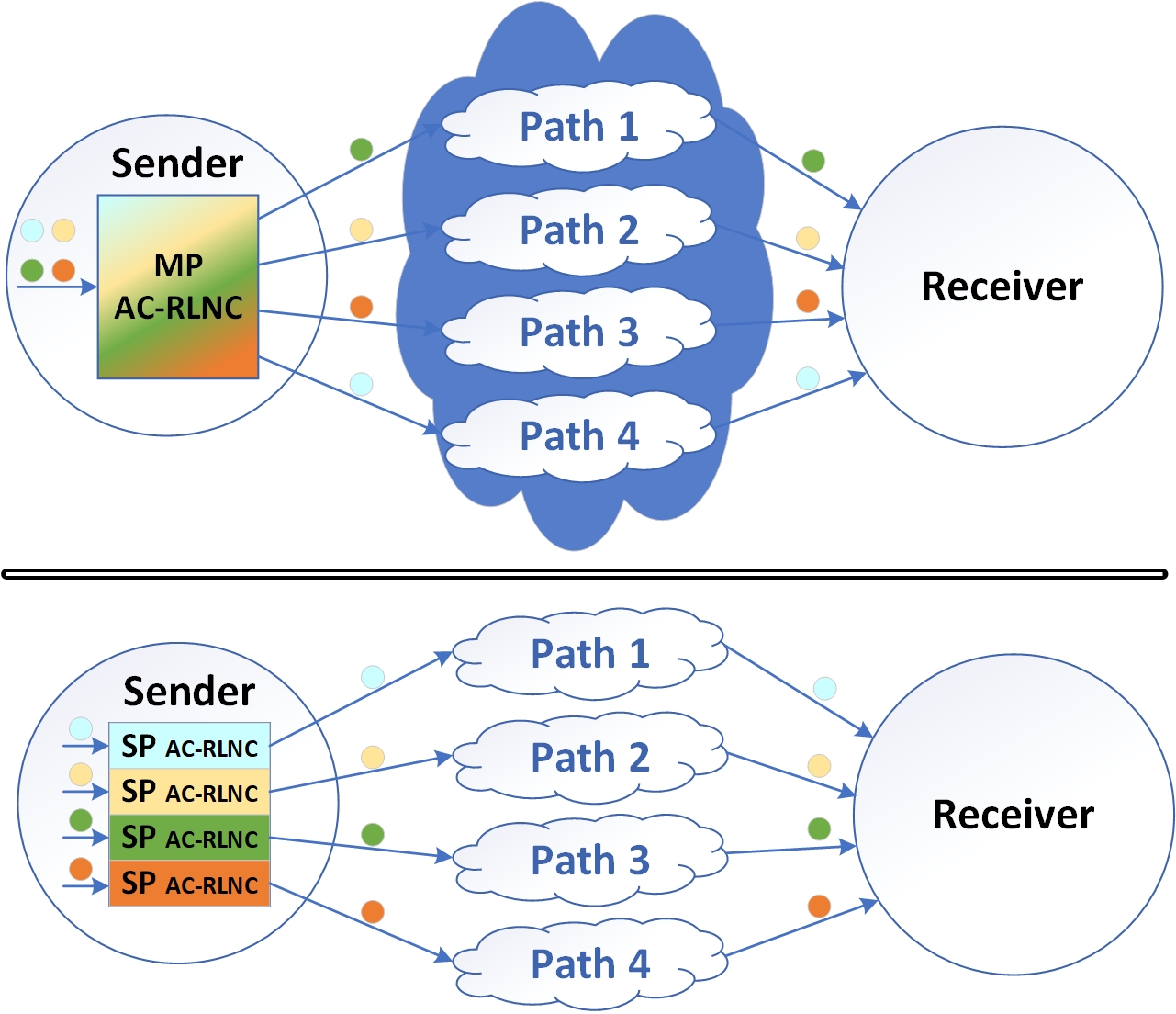}}
     \caption{(a) System model and encoding process of the coded RLNC combinations in MP network. In this example, for simplicity of notation, $w_{min}=1$. (b) \jr{Difference between MP AC-RLNC protocol (upper figure) and independent SP AC-RLNC protocols on each individual path (bottom figure)}. In this example, $H=1$ and $P=4$. By considering a joint MP protocol instead of independent SP ones, we reach better trade-offs between throughput and in-order delay.}
    \label{fig:encoding_pros}
\end{figure}
\fi

\paragraph{SP AC-RLNC algorithm} In \cite{cohen2019adaptive}, an adaptive causal random linear network coding  algorithm for a single-hop single-path setting is described. 
We review here several key features of that protocol.

\paragraph*{Coded packets and sliding window mechanism} Raw packets of information are encoded using RLNC. Each coded packet $c_t$, called a degree of freedom (DoF), is obtained as
\begin{eqnarray}
  \textstyle      c_t &=& \sum_{i=w_{min}}^{w_{max}}\mu_i \cdot x_i,
  \label{eq:coding_process}
\end{eqnarray}
with $\{\mu_i\}_{i=w_{min}}^{w_{max}}$, random coefficients in the field $\mathbb{F}_z$ of size z, $\{x_i\}_{i=w_{min}}^{w_{max}}$, the raw information packets and $w_{min}$ and $w_{max}$, the limits of the current window\footnote{   As shown in \cite{ho2006random}, if $z$ is large enough, a generation of $k$ raw information packets can be decoded with high probability using Gaussian elimination on $k$ coded packets.}. At each time step, the sender can either transmit a new coded packet or repeat the last sent combination\footnote{\rev{``}Same" and \rev{``}new" refer here to the raw packets of information contained in the linear combination. Sending the same linear combination thus means that the raw packets are the same but with different random coefficients.}. $w_{max}$ is thus incremented each time a new DoF is sent while $w_{min}$ corresponds to the oldest raw packet that is not yet decoded.  We denote by DoF($c_t$) the DoF's contained in $c_t$, i.e. the number of information packets in $c_t$.

\paragraph*{Tracking the path rate and DoF rate} Given the feedback acknowledgments, the sender can track the erasure probability $\epsilon$ (and thus the path rate $r$, defined as $r=1-\epsilon$) as well as the DoF rate $d = md/ad$, with $md$ and $ad$ being respectively the number of erased and repeated DoF's. These quantities are needed by the two FEC mechanisms that counteract the erasures.

\paragraph*{A priori mechanism} When $k = RTT-1$ new packets of information have been transmitted, $\lceil \epsilon \cdot k \rfloor$ repetitions\footnote{ $\lceil x \rfloor$ corresponds to rounding $x$ to the nearest integer.} of the same RLNC are sent in order to compensate for the expected erasures\footnote{In the  following, this mechanism is referred to as ``\textit{sending FEC's}''.}.
\paragraph*{A posteriori mechanism} When a priori repetitions are not sufficient, the retransmission  criterion $r-d\leq th$, with $th$ being a tunable parameter, determines when additional FEC's, called feedback FEC's (FB-FEC) are needed. Intuitively, when the DoF rate is higher (resp. lower) than the rate of the channel, \jrevm{then many (resp. few) coded packets are erased}\rev{,} and retransmissions are (resp. are not) needed.

\paragraph*{Size limit mechanism} In order to reduce $D_{max}$, the window size $w = w_{max}-w_{min}+1$ is limited to $\bar{o}$. Like $th$, $\bar{o}$ can also be tuned to achieve different trade-offs. When that limit is reached, the sender transmits the same RLNC until \rev{it} knows, \rev{as result of} the feedback \rev{acknowledgments}, that all information packets are decoded.

\paragraph{Multi-hop transmission}
As proved in \cite{mincutmaxflowtheorem}, the capacity of a network equals its maximum flow (or equivalently its minimum cut). Using random linear network coding, that capacity can be achieved for instance \rev{by} using the protocol suggested in \cite{LunMedKoeEff2008}. In this protocol, each node generates and sends random linear combinations of all packets in its memory whatever the structure of the network. Specifically to the MP and MH setting, at each time slot, each node (the sender included) stores received RLNC's (or the previously sent RLNC's for the encoder). Then, each node sends on the $P$ next paths a linear combination of all received RLNC's. In the asymptotic regime, that protocol achieves the capacity\cite{LunMedKoeEff2008}.

\section{Multipath Communication}\label{MP}
In this section, \jr{as illustrated in \ifdouble\Cref{fig:dif_sp_mp}}\else\Cref{fig:encoding_pros}-(b)\fi, we propose to merge the AC-RLNC solution described in \Cref{rel} (for \jr{each individual} SP), with an adaptive algorithm balancing the allocation of new RLNC's and FEC RLNC's on the paths. \ifdouble\Cref{fig:encoding_pros} \else\Cref{fig:encoding_pros}-(a) \fi shows the adaptive causal coding process on the single hop multipath network. The adaptive algorithm is based on a discrete water filling approach, i.e., bit-filling \rev{(BF)}, as given in \cite{papandreou2007bit}. However, the \rev{BF} is modified in order to take into account both rate and in-order delay objectives. To reach the desired trade-off between throughput and delay in the MP network, we suggest to utilize the key features of the SP AC-RLNC algorithm,  especially the a priori and a posteriori FEC mechanisms, as well as the tracking of the channel rate and the DoF rate via the feedback acknowledgments. Yet, in the MP network, to maximize the throughput while minimizing the in-order delay, allocation of new coded packets and retransmissions demands to consider adaptively the available resources across all the channels. The symbol definitions are provided in \ifdouble\Cref{table : definition}\else Table I\fi. The main components of the packet level protocol and the balanced allocation algorithm over the different paths are described next in \Cref{MPCodeConstruction}, \jrevm{while theoretical analyses to assess the achieved trade-offs are presented in \Cref{Analytical_results_MP}}.  In \Cref{MPsimulation}, the simulation results of the solution we suggest for the MP network are presented. \jr{The theoretical analyses and the provided simulation \jrevm{results provide insights} on the advantage of our MP protocol with regards to the SP ones, }applied independently on each path.
\subsection{Adaptive Coding Algorithm}\label{MPCodeConstruction}
Here we detail the MP solution, described in \ifdouble\Cref{pseudo_code}\else Algorithm 1\fi.
\paragraph{A priori mechanism (FEC)}
After the transmission of $k=P(RTT-1)$ new RLNC's, $m_p = \lceil\epsilon_p (RTT-1)\rfloor$ FEC's are sent on the $p$-th path. This mechanism tries to provide a sufficient number of DoF's to the receiver, by balancing the expected number of erasures. Note that $m_p$ may vary from path to path, according to the estimated erasure probability of each path.

\ifdouble
\begin{figure}
\centering
\includegraphics[trim=0cm 0.4cm 0cm 0cm,width=0.83\columnwidth]{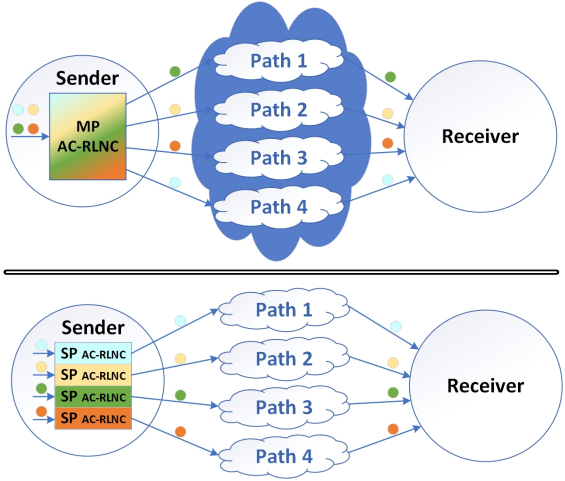}
	\caption{\jr{Difference between MP AC-RLNC protocol (upper figure) and independent SP AC-RLNC protocols on each individual path (bottom figure)}. In this example, $H=1$ and $P=4$. By considering a joint MP protocol instead of independent SP ones, we reach better trade-offs between throughput and in-order delay.}
	\label{fig:dif_sp_mp}
\end{figure}
\fi

\ifdouble\else
\begin{figure}
    \centering
    \includegraphics[trim=0cm 0.7cm 0cm 0cm,width = 0.5 \columnwidth]{algo1}
\end{figure}
\fi

\paragraph{A posteriori mechanism (FB-FEC)}
\ar{The retransmission criterion of the FB-FEC mechanism has to reflect, at the sender's best knowledge, the ability of the receiver to decode RLNC's.
Letting $md_g$ be the number of missing DoF's (i.e. the number of new coded packets that have been erased) and $ad_g$ be the number of added DoF's (i.e. the number of repeated RLNC's that have reached the receiver), the retransmission criterion can be expressed as $md_g > ad_g$.
Indeed, if the number of erased new packets is not balanced by enough repetitions, then decoding is not possible.
However, at the sender side, $md_g$ and $ad_g$ cannot be computed exactly due to the $RTT$ delay. At time step $t$, the sender can only compute accurately these quantities for the RLNC's sent before $t^{-} = t-RTT$ and that have thus a feedback acknowledgment. But for those sent between $t^{-}$ and $t$, these two quantities have to be estimated, for instance using the average rate of each path. Thus, letting $md_1$ and $ad_1$ (resp. $md_2$ and $ad_2$) correspond to the RLNC's with (resp. without) feedback acknowledgments, $md_g = md_1+md_2$ and $ad_g = ad_1+ad_2$ are computed through (\ref{mdg}) and (\ref{adg}). Note that in these equations, the different sets\footnote{Note that $\mathcal{C} = \mathcal{C}^{r}\cup\mathcal{C}^{n} = \mathcal{A}\cup\mathcal{N}\cup \mathcal{F}$} are defined in \ifdouble\Cref{table : definition} \else Table I \fi and that $|\mathcal{S}|$ denotes the cardinality of set $\mathcal{S}$.
\begin{eqnarray}
     &&\hspace{-1.2cm}\textstyle md_1 = |\mathcal{N}\cap\mathcal{C}^{n}\cap\mathcal{U}|, \hspace{-0.0cm}\textstyle md_2 = \sum_{p=1}^{P}\epsilon_p|\mathcal{P}_p\cap\mathcal{C}^{n}\cap\mathcal{F}\cap\mathcal{U}| \label{mdg}\\
     &&\textstyle\hspace{-1.2cm} ad_1 = |\mathcal{A}\cap\mathcal{C}^{r}\cap\mathcal{U}|, \hspace{0.3cm}\textstyle ad_2 =  \sum_{p=1}^{P}r_p|\mathcal{P}_p\cap\mathcal{C}^{r}\cap\mathcal{F}\cap\mathcal{U}|\label{adg}
\end{eqnarray}
Now defining the DoF rate of MP network as $d= md_g/ad_g$, and using a tunable parameter $th$,  the retransmission criterion can be rewritten as $d-1> th$.
   \off{\begin{equation*}
       \text{FB-FEC:} \quad retransmisison \quad  \iff d-1> th.
   \end{equation*}}}
Finally, defining the DoF rate gap $\Delta$ as
\begin{equation}\label{criterion1}
    \textstyle \Delta = P\cdot \left(d-1-th\right),
\end{equation}
the FB-FEC criterion we suggest is
\begin{equation}\label{criterion}
   \textstyle  \text{ FB-FEC:} \quad retransmission \iff \Delta>0.
\end{equation}
\off{Setting $th$ to 0 will result in a protocol that follows the average erasure of the channel, while increasing (resp. decreasing) it will decrease (resp. increase) the throughput and the delay since more (resp. less) FB-FEC's will be sent.}

\paragraph{Packet allocation}
In order for the sender to decide on which paths allocate the new coded packets of information and retransmissions of FB-FEC's, we propose a new algorithm inspired by a \rev{BF} algorithm given in \ar{\cite{papandreou2005new} and \cite{papandreou2007bit}}.  However, unlike the optimization problem considered in \ar{\cite{papandreou2005new} and \cite{papandreou2007bit}}, in which there is one  objective in order to optimally allocate bits under power constraints, the optimization problem in this paper contains two objectives. On one hand, the throughput needs to be maximized through the allocation of new coded packets while on the other hand, the in-order delay has to be reduced through FB-FEC retransmissions. \Cref{bit_filling} illustrates the \rev{BF} packets allocation for a multipath network.
\ifdouble
\else
\footnotetext[10]{\label{note1}The paths that have not yet an assigned RLNC for that time step.}
\setcounter{footnote}{10} sets the counter to have the specified value.
\fi
\ifdouble
\begin{figure}
\centering
\includegraphics[trim=0cm 0.7cm 0cm 0.4cm,width=1\columnwidth]{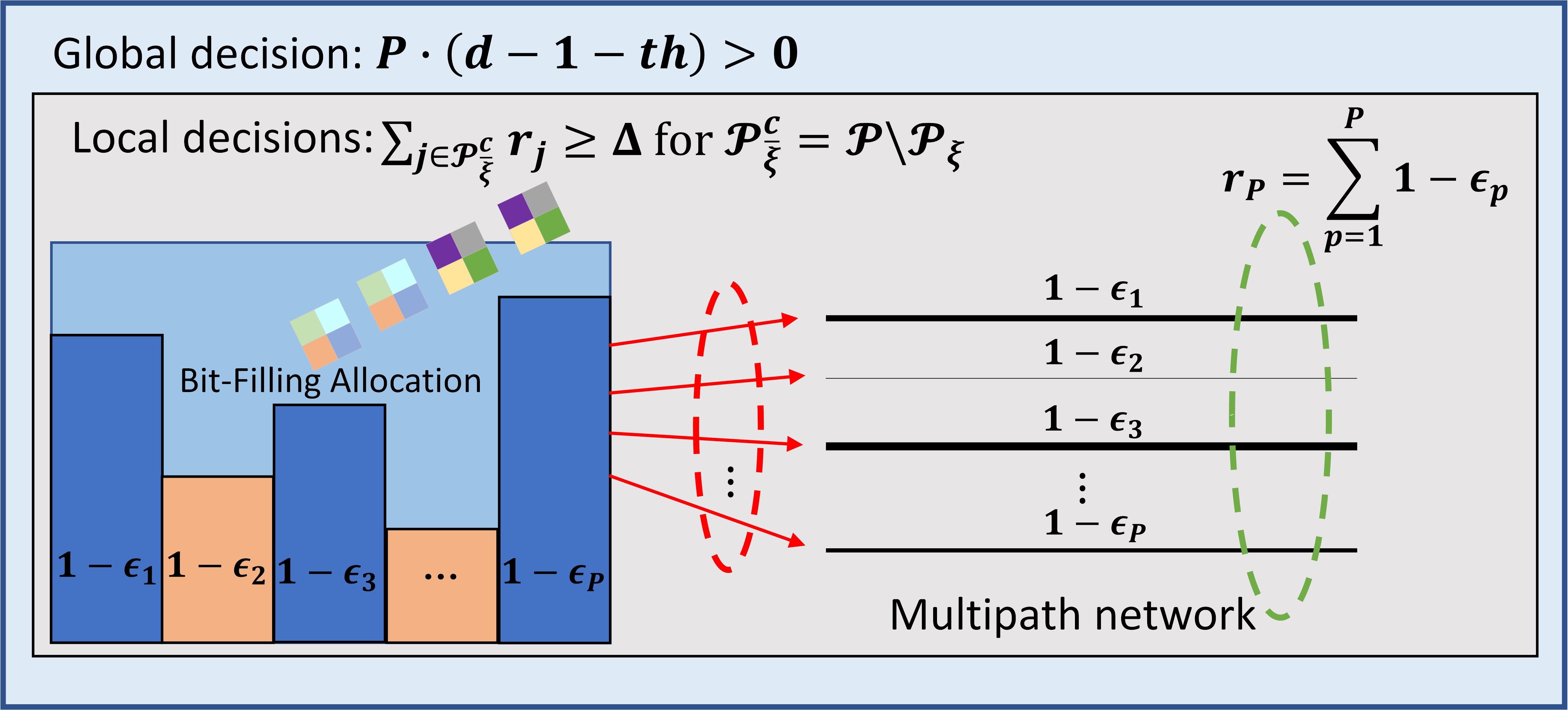}
     \caption{Bit-Filling packets allocation in multipath network. \rev{At the sender, given the estimated rates (green), according to the retransmission criterion given in \eqref{criterion}, first a global decision is made (i.e., check if retransmission is needed). Then if retransmission is needed, a local decision on which paths to send new packets of information and retransmissions (described by different colors as in \Cref{fig:encoding_pros}) are made  according to the bit-filling given in \Cref{PDWF} (red).}}
    \label{bit_filling}
\end{figure}
\else
\begin{figure}
\centering
\includegraphics[trim=0cm 0.7cm 0cm 0.4cm,width=0.65\columnwidth]{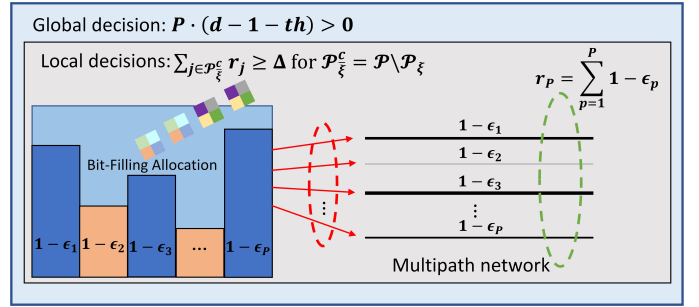}
     \caption{Bit-Filling packets allocation in multipath network. \rev{At the sender, given the estimated rates (green), according to the retransmission criterion given in \eqref{criterion}, first a global decision is made (i.e., check if retransmission is needed). Then if retransmission is needed, a local decision on which paths to send new packets of information and retransmissions (described by different colors as in \ifdouble\Cref{fig:encoding_pros}\else \Cref{fig:encoding_pros}-(a)\fi) are made  according to the bit-filling given in \Cref{PDWF} (red).}}
    \label{bit_filling}
\end{figure}
\fi

We define the set of all the paths as $\mathcal{P}$, and the index of all the possible sub-sets as $\xi\in\{1,\ldots,2^{P}\}$.  We denote the possible subset of paths over which the sender will transmit the new coded packets of information as $\mathcal{P}_\xi$, and the possible subsets of paths over which the sender will transmit the retransmissions of FB-FEC packets as $\mathcal{P}_{\bar{\xi}}^c$.
\begin{prop}[Bit-Filling]\label{PDWF}
Given the estimated rates of all the paths, $r_p$ with $p\in\{1,\ldots,P\}$, the sender wants to maximize the throughput of the new packets of information. \jrev{The set of paths $\hat{\mathcal{P}}^{\xi}$ on which new packets are sent is obtained as}
\ifdouble
\begin{equation}
\label{DWF}
\begin{aligned}
&  \jrev{\hat{\mathcal{P}}^{\xi}} =\rev{\underset{(\mathcal{P}_\xi)}{\arg\max}}
& &  \sum_{i\in \mathcal{P}_\xi} r_i,   \\			
& \hspace{0.4cm}\text{s.t.}
& &\sum_{j\in \mathcal{P}_{\bar{\xi}}^c} r_j \geq \Delta \quad  \text{for} \quad \mathcal{P}_{\bar{\xi}}^c = \mathcal{P}\setminus\mathcal{P}_\xi,
\end{aligned}
\end{equation}
\else
\begin{equation}
\label{DWF}
\jrev{\hat{\mathcal{P}}^{\xi}} = \rev{\underset{(\mathcal{P}_\xi)}{\arg\max}} \sum_{i\in \mathcal{P}_\xi} r_i, \quad \text{s.t.} \quad \sum_{j\in \mathcal{P}_{\bar{\xi}}^c} r_j \geq \Delta \quad  \text{for} \quad \mathcal{P}_{\bar{\xi}}^c = \mathcal{P}\setminus\mathcal{P}_\xi,
\end{equation}
\fi
where the optimization problem minimizes the in-order delivery delay, by providing over the selected paths a sufficient number of DoF's for decoding. \off{\textcolor{red}{Alejandro : the following sentence is a bit redundant with the previous one no? I don't think we should keep both.}In other words, the sender needs to send sufficient retransmissions of FB-FEC packets over the selected paths so that the receiver is able to decode.}
\end{prop}

\ifdouble
\begin{algorithm}
\caption{multipath protocol for packet scheduling}
\label{pseudo_code}
\begin{algorithmic}[1]
\State \textbf{Initialize transmission:}
\While{packets to transmit}
    \If{Feedback available}
        \State Update $\epsilon_p$ for each path
        \State Update $md_g$ and $ad_g$
        \State Update $\Delta$
    \EndIf
    \State \textbf{Size limit transmissions:}
    \If{$w >\bar{o}$}
        \State Retransmit same RLNC until DoF($c_t$)$=0$
    \Else
    \State \textbf{FEC transmissions:}
        \ForAll{paths with $m_p>0$}
            \State Retransmit same RLNC
            \State $m_p = m_p-1$
        \EndFor
        \If{remaining paths\footnote{}}
            \State \textbf{FB-FEC transmissions:}
            \If{$\Delta>0$}
                \State Determine FB-FEC paths
                \State Transmit same RLNC on these paths
            \EndIf
            \State \textbf{New transmissions:}
            \ForAll{remaining paths}
                \If{not EW}
                    \State Transmit new RLNC
                \EndIf
            \EndFor
            \State \textbf{FEC transmission (initialization):}
            \If{EW}
                \State Set $m_p := \lceil \epsilon_p (RTT-1)\rfloor$
                \ForAll{remaining paths}
                    \State Transmit same RLNC
                    \State $m_p = m_p-1$
                \EndFor
            \EndIf
        \EndIf
    \EndIf
\EndWhile
\end{algorithmic}
\end{algorithm}
\fi
\ifdouble
\footnotetext[10]{\label{note1}The paths that have not yet an assigned RLNC for that time step.}
\fi

It is important to note that by \rev{tuning} the chosen parameter $th$ (reflected in $\Delta$ \rev{\eqref{criterion1}}) it is possible to obtain the desired throughput-delay trade-off.
Moreover, to maximize the performance of the proposed approach, it is required to solve problem (\ref{DWF}) only when the estimations of the rates change.
To reduce the complexity of the optimization problem, once the number of paths is high, we can consider a relaxation of the optimization, \rev{e.g.,} using Knapsack problem algorithms \ar{\cite{Kap_relax,Kap_relax2}}.

\subsection{Analytical Results for Delay and Throughput}\label{Analytical_results_MP}

\subsubsection{An Upper Bound for the Throughput}\label{SecThroughputMP}
Here,  the achieved throughput in the MP network is upper bounded with zero error probability, by generalizing the techniques given in \cite{cohen2019adaptive} for point-to-point communications. \jrev{Considering zero error probability throughput means that all packets need to be decoded in a given delay budget, both for the mean and maximum delay.} In the suggested AC-RLNC solution, the sender follows the retransmission criterion \eqref{criterion}, which is computed according to the acknowledgments provided by the feedback channel. Yet, due to the transmission delay, those acknowledgments are obtained at the sender with a delay of RTT. Thus, the estimated rates of the paths may be different from the actual ones.

We consider the case for which the actual sum-rate of the paths at time slot t, i.e. $\sum_{p=1}^{P} r_p(t)$, is higher than the estimated rate at the sender side, i.e. $\sum_{p=1}^{P} r_p(t^{-})$ with $t^-=t-RTT$. In this case, throughput will be spoilt as coded retransmissions will be sent while not being necessary for the decoding. Let us denote by $\textbf{c}_{p}=(c_{t^-,p},\ldots,c_{t,p})$ and by $\textbf{c}_{p}^{\prime}=(c_{t^-,p}^{\prime},\ldots,c_{t,p}^{\prime})$, the vectors of coded packets transmitted on the $p$-th path according to the retransmission criterion, given respectively the estimated rate $r_p(t^-)$ and the actual rate $r_p(t)$ available at the sender non-causally.

Various methods have been considered in order to bound the channel variations in the non-asymptotic regime \cite{polyanskiy2010channel,polyanskiy2011dispersion}. Here, to upper bound the throughput, we bound the distance between the realization during one RTT period given the actual rate and calculated rate at each path utilizing the following minimum Bhattacharyya distance \cite{shannon1967lower, viterbi2013principles,dalai2014elias,barg2005distance}.
\begin{definition}\label{BhattacharyyaD}
Given a probability density function $W(y)$ defined on a domain $\mathcal{Y}$, the Bhattacharyya distance between two sequences $c_{p}$ and $c_{p}'$ is given by \cite{dalai2014elias}
\begin{equation*}
    l(c_{p},c_{p}^{\prime})=-ln(BC(c_{p},c_{p}')),
\end{equation*}
where $BC(c_{p},c_{p}')$ is the Bhattacharyya coefficient, defined as
\begin{align}\label{eq:BC}
  BC(c_{p},c_{p}') = \sum_{y\in \mathcal{Y}}\sqrt{W(y|c_{p})W(y|c_{p}^{\prime})},
\end{align}
with $W(y|c_{p})$ and $W(y|c_{p}^{\prime})$ corresponding to $W(y)$ conditioned on the sequences $c_{p}$ and $c_{p}^{\prime}$, respectively.
\end{definition}

\begin{theorem}{\jrev{ An upper bound on the throughput of AC-RLNC in MP network is}}\label{theoremMP}
\begin{align}
    \eta\leq \sum_{p=1}^{P} \jrevm{\left (r_p(t^-) - l(r_p(t),r_p(t^-))\right )},
    \label{eq:bat_them}
\end{align}
where $l(\cdot,\cdot)$ is the Bhattacharyya distance \jrevm{and the rate $r_p(t)$ of the $p$-th path at time slot $t$ is bounded by \eqref{eq:max_rate}.}
\end{theorem}
\begin{proof}
Bounding the rate of each individual path with \cite[Theorem 1]{cohen2019adaptive}, the throughput $\eta$ is upper bounded by the sum of these bounds. For completeness, the main steps of the single path proof are briefly reviewed here.
Given the calculated rate $r_p(t^-)$, the actual rate of each path in the MP network at time slot $t$ is bounded by
\begin{equation}\label{eq:max_rate}
     r_p(t) \leq r_p(t^-)+ \frac{\sqrt{V_p(t)}}{RTT-1+m_p},
\end{equation}
where $V_p(t)$ denote the variance of each path during the period of RTT.

To conclude, using the summation range in (\ref{eq:BC}) to be from $t=0$ to $RTT-1$, and letting $W(y\vert c_{p}')=r_p(t^-)$ and $W(y\vert c_{p})=r_p(t)$, we obtain
\ifdouble
\begin{align*}
    \eta\leq \sum_{p=1}^{P} \jrevm{\left(r_p(t^-) - l(r_p(t),r_p(t^-))\right)}.
\end{align*}
\else
    $\eta\leq \sum_{p=1}^{P} \jrevm{\left(r_p(t^-) - l(r_p(t),r_p(t^-))\right)}$.
\fi
\jrevm{
Note that by using the upper bound on the rate $r_p(t)$ given by \eqref{eq:max_rate} in the Bhattacharyya distance, the distance $l(r_p(t),r_p(t^-))$ is bounded. Thus enables to obtain the result in \Cref{theoremMP}.}
\end{proof}

\ifdouble
\begin{figure}
\centering
\includegraphics[trim=0cm 0.7cm 0cm 0cm,width=1\columnwidth]{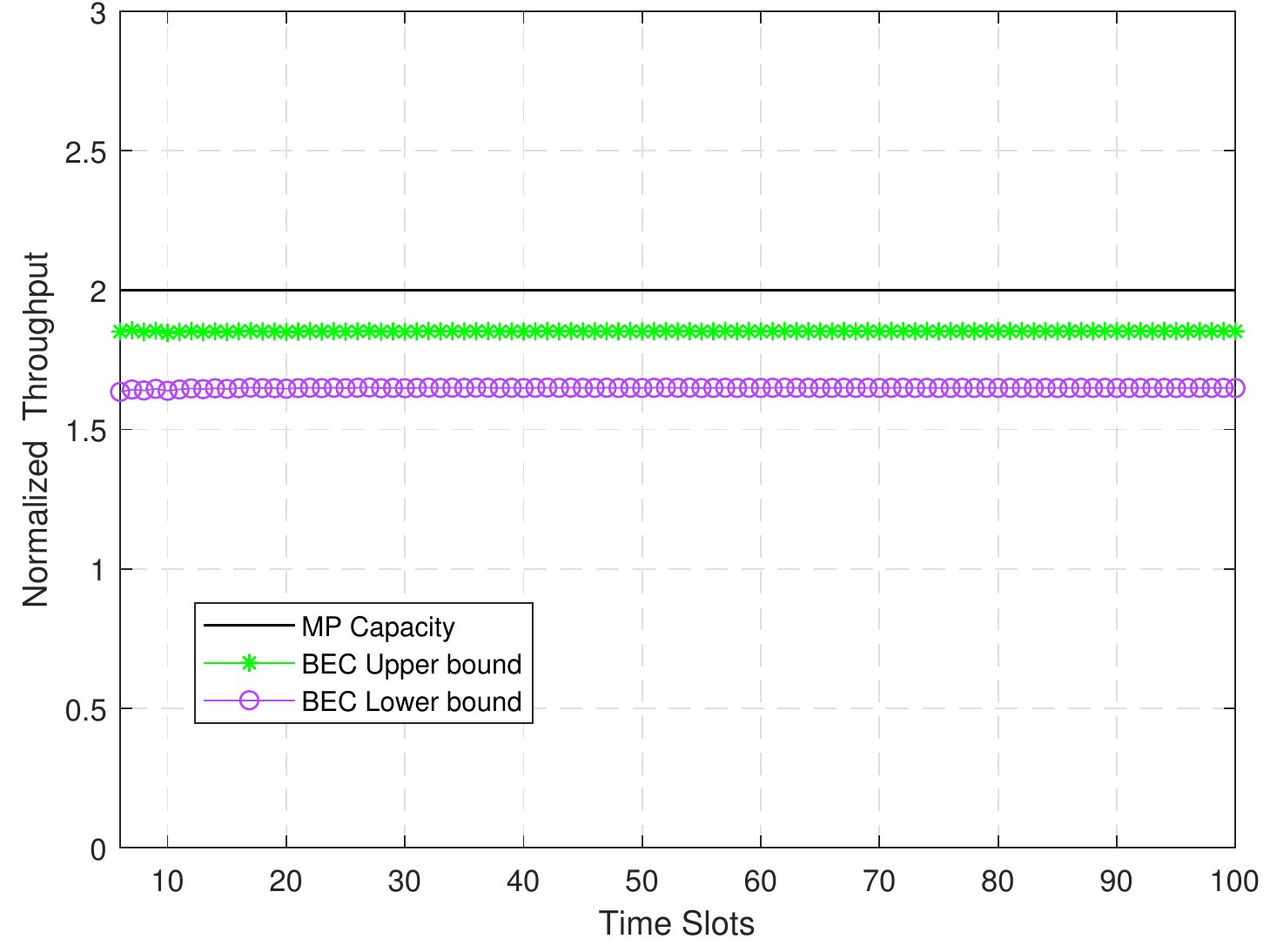}
	\caption{\jrev{Throughput upper and lower bounds in MP network with $H=1$ and $P=4$ for BEC channels with erasure probability of $\epsilon_{11}=0.2$, $\epsilon_{21}=0.4$, $\epsilon_{31}=0.6$ and $\epsilon_{41}=0.8$. Note that the range of the abscissa is from $2$ (the theoretical minimum of the $RTT$ delay) to $100$. Moreover, the throughput is not degraded by increasing the $RTT$. In the asymptotic regime, the MP AC-RLNC code may attain the capacity.}}
	\label{fig:ThroughputErrorBecMP}
\end{figure}

\begin{figure}
\centering
\includegraphics[trim=0cm 0.7cm 0cm 0cm,width=1\columnwidth]{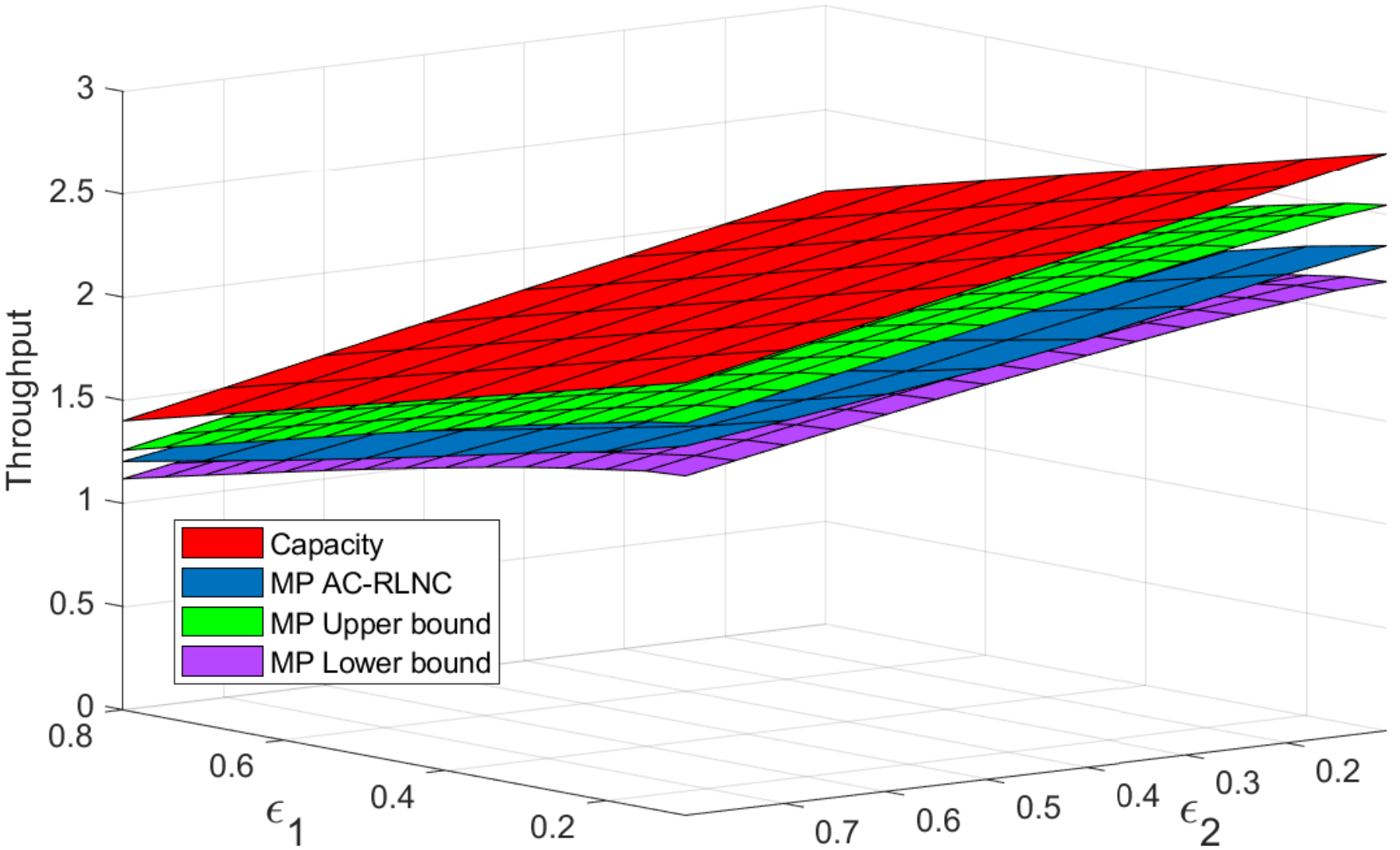}
	\caption{\jrev{Throughput upper and lower bounds for BECs in MP network with $H=1$, $P=4$, $RTT=20$, and with $\epsilon_{31}= 0.2$ and $\epsilon_{41} =0.8$, while the erasure probabilities of the two other paths ($\epsilon_{11}$ and $\epsilon_{21}$) vary in the range $[0.1 \;  0.8]$.}}
	\label{fig:ThroughputMP}
\end{figure}
\else
\begin{figure}
\centering
\subfigure[]{\includegraphics[trim=0cm 0.7cm 0cm 0cm,width=0.4\columnwidth]{BhattacharyyaDistance_lower}}
\subfigure[]{\includegraphics[trim=0cm 0.7cm 0cm 0cm,width=0.5\columnwidth]{ThroughputUpper_lower2}}
\subfigure[]{\includegraphics[trim=0cm 0.8cm 0cm 0cm,width=0.5\columnwidth]{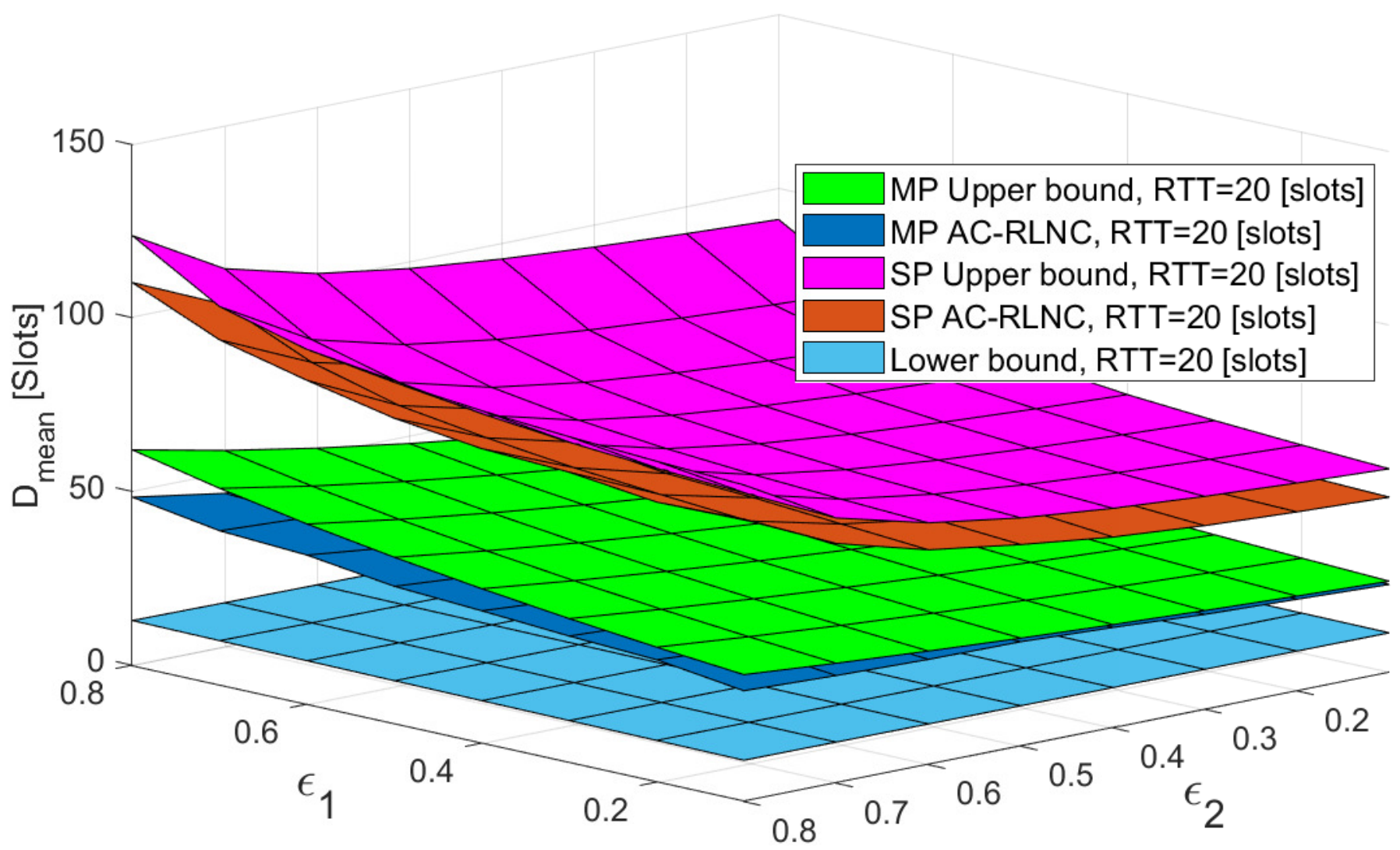}}
	\caption{(a) \jrev{Throughput upper and lower bounds in MP network with $H=1$ and $P=4$ for BEC channels with erasure probability of $\epsilon_{11}=0.2$, $\epsilon_{21}=0.4$, $\epsilon_{31}=0.6$ and $\epsilon_{41}=0.8$. Note that the range of the abscissa is from $2$ (the theoretical minimum of the $RTT$ delay) to $100$. Moreover, the throughput is not degraded by increasing the $RTT$. In the asymptotic regime, the MP AC-RLNC code may attain the capacity.} (b) \jrev{Throughput upper and lower bounds for BECs in MP network with $H=1$, $P=4$, $RTT=20$, and with $\epsilon_{31}= 0.2$ and $\epsilon_{41} =0.8$, while the erasure probabilities of the two other paths ($\epsilon_{11}$ and $\epsilon_{21}$) vary in the range $[0.1 \;  0.8]$.} (c) \jrev{Mean in-order delivery delay upper and lower bounds for a MP network with $H=1$ and $P=4$ for BEC's with erasure probability of $\epsilon_{11}=0.2$, $\epsilon_{21}=0.4$, $\epsilon_{31}=0.6$ and $\epsilon_{41}=0.8$, considering the new MP AC-RLNC protocol and independent SP AC-RLNC protocols of \cite{cohen2019adaptive}.}}
	\label{fig:ThroughputErrorBecMP}
\end{figure}
\fi

\begin{corollary}{\bf BEC.}\label{BEC_MP}
An upper bound on the throughput of AC-RLNC for BEC in the MP network is
\begin{eqnarray*}
    \eta^{BEC}\leq \sum_{p=1}^{P} \jrevm{\left(1-\epsilon_p - l(r_{p}^{BEC}(t),r_{p}^{BEC}(t^-))\right)}.
\end{eqnarray*}
\end{corollary}

\begin{proof}
In the same manner, using Bhattacharyya distance, the proof is obtained directly from \cite[Corollary 1]{cohen2019adaptive}. That is, we consider the sum upper bound throughput of all the available path in the BEC independently. Thus, in the BEC channel $r_p(t^-)=1-\epsilon_p$, and
\[
    r_{p}^{BEC}(t)\leq r_{p}^{BEC}(t^-) + \frac{\sqrt{V_{p}^{BEC}(t)}}{RTT-1+m_p},
\]
where $V_{p}^{BEC}(t)$ is the variance of the BEC during the period of RTT given by
\[
V_{p}^{BEC}(t) = RTT(1-r_{p}^{BEC}(t^-))r_{p}^{BEC}(t^-).
\]

Hence, using the summation range in (\ref{eq:BC}) to be from $t=0$ to $RTT-1$, and letting $W(y\vert c')=r_{p}^{BEC}(t^-)$ and $W(y\vert c)=r_{p}^{BEC}(t)$,
\begin{eqnarray*}
    \eta^{BEC}\leq \sum_{p=1}^{P} \jrevm{\left(1-\epsilon_p - l(r_{p}^{BEC}(t),r_{p}^{BEC}(t^-))\right)}.
\end{eqnarray*}
\end{proof}

The upper bound of \Cref{BEC_MP} can be analyzed in two ways. First, the impact of the $RTT$ on the upper bound is analyzed in \ifdouble\Cref{fig:ThroughputErrorBecMP}\else\Cref{fig:ThroughputErrorBecMP}-(a)\fi. Secondly, the upper bound and the throughput achieved by the suggested protocol are compared in \ifdouble\Cref{fig:ThroughputMP}\else\Cref{fig:ThroughputErrorBecMP}-(b)\fi.

In \ifdouble\Cref{fig:ThroughputErrorBecMP}\else\Cref{fig:ThroughputErrorBecMP}-(a)\fi, the throughput upper bound of AC-RLNC is shown as a function of $RTT$, for MP networks with 4 BEC's, respectively with rate $\epsilon_{11}=0.2$, $\epsilon_{21}=0.4$, $\epsilon_{31}=0.6$ and $\epsilon_{41}=0.8$. We can note that in this specific network, any AC-RLNC solution can at most achieve around \jrevm{$90\%$} of the capacity of the network. In \ifdouble\Cref{fig:ThroughputMP}\else\Cref{fig:ThroughputErrorBecMP}-(b)\fi, we present the MP throughput upper bound for BEC in a MP network with $H=1$, $P=4$, $RTT=20$, and with $\epsilon_{31}= 0.2$ and $\epsilon_{41} =0.8$, while the erasure probabilities of the two other paths ($\epsilon_{11}$ and $\epsilon_{21}$) vary in the range $[0.1 \;  0.8]$. With the chosen algorithm's parameters (namely, $th=0$ and $\bar{o} = 2k$), the achieved throughput is close to the upper bound. By changing these parameters, that upper bound may be obtained but to the detriment of the in-order delay. Hence, the throughput-delay trade-offs can be managed to meet specific application's constraints by tuning these two parameters. \jrev{Note that as  the bounds are obtained path by path, they also hold for the parallel AC-RLNC singlepath protocols whose performances are thus not represented in the figure.}

\subsubsection{\jrev{A Lower Bound for the Throughput}}\label{SecThroughputLowMP}
\off{To obtain an upper bound on the throughput, the rate loss caused by the additional FEC packets sent due to the delayed feedback is analyzed. Yet, with regards to this upper bound, it remains to determine the impact of the size limit mechanism, which will degrade the algorithm's performances. Analyzing this impact will therefore provide a lower bound on the throughput.}
\jrevm{A lower bound on the throughput can be obtained by determining the impact of the size limit mechanism and its degradation to the algorithm’s performance.}
\jrev{ If the window size limit is reached, the sender transmits FEC transmissions until he knows from the acknowledgments that the receiver managed to decode all the information packets in the window. However, the delayed feedback induces the transmission of useless repetitions as the decoder has already been able to decode. Letting $n^{\text{EW}}_{p}$ and $n^w_p$ denote respectively the number of such useless transmissions and the total number of transmissions during the window, per path, a lower bound on the throughput is obtained in the following by bounding above the performance loss with regards to the upper bound provided by the right-hand side of \eqref{eq:max_rate} and denoted by $r_{p,up}$.}

\jrev{\begin{theorem}{A lower bound on the throughput of AC-RLNC in MP network is}\label{theoremMPlb}
\begin{align}
    \eta_{lb} \geq \sum_{p=1}^{P} \jrevm{\left(r_{p,up} - \frac{ n^{\text{EW}}_{p}}{n^{w}_{p}}\right)}.
\end{align}
\end{theorem}}

\begin{proof}
\jrev{First, an expression for $n^{\text{EW}}_{p}$ can be obtained by considering the sender has reached the window size limit. In this case, he sends the same packets until he receives an acknowledgment enabling it to know the receiver has managed to decode all packets. Suppose the packet allowing to decode was sent at time $t$. Then, transmissions between $t$ and $t+RTT$ do not bring any information to the receiver. Nevertheless, for erased packets in this interval, one does not suffer for throughput loss as the corresponding time slot is anyway useless. Hence, the number of useless packets is obtained as
\begin{equation*}
    \jrevm{n^{\text{EW}}_{p} \leq  \left(1-erf\left(\frac{1}{\sqrt{2}} \right)\right) (1-\epsilon_p) \text{RTT} + O\left(\frac{1}{\sqrt{t}}\right),}
\end{equation*}
\jrevm{where for $t$ sufficient large, $O\left(\frac{1}{\sqrt{t}}\right)$ approaches zero.} In this equation, the $(1-\epsilon_p)$ factor enables to only consider received packets while the $ \left(1-erf\left(\frac{1}{\sqrt{2}} \right)\right)$ factor takes into account the fact that the end of window is reached only if the erasure probability significantly deviates from the mean, as otherwise the retransmission mechanisms are sufficient to decode. As a first approximation, we consider that such events occur when the deviation is greater that one standard deviation, using the so called $68-95-99.7$ rule. \jrevm{Yet, such rule holds for sufficient large $t$. Applying the Berry-Essen theorem \cite[Theorem 13] {polyanskiy2010channel_thesis} to the independent erasure events, the error term of the asymptotic rule is proven to decrease as $\frac{1}{\sqrt{t}}$.}\\
~~\\This number needs to be compared to the total window size, which can be written as
\begin{equation}\label{eq:nw}
    \jrevm{n^{w}_{p} \geq \left(k +   k \epsilon_p + k \epsilon_p^2\right)  \cdot f  + 1},
\end{equation}
\jrevm{where $f =\bar{o}/k$ denotes the window size factor translating the ratio between the maximum window size $\bar{o}$ and the number of information packets per generation $k$. In the above equation, for each window, we consider the transmission of $k$ packets, the a priori mechanism and the a posteriori mechanism. Furthermore, we consider the window size limit is reached if these a priori and a posteriori retransmission mechanisms fail to provide enough DoFs, i.e. if, compared to expected number of erasures, one additional erasure occurs during the first window. }
}
\end{proof}
\jrev{In \ifdouble\Cref{fig:ThroughputErrorBecMP}\else\Cref{fig:ThroughputErrorBecMP}-(a)\fi, the upper bound and lower bound are compared to the capacity of a multipath single-hop network with $4$ paths, for RTTs ranging from $2$ to $100$ time slots. From this figure, one can observe that increasing the RTT does not decrease the protocols performances. In \ifdouble\Cref{fig:ThroughputMP}\else\Cref{fig:ThroughputErrorBecMP}-(b)\fi, the same quantities are compared with regards to numerical results for a $RTT=20$ time slots. In this figure, one can observe that the obtained bounds are relatively tight and close to the capacity. Finally, in \Cref{fig:ThroughpuFactorMP}, we compare the obtained bounds for different window sizes through the parameter $f$. One the one hand, the ratio between the lower and upper bounds, computed as $F_{\eta}^{\text{AC-RLNC}} = 100 \cdot \frac{\eta_{lb}}{\eta_{ub}}$ is observed to tend towards $100\%$ when $f$ increases, showing thus that the bounds get tighter for large windows. On the other hand, the quantity $F_{\text{capacity}}^{\text{AC-RLNC}} = 100 \cdot \frac{\eta_{lb}}{\sum_{i=1}^{P} 1-\epsilon_p}$  compares the lower bound to the capacity, showing thus that the protocol may attain the capacity for large window size\jrevm{\footnote{\label{foot:numerical_evaluation}\jrevm{Both observations come from the numerical evaluation of the theoretical factors.}}}. To provide a comparison, SR-ARQ is at $45\%$ of the capacity for this network in our numerical validation (as one can observe from \Cref{MP_perf} and as represented in \Cref{fig:ThroughpuFactorMP}).}

\begin{figure}
\centering
\ifdouble
\includegraphics[trim=0cm 0.7cm 0cm 0cm,width=1\columnwidth]{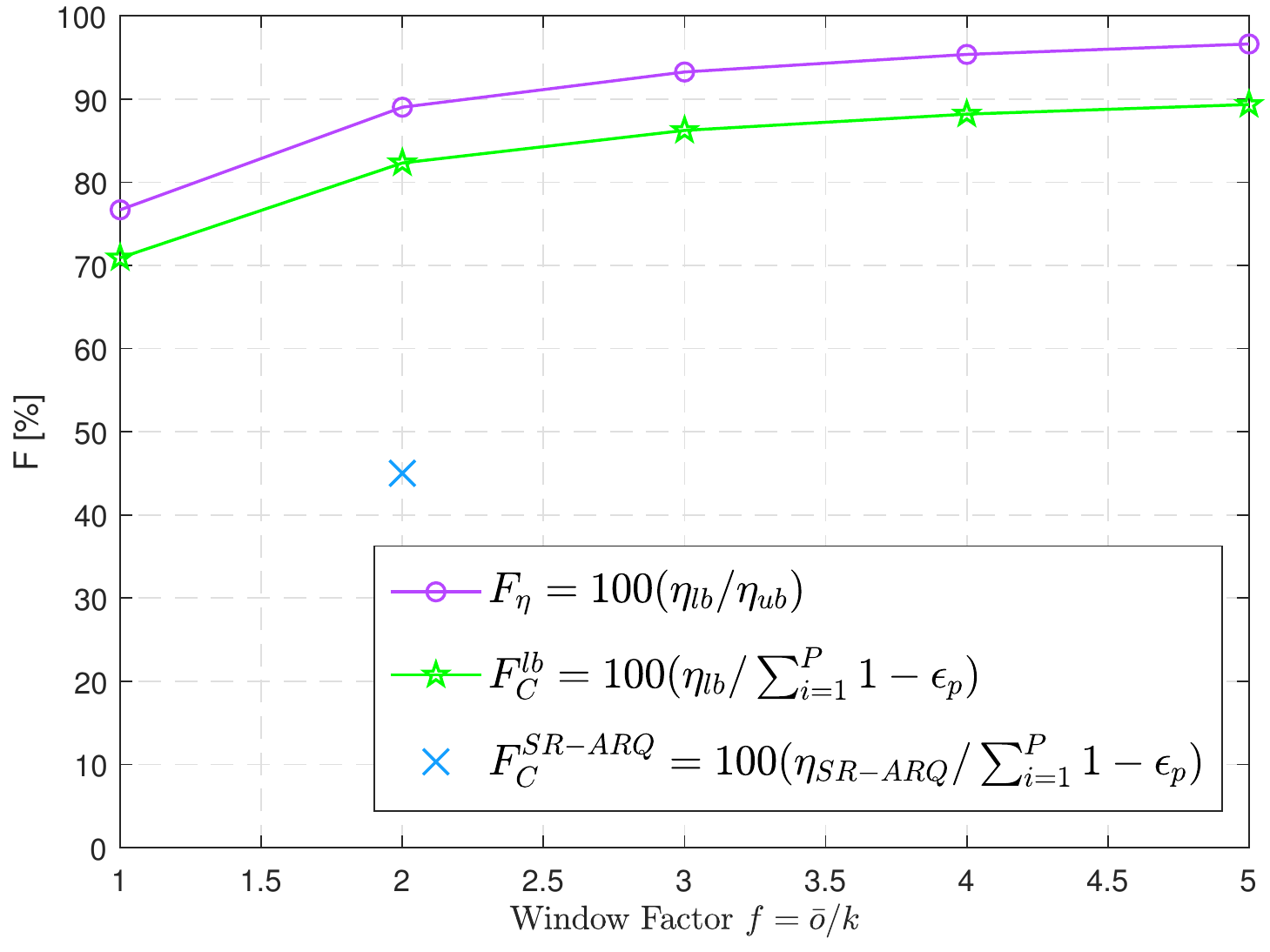}
\else
\includegraphics[trim=0cm 0.7cm 0cm 0cm,width=0.45\columnwidth]{MP_factor_rate2}
\fi
	\caption{\jrev{Throughput factor in MP network with $H=1$ and $P=4$ for BEC's with erasure probability of $\epsilon_{11}=0.2$, $\epsilon_{21}=0.4$, $\epsilon_{31}=0.6$ and $\epsilon_{41}=0.8$. Note that for large maximum size window, the lower bound tends towards the upper bound and the capacity. In this figure, the solid lines correspond to the bounds while the SR-ARQ point comes from the numerical validation of the protocol.}}
	\label{fig:ThroughpuFactorMP}
\end{figure}

~~\\\subsubsection{An Upper Bound for the Mean In-Order Delivery Delay}\label{mean_delay_mp}
In the adaptive coding algorithm suggested in \Cref{MPCodeConstruction}, the number of distinct information packets in $c_t$ is bounded by $\bar{o}$. Hence, for the analysis of the mean in-order delay, we consider the end of a window of $\bar{o}$ new packets, denoted by E$\bar{o}$W.

The retransmission criterion given in \eqref{criterion} reflects, at the sender's best knowledge, the total number of erased packets (i.e. taking into account all the paths together). Hence, the average erasure probability of the MP network $\bar{\epsilon}$ is defined as
\begin{equation}
\label{eq:ebar}
    \bar{\epsilon} = \frac{1}{P} \sum_{p=1}^{P}\epsilon_p .
\end{equation}
In the same manner the maximum rate in \eqref{eq:max_rate} is bounded, we bound the maximum of the mean erasure rate $\bar{\epsilon}_{\max}$ as
\begin{equation*}
    \bar{\epsilon}_{\max} \leq \bar{\epsilon} + \frac{\sqrt{\bar{V}}}{2RTT} =
     \bar{\epsilon} + \frac{\sqrt{2RTT(1-\bar{\epsilon})\bar{\epsilon}}}{2RTT},
\end{equation*}
with $\bar{V}$ denoting the average variance during the period of $2RTT$. In the BEC, $\bar{V}=\sqrt{2RTT(1-\bar{\epsilon})\bar{\epsilon}}$.

Now, using the techniques given in \cite{cohen2019adaptive} for the single path scenario, we upper bound the mean in-order delivery delay of a virtual path with the average erasure probability of the MP network. That is, for this virtual path we define $k_p=k/P$ and $m_e=\bar{o}\bar{\epsilon}=2k_p\bar{\epsilon}$, respectively corresponding to the number of new packets sent over one window on this virtual path and to the effective number of DoFs required by the receiver, and we then apply the procedure described in \cite{cohen2019adaptive} to that virtual path.

\ifdouble
\begin{figure}
\centering
\includegraphics[trim=0cm 0.0cm 0cm 0cm,width=1\columnwidth]{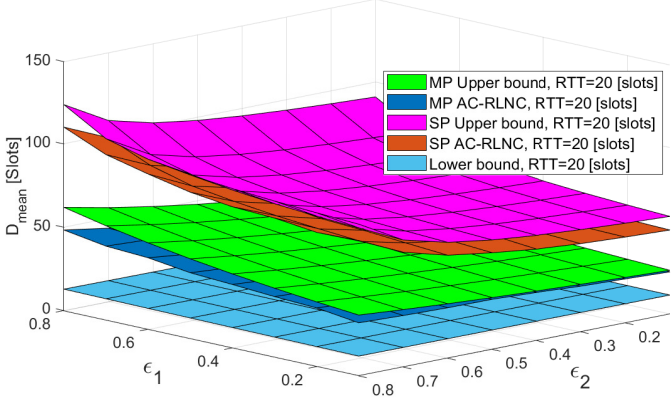}
	\caption{\jrev{Mean in-order delivery delay upper and lower bounds for a MP network with $H=1$ and $P=4$ for BEC's with erasure probability of $\epsilon_{11}=0.2$, $\epsilon_{21}=0.4$, $\epsilon_{31}=0.6$ and $\epsilon_{41}=0.8$, considering the new MP AC-RLNC protocol and independent SP AC-RLNC protocols of \cite{cohen2019adaptive}.}}
	\label{fig:meam_delay_MP}
\end{figure}
\fi

\off{
\paragraph{Average ACK}
When the feedback message is an equivalent ACK on the virtual path,
\ifdouble
\begin{multline}
\label{ACK_BEC}
D_{{mean}[ack\mbox{ }feedback]} \leq (1-\bar{\epsilon}_{\max})\Big[ \mathbb{P}_{E\bar{o} W}(m_e+k_p)\\  + (\mathbb{P}_{\Delta<0})RTT  + (1-\mathbb{P}_{\Delta<0})RTT\Big].
\end{multline}
\else
\begin{equation}
\label{ACK_BEC}
D_{{mean}[ack\mbox{ }feedback]} \leq (1-\bar{\epsilon}_{\max})\Big[ \mathbb{P}_{E\bar{o} W}(m_e+k_p)\\  + (\mathbb{P}_{\Delta<0})RTT  + (1-\mathbb{P}_{\Delta<0})RTT\Big].
\end{equation}
\fi
Grouping together the bounds \eqref{no_FB_BEC}, \eqref{NACK_BEC} and \eqref{ACK_BEC}, the mean delay is bounded by
\ifdouble
\begin{multline}
D_{mean} \leq \lambda D_{{mean}[no\mbox{ }feedback]}+(1-\lambda )\\
(D_{{mean}[nack\mbox{ }feedback]} +D_{{mean}[ack\mbox{ }feedback]} ),
\end{multline}
\else
\begin{equation}
D_{mean} \leq \lambda D_{{mean}[no\mbox{ }feedback]}+(1-\lambda )\\
(D_{{mean}[nack\mbox{ }feedback]} +D_{{mean}[ack\mbox{ }feedback]} ),
\end{equation}
\fi
with $\lambda$ denoting the fraction of time without feedback compared to the total time of transmission.
}

\jrevm{Letting $D_{mean}^{no\mbox{ }feedback}$, $D_{mean}^{nack\mbox{ }feedback}$ and $D_{mean}^{ack\mbox{ }feedback}$ denote respectively the mean in-order delivery delay in case of different feedback states\footnote{\jrevm{The bounds on the mean in-order delivery delay in case of different feedback states are given in \eqref{no_FB_BEC}, \eqref{NACK_BEC} and \eqref{ACK_BEC}.}}, namely with no feedback, NACK feedback and ACK feedback, the upper bound on the mean in-order delivery delay $D_{\text{mean}}$ in MP network is given by the following theorem.
\begin{theorem}\label{upper_mean_mp}
The mean in-order delivery delay $D_{\text{mean}}$ of AC-RLNC in MP network is upper bounded as:
\ifdouble
\begin{multline}\label{eq:mean_bound}
D_{mean} \leq \lambda D_{mean}^{no\mbox{ }feedback}+(1-\lambda )\\
(D_{mean}^{nack\mbox{ }feedback} +D_{mean}^{ack\mbox{ }feedback} ),
\end{multline}
\else
\begin{equation}\label{eq:mean_bound}
D_{mean} \leq \lambda D_{mean}^{no\mbox{ }feedback}+(1-\lambda )\\
(D_{mean}^{nack\mbox{ }feedback} +D_{mean}^{ack\mbox{ }feedback}),
\end{equation}
\fi
where $\lambda$ is the fraction of time without feedback compared to the total time of transmission and where the mean in-order delivery delay in case of the different feedback states are bounded by \eqref{no_FB_BEC}, \eqref{NACK_BEC} and \eqref{ACK_BEC}.
\end{theorem}
\begin{proof}
First, the following probabilities are computed:
\\\noindent (1) {\bf Condition for starting a new generation.} The probability that it is E$\bar{o}$W for the virtual path is:
\begin{equation*}
    \mathbb{P}_{E\bar{o} W} = (1-\bar{\epsilon})^{\bar{o}}.
\end{equation*}
\noindent (2) {\bf Condition for retransmission.}  The probability that $\Delta<0$ for a virtual path is:
\begin{equation*}
   \mathbb{P}_{\Delta<0} = \sum_{i=1}^{\lfloor \bar{o}\bar{\epsilon}_{\max}\rfloor} \binom{\bar{o}}{i}\bar{\epsilon}^{i}(1-\bar{\epsilon})^{\bar{o} -i}.
\end{equation*}
Then, upper bounds for the mean in-order delay for BEC under different feedback states are derived\footnote{Note that as the paths are grouped together in one virtual path, the meaning of NACK and ACK is replaced with an equivalent notion of average NACK and average ACK.}.
\paragraph{No Feedback}
Given that there is no feedback in the virtual path, we have
\ifdouble
\begin{multline}
\label{no_FB_BEC}
 D_{mean}^{no\mbox{ }feedback} \leq \frac{1}{1-\bar{\epsilon}_{\max}}\\ \Big[ \mathbb{P}_{E\bar{o} W}(m_e+k_p)+(1-\mathbb{P}_{E\bar{o} W}) RTT \Big].
\end{multline}
\else
\begin{equation}
\label{no_FB_BEC}
 D_{mean}^{no\mbox{ }feedback} \leq \frac{1}{1-\bar{\epsilon}_{\max}}\\ \Big[ \mathbb{P}_{E\bar{o} W}(m_e+k_p)+(1-\mathbb{P}_{E\bar{o} W}) RTT \Big].
\end{equation}
\fi
\paragraph{Average NACK}
When the feedback message is an equivalent NACK for the virtual path,
\ifdouble
\begin{multline}
\label{NACK_BEC}
D_{mean}^{nack\mbox{ }feedback} \leq \bar{\epsilon}_{\max}\frac{1}{1-\bar{\epsilon}_{\max}}\\ \Big[ \mathbb{P}_{\Delta<0}\big[ (1-\mathbb{P}_{E \bar{o} W})RTT+\mathbb{P}_{E\bar{o} W}(m_e+k_p)\big] \\  +(1-\mathbb{P}_{\Delta<0})\big[ RTT + \mathbb{P}_{E\bar{o} W}(m_e+k_p) \big]\Big].
\end{multline}
\else
\begin{multline}
\label{NACK_BEC}
D_{mean}^{nack\mbox{ }feedback} \leq \bar{\epsilon}_{\max}\frac{1}{1-\bar{\epsilon}_{\max}} \Big[ \mathbb{P}_{\Delta<0}\big[ (1-\mathbb{P}_{E \bar{o} W})RTT+\mathbb{P}_{E\bar{o} W}(m_e+k_p)\big] \\  +(1-\mathbb{P}_{\Delta<0})\big[ RTT + \mathbb{P}_{E\bar{o} W}(m_e+k_p) \big]\Big].
\end{multline}
\fi
\paragraph{Average ACK}
When the feedback message is an equivalent ACK on the virtual path,
\ifdouble
\begin{multline}
\label{ACK_BEC}
D_{mean}^{ack\mbox{ }feedback} \leq (1-\bar{\epsilon}_{\max})\Big[ \mathbb{P}_{E\bar{o} W}(m_e+k_p)\\  + (\mathbb{P}_{\Delta<0})RTT  + (1-\mathbb{P}_{\Delta<0})RTT\Big].
\end{multline}
\else
\begin{equation}
\label{ACK_BEC}
D_{mean}^{ack\mbox{ }feedback} \leq (1-\bar{\epsilon}_{\max})\Big[ \mathbb{P}_{E\bar{o} W}(m_e+k_p)\\  + (\mathbb{P}_{\Delta<0})RTT  + (1-\mathbb{P}_{\Delta<0})RTT\Big].
\end{equation}
\fi
Grouping together the bounds \eqref{no_FB_BEC}, \eqref{NACK_BEC} and \eqref{ACK_BEC}, we achieves the upper bound provided in Theorem \ref{upper_mean_mp}, that is the mean delay is bounded
\ifdouble
\begin{multline}
D_{mean} \leq \lambda D_{mean}^{no\mbox{ }feedback}+(1-\lambda )\\
(D_{mean}^{nack\mbox{ }feedback} +D_{mean}^{ack\mbox{ }feedback} ),
\end{multline}
\else
\begin{equation}
D_{mean} \leq \lambda D_{mean}^{no\mbox{ }feedback}+(1-\lambda )\\
(D_{mean}^{nack\mbox{ }feedback} +D_{mean}^{ack\mbox{ }feedback}),
\end{equation}
\fi
with $\lambda$ denoting the fraction of time without feedback compared to the total time of transmission.
\end{proof}
}

~~\\In the following, the upper bound \eqref{eq:mean_bound} is compared to the MP AC-RLNC simulations in \ifdouble\Cref{fig:meam_delay_MP}\else\Cref{fig:ThroughputErrorBecMP}-(c)\fi, where the bound is shown in green and  the simulation results of \Cref{MPsimulation} in blue. We can note that the  analytical results are in agreement with the simulation of the MP AC-RLNC solution. To gain understanding on the advantage of the MP solution with regards to the SP solution, the MP AC-RLNC performances are also compared with independent SP AC-RLNC protocols suggested in \cite{cohen2019adaptive} on each of the paths, as represented on the bottom of \ifdouble\Cref{fig:dif_sp_mp}\else\Cref{fig:encoding_pros}-(b)\fi. The mean in-order delivery bound (in pink) and the simulation result (in orange) highlight the advantage of using the MP solution suggested in this paper compared to independent SP solutions proposed in \cite{cohen2019adaptive}.

\begin{figure}
    \centering
    \includegraphics[trim=0cm 0.0cm 0cm 0cm,width=1\columnwidth]{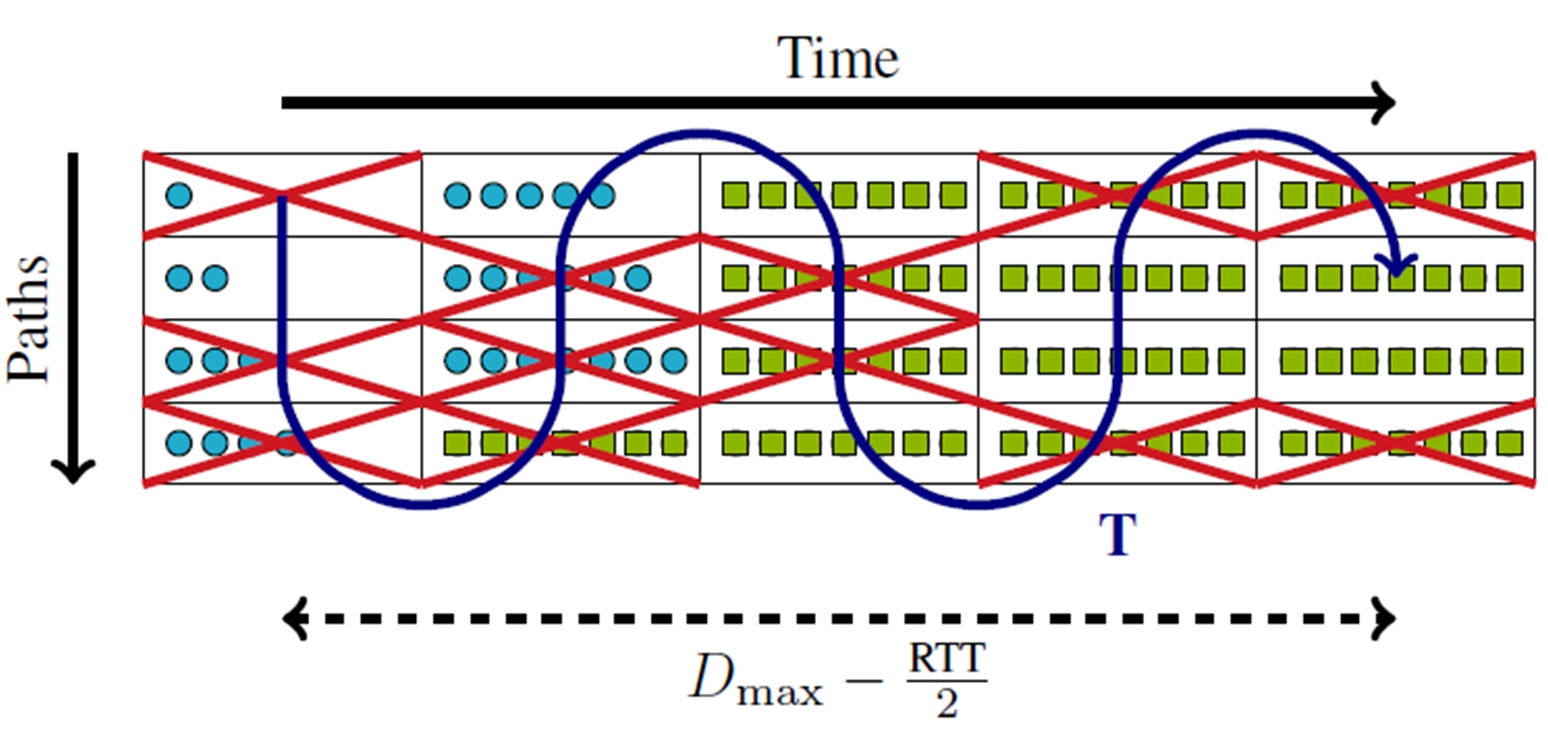}
\caption{Maximum delay transmission, with $\bar{o}=7$ : for each time slot and path, the circles or squares represent the raw packets of information coded together in the RLNC.  The blue circles highlight the new transmissions while the green squares correspond to repetitions due to the \textit{size-limit} mechanism. The red crosses show which RLNC's are erased. One can check that at the last time slot, 7 RLNC's have been collected, hence allowing the decoding. The interval of time before decoding corresponds to the in-order delivery delay up to a $\nicefrac{\text{RTT}}{2}$ term, as the transmission time has to be taken into account. In this example, the number of needed transmissions is $T = 18$, and $D_{\max} = 5 + \nicefrac{\text{RTT}}{2}$ . }\label{fig:max_delay}
\end{figure}
~~\subsubsection{An Upper Bound for the Maximum In-Order Delivery Delay}
Contrarily to the two previous analyses, the maximum in-order delivery delay is bounded in a different way than in \cite{cohen2019adaptive}, as the SP bounds cannot be easily generalized to the MP network.
\\Considering the transmission of a new generation of raw packets, the decoding of the first one can occur at 4 different moments (ranked from the earliest to the latest): (1) after the first transmission; (2) after a FEC transmission; (3) after a FB-FEC transmission; (4) after a transmission due to the \textit{size-limit mechanism}. In a worst case approach, the FEC and FB-FEC mechanisms are neglected, and the transmissions hence occur as represented in \Cref{fig:max_delay}: the $\bar{o}$ first transmissions each contain a new packet of information (blue circles).  Once the size limit is reached, the same RLNC is sent till successful decoding (green squares), after $T$ transmissions.

\jrevm{\begin{theorem}
The maximum in-order delivery delay $D_{\max}$ of AC-RLNC in MP network is upper-bounded with error probability $P_e$ as:
\begin{equation}
 D_{\max}\leq \left \lceil \frac{RTT}{2} \right \rceil +\left \lceil \frac{T_{\max}}{P}  \right \rceil,
 \label{eq:Dmaxbound}
\end{equation}
with the maximum number of transmissions being defined as
\ifdouble
\begin{equation}
\begin{split}
    T_{\max} \triangleq \left \lceil 1+\frac{\bar{o}-1}{1-\bar{\epsilon}}\right.\hspace{4.5cm}\\\hspace{1cm}\left.+\frac{\alpha}{4(1-\bar{\epsilon})^2} + \frac{\sqrt{\alpha\left(\alpha+4(1-\bar{\epsilon})(\bar{o}-1)\right) }}  {2} \right \rceil,
    \label{eq:Tmax1}
\end{split}
\end{equation}
\else
\begin{equation}
    T_{\max} = \left \lceil 1+\frac{\bar{o}-1}{1-\bar{\epsilon}}+\frac{\alpha}{4(1-\bar{\epsilon})^2} + \frac{\sqrt{\alpha\left(\alpha+4(1-\bar{\epsilon})(\bar{o}-1)\right) }}  {2} \right \rceil,
    \label{eq:Tmax1}
\end{equation}
\fi
where $\alpha \triangleq  \log\left(\frac{\epsilon_{\max}}{P_e}\right)$ and $\epsilon_{\max} \triangleq \max_{p=1...P}\epsilon_p  $.
\end{theorem}}
\jrevm{\begin{proof}
Given an error probability $P_e$\off{(typically $10^{-3}$)}, one is interested in determining the number of transmissions $T_{\max}$ needed to decode the first packet with probability $1-P_e$:
\ifdouble
 \begin{equation*}
     T_{\max}\quad s.t.\quad  \mathbb{P}\left[T>T_{\max}\right]\leq P_e.
 \end{equation*}
\else
 $T_{\max}\quad s.t.\quad  \mathbb{P}\left[T>T_{\max}\right]\leq P_e.$
\fi
Decoding of the first packet is not possible after $T_{\max}$ transmissions if two conditions are fulfilled: (1) the first transmission is erased; (2) Among the $T_{\max}-1$ remaining transmissions, at most $\bar{o}-1$ successful transmissions occur (or equivalently, at least $T_{\max}-\bar{o}$ erasures occur). Indeed, once the first packet is erased, no decoding is possible before reaching the size limit, as the number of received RLNC's will always be at least one step behind the number of raw packets coded in the RLNC. Hence, the $\bar{o}$ packets will be decoded jointly. Letting  $E_i$ be the random variable equal to $1$ if the $i$th transmission is erased and $0$ otherwise, the probability of no-decoding can be bounded as\footnote{The inequality comes from the fact that the FEC and FB-FEC mechanisms are neglected.},
\begin{equation}
    \mathbb{P}\left[T>T_{\max}\right] \leq  \mathbb{P}\left[E_1\right]\mathbb{P}\left[\sum_{i=2}^{T_{\max}}E_i\geq T_{\max} - \bar{o}\right].
    \label{eq:tmax}
\end{equation}
Defining
\[S_e \triangleq \frac{1}{T_{\max}-1} \sum_{i=2}^{T_{\max}}E_i,\]
$T_{\max}$ can be identified through the cumulative distribution function (CDF) of that random variable, whose expectation is $\bar{\epsilon}$, defined in \eqref{eq:ebar}. However, since the erasure probability of each path is different, $S_e$ is the average of independent, but not identically distributed, random variables. It hence follows a Poisson Binomial distribution \cite{PoissonBinomial}, whose CDF becomes quickly difficult to compute in an efficient way. Hence, to obtain a closed-form expression of $T_{\max}$, the Hoeffding inequality \cite{Hoeffding} is used to obtain the following,
\ifdouble
\begin{multline*}
\mathbb{P}\left[S_e\geq \frac{T_{\max} - \bar{o}}{T_{\max}-1}\right]
    \leq \\ \exp{\left(-2\left(T_{\max}-1\right)\left(\frac{T_{\max}-\bar{o}}{T_{\max}-1} - \bar{\epsilon}\right)^2\right)}.
\label{eq:hoeffding}
\end{multline*}
\else
\begin{equation*}
\mathbb{P}\left[S_e\geq \frac{T_{\max} - \bar{o}}{T_{\max}-1}\right]
    \leq \\ \exp{\left(-2\left(T_{\max}-1\right)\left(\frac{T_{\max}-\bar{o}}{T_{\max}-1} - \bar{\epsilon}\right)^2\right)}.
\label{eq:hoeffding}
\end{equation*}
\fi
Now, since $\mathbb{P}\left[E_1\right] \leq  \epsilon_{\max}$, requiring the upper bound of \eqref{eq:tmax} to be smaller than $P_e$, $T_{\max}$ is such that
\begin{equation*}
    \left(T_{\max}-1\right)\left(\frac{T_{\max}-\bar{o}}{T_{\max}-1} - \bar{\epsilon}\right)^2\geq \frac{1}{2}\log\left(\frac{\epsilon_{\max}}{P_e}\right).
\end{equation*}
Using the definition of $\alpha$ and solving for $T_{\max}$, \eqref{eq:Tmax1} is obtained. Directly from \Cref{fig:max_delay}, the maximum delay is thus bounded, with a probability $P_e$,  as given in \eqref{eq:Dmaxbound}, with the $RTT$ factor coming from the transmission time.
\end{proof}}

From (\ref{eq:Tmax1}) and (\ref{eq:Dmaxbound}), one can see that the benefit lies in the average of the erasure probabilities, as $T_{\max}$ grows rapidly when $\bar{\epsilon}$ is close to 1. However, as $\bar{\epsilon}$ is the average erasure probability, the worst paths will be balanced by the better ones, hence pushing $\bar{\epsilon}$ away from $1$, and leading to smaller delays.

\begin{figure}
\centering
\ifdouble
\includegraphics[trim=0cm 0.4cm 0cm 0cm,width=1\columnwidth]{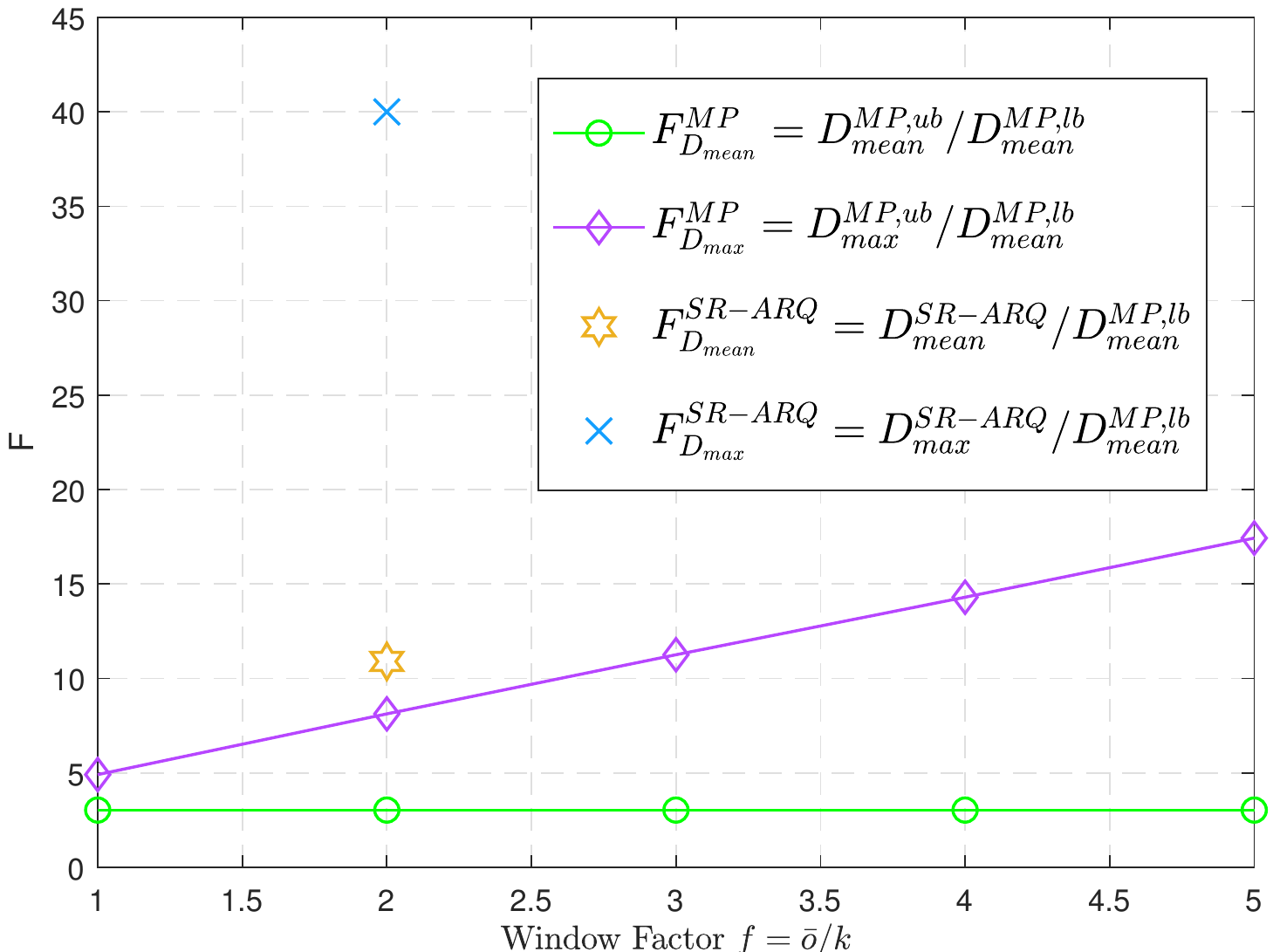}
\else
\includegraphics[trim=0cm 0.5cm 0cm 0cm,width=0.45\columnwidth]{MP_delay_factor2}
\fi
	\caption{\jrev{Delay factor in MP network with $H=1$ and $P=4$ for BEC channels with erasure probability of $\epsilon_{11}=0.2$, $\epsilon_{21}=0.4$, $\epsilon_{31}=0.6$ and $\epsilon_{41}=0.8$. Note that the mean delay is not influenced by the window size factor while the maximum delay increases linearly with it. In this figure, the solid lines correspond to the bounds while the SR-ARQ points come from the numerical validation of the protocol. }}
	\label{fig:DelayFactorMP}
\end{figure}

\ifdouble
\begin{figure*}[!ht]
\centering
\includegraphics[trim=0cm 0.0cm 0cm 0cm,width=1\textwidth]{RTT20_MP_paper_v3}
    \caption{Performances of multipath protocol with RTT$=20$ [slots], $th=0$ and $\bar{o}=2k$, on 4 paths with $\epsilon_3 = 0.2$ and $\epsilon_4 = 0.8$, averaged on $150$ iterations. The vertical bars correspond to the standard deviation of the simulated results.}
    \label{MP_perf}
\end{figure*}
\fi

\subsubsection{\jrev{A Lower Bound on the Mean and Maximum in-Order Delivery Delay}}\label{mean_delay_mp_lb}
\jrev{To decrease the delay as much as possible, one could send the same packet on all the paths at each time slot until an acknowledgment is received. In that case, that probability that the packet is not received at time $t$ can be written as
\begin{equation*}
    P[\text{Not received}] = P[\text{Erased on all paths}] = \prod_{p=1}^P\epsilon_p,
\end{equation*}
which we define as $\epsilon_{\text{prod}}$. As this is clearly the best anyone can do, the in-order delay $D$ is thus bounded as
\begin{equation}
    D \geq \frac{\text{RTT}}{2}+\frac{1}{1-\epsilon_{\text{prod}}},
    \label{eq:bound_prod},
\end{equation}
where the delay suffers from the propagation delay and the expected number of transmissions needed to receive the packet.\\
Note that if the paths are of relatively good quality, $\epsilon_{\text{prod}}$ will be very close to $0$, the bound reducing hence to half of the RTT. In the following, we make the choice not to compare the in-order delay results to this bound, but instead to assess them in light of the optimal genie-aided lower delay of the system when different packets are sent on each of the paths.}\jrevm{
\begin{theorem}
The in-order delivery delay $D$ of AC-RLNC in MP network is lower-bounded as
\begin{equation*}
    D^{\text{genie-aided}} \geq \frac{\text{RTT}}{2}+\frac{1}{1-\bar{\epsilon}},
\end{equation*}
if different packets are sent on each of the paths.
\end{theorem}
\begin{proof}
If different packets are sent on each path, the MP network can be seen from the delay point of view as a virtual path suffering from the mean erasure probability. Building on \eqref{eq:bound_prod}, the theorem is obtained.
\end{proof}
It is important to note that the above bound can only be achieved when the sender is able to predict which packets will not be delivered to the receiver (hence the name \textit{genie-aided}).} \jrev{In \ifdouble\Cref{fig:meam_delay_MP}\else\Cref{fig:ThroughputErrorBecMP}-(c)\fi, this lower bound is compared to the performance of the MP AC-RLNC protocol for a 4-path network. As one can expect, a non-negligible gap remains between the bound and the achieved performances. This gap is highlighted in details in \Cref{fig:DelayFactorMP}, where one can observe that the mean delay bound does not depend on the window size factor while the max delay bound increases linearly with it. For the comparison, the SR-ARQ numerical validation leads to an in-order delay $10.9$ times above the bound for the mean delay and $40$ times above for the maximum one. }

\off{
To obtain a optimal lower bound for the mean in-order delivery delay we will consider the average erasure probability given in \eqref{eq:ebar}. To compensate the erasure packets on average it is required $1/(1-\bar{\epsilon})$ transmissions. Hence, when the delay in the forward channel between the sender to the receiver is $RTT/2$, the lower bound for the mean in-order delivery delay is given by,
\begin{equation*}
    D_{mean}^{lb} \geq \frac{RTT/2}{1-\bar{\epsilon}}.
\end{equation*}
See \Cref{fig:meam_delay_MP}. Need to explain the results.... Comparison with approximation factor: $F_{D_{mean}}^{\text{AC-RLNC}} = \frac{D_{mean}^{ub}}{D_{mean}^{lb}}$ (Note that this factor is equal to 2, probably since our solution is causal from the feedback acknowledgments. Hence we now have the effect of the feedback channel delay, which is in our analysis precisely the same as the delay in the forward channel, $RTT/2$) with $F_{D_{mean}}^{\text{SR-ARQ}} = \frac{D_{mean,\text{SR-ARQ}}}{D_{mean}^{lb}}$ (for SR-ARQ we need to add example from the simulations).
\revGui{
Consider one packet is sent on 1 path with erasure probability $\epsilon$,
then in average, the number of retransmissions $m$ will be the solution of
\begin{equation*}
    (1-\epsilon)n = k,
\end{equation*}
leading to $m=\frac{\epsilon}{1-\epsilon}$. Adding the propagation delay, we have}}

\ifdouble\else
\begin{figure*}
\centering
\includegraphics[trim=0cm 0.0cm 0cm 0cm,width=1\textwidth]{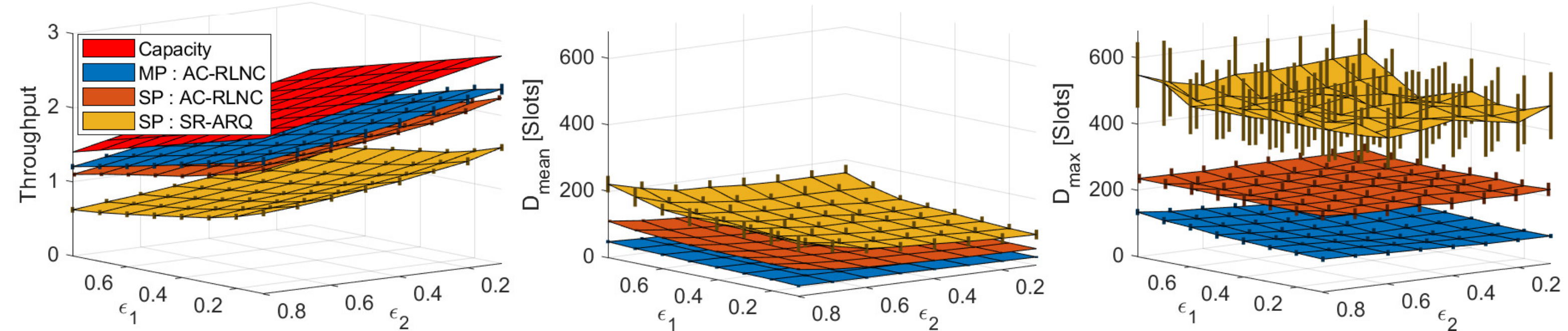}
    \caption{Performances of multipath protocol with RTT$=20$ [slots], $th=0$ and $\bar{o}=2k$, on 4 paths with $\epsilon_3 = 0.2$ and $\epsilon_4 = 0.8$, averaged on $150$ iterations. The vertical bars correspond to the standard deviation of the simulated results.}
    \label{MP_perf}
\end{figure*}
\fi

\subsection{Simulation Results}\label{MPsimulation}

The performance of the MP AC-RLNC protocol is compared with two other protocols, as presented in \Cref{MP_perf}.

\paragraph{Setting and protocols}
We consider the setting of \Cref{MH_MP_setting}, with $H=1$, $P=4$, $RTT=20$, and with $\epsilon_{31}= 0.2$ and $\epsilon_{41} =0.8$, while the erasure probabilities of the two other paths ($\epsilon_{11}$ and $\epsilon_{21}$) vary in the range of $[0.1 \;  0.8]$. \jrev{These erasure probabilities correspond to those observed in the controlled-congested setup considered by Intel for WiFi standards tests which is described in our singlepath AC-RLNC experimental validation \cite{cohen2019adaptive}.}

The MP AC-RLNC protocol has been simulated with $th=0$ and $\bar{o}=2w$. To emphasize the gain we get by tracking the channel condition, adjusting the retransmissions, and with the discrete \rev{BF} algorithm, we compare the MP protocol with the SP AC-RLNC \cite{cohen2019adaptive}, applied independently on each path as described in the bottom part of \ifdouble\Cref{fig:dif_sp_mp}\else\Cref{fig:encoding_pros}-(b)\fi. The tunable parameters of the SP AC-RLNC protocol are the same as the MP ones. Furthermore, we compare these protocols with SR-ARQ \ar{\cite{weldon1982improved,anagnostou1986performance,ausavapattanakun2007analysis}}, to show the gain we get with adaptive and causal network coding. Again, the SR-ARQ protocol is used independently on each path. \jrev{Other traditional coding solutions could also be applied independently on each path. For a discussion on the performance on such solutions, we refer the reader to \cite{cohen2019adaptive} where delay-aware coding techniques are presented and compared to SP AC-RLNC.} The results of \Cref{MP_perf} have been averaged on 150 different channel realizations, where the filled curves correspond to the mean performances while the error bars represent the standard deviation from the mean.

\paragraph{Results}
Based on the results of \Cref{MP_perf}, the throughput is slightly increased (around $15\%$) compared to the SP AC-RLNC protocol but nearly doubled with regards to SR-ARQ, independently of the erasure probabilities. On the delay point of view, both the mean and max in-order delay are dramatically reduced with the MP protocol. More precisely, compared to the SP AC-RLNC protocol, the MP protocol performs nearly twice as better in term of mean delay and between 2 to 3 times better in terms of max delay. With regards to the SR-ARQ protocol, performances are nearly 4 times better for the mean delay and between 4 to 6 times better for the max one. For higher RTT's, the gap between our solution and other protocols increases\footnote{Those results are not shown here \rev{due to limited space}.}. \jrev{The above simulation results as well as the bounds ensure that the protocol performances do not decrease dramatically when the RTT increases. This also supports situations in which the RTT fluctuates, as in the case, one can consider in a worst case approach the largest RTT and still have guaranteed throughput and in-order delay. Finally, if the feedback channel suffers from erasures, then one can use a cumulative feedback (e.g. \cite{malak2019tiny}) which will translates the feedback erasures in a larger RTT as the ACK or NACK will simply be delayed. Hence, the bounds also provide guarantees in this situation.}
\ifdouble
\ar{\begin{table}
    \centering
    \normalsize
    \begin{tabular}{|l|l|}
        \hline
        \textbf{Param.} & \textbf{Definition}\\\hline
        $C$ & capacity\\\hline
        $L$, $G$ & local and global matching \\\hline
        $\mathcal{L}$, $\mathcal{G}$ & set of admissible local and global matchings\\\hline
        $\eta_{max}$ & maximum throughput of the global paths\\\hline
        $r_{G(p,h)h}$ & rate of the $G(p,h)$-th path in the $h$-th hop \\\hline
        $r_{G_p}$ & rate of each $p$-th global path $\forall p=\{1;\ldots;P\}$ \\\hline
    \end{tabular}
    \vspace{0.1cm}
    \caption{Symbol definition for the multihop protocol}
    \label{table : definition_MH}
\end{table}}
\else
\begin{figure}
    \centering
    \includegraphics[trim=0cm 0.7cm 0cm 0cm,width = 1 \columnwidth]{tabal_algo_2.jpg}
\end{figure}
\fi

\section{Multi-hop Multipath Communication}\label{MH}
In this section, we generalize the MP solution given in \Cref{MP} to the MP and MH setting introduced in \Cref{sys}, in which each node can estimate the erasure probability of the incoming and outgoing paths according to the local feedback. To present our MH protocol, we use the example of \Cref{ex_balancing}.

In that MH setting, in the asymptotic regime, using RLNC, the min-cut max-flow capacity $C$ ($2.6$ in the example) can be achieved by mixing together coded packets from all the paths at each intermediate node (see \Cref{rel}). Hence, one could use $P$ parallel SP AC-RLNC protocols with the node recoding protocol to get a throughput very close to the min-cut max-flow capacity. However, due to the mixing between the paths, dependencies are introduced between the FEC's and the new RLNC's. This will thus result in a high in-order delay.

To reduce the in-order delay, the MP algorithm we suggest in \Cref{MP} can be used on $P$ global paths, using RLNC independently on each path. A naive choice of global paths is shown in the upper part of \Cref{ex_balancing}\footnote{The paths are described by their color and their type of arrow.}. In that setting, due to the min-cut max-flow capacity, the maximum throughput of each path is limited by its bottleneck\footnote{ In \Cref{ex_balancing}, bottlenecks are denoted by a curvy symbol behind the rate.} (i.e. the link with the smallest rate). Doing so, the maximum throughput (i.e. the sum of the min-cut of each path) can thus be much lower than the capacity of the network, as it can be seen from the example of \Cref{ex_balancing}, in which the throughput is $1.5$.

\rev{With these two previous attempts minded}, we finally suggest to determine the global paths using a decentralized balancing algorithm whose \jrevm{aim is to maximize the throughput} of the network. The second part of \Cref{ex_balancing} shows the global paths resulting from that balancing, and as a result a maximum throughput equal to $2.4$. Moreover, it can be seen that only 2 global paths are now affected by the bottleneck links. Once these global paths are defined, the MP AC-RLNC protocol can be used as described in \ifdouble\Cref{pseudo_code}\else Algorithm 1\fi.
\ifdouble
\begin{figure}
\centering
\includegraphics[trim=0cm 0.7cm 0cm 0cm,width=1\columnwidth]{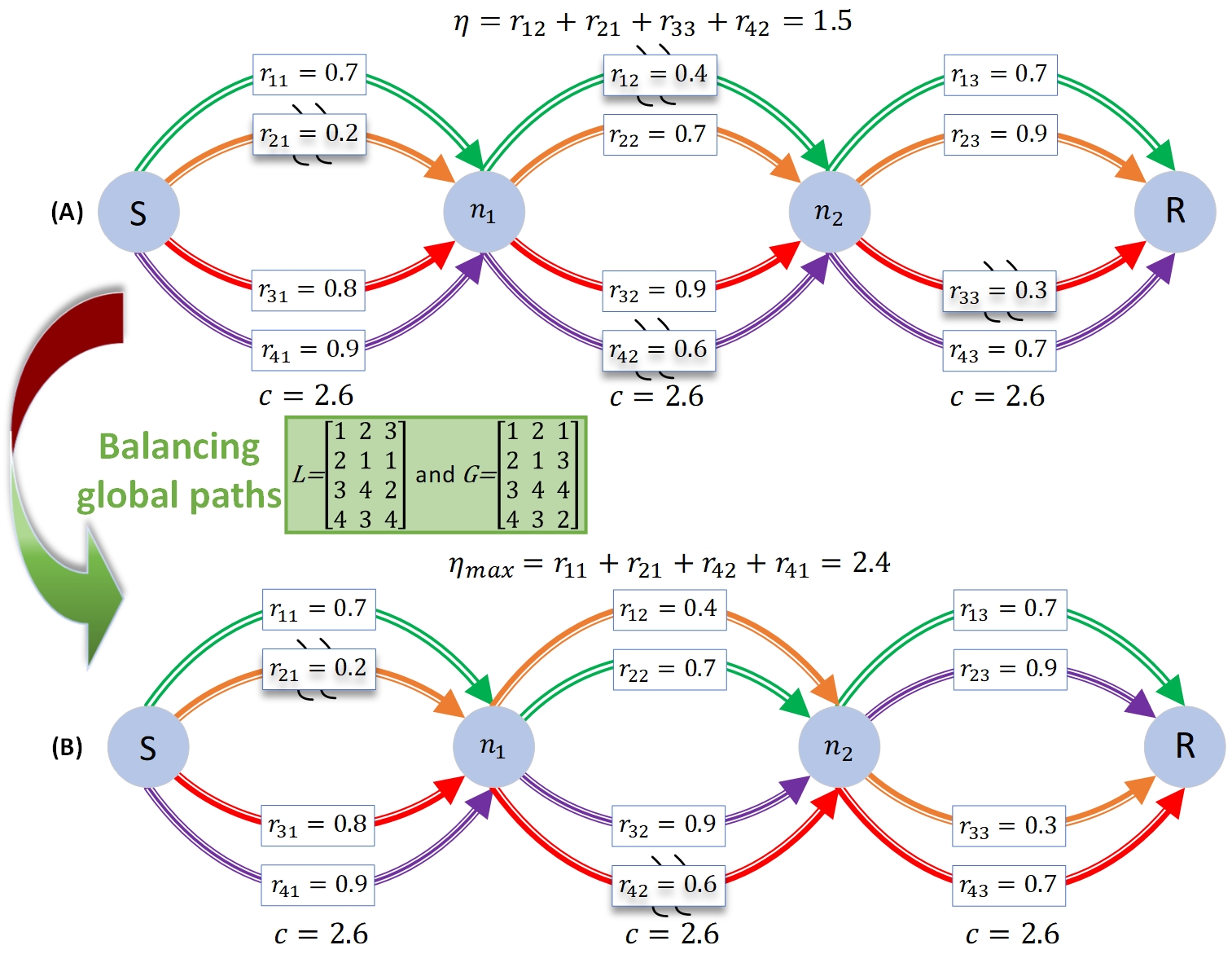}
     \caption{Balancing optimization and global paths example. The numbers correspond to the rates of the paths.}
    \label{ex_balancing}
\end{figure}
\else
\begin{figure}
\centering
\includegraphics[trim=0cm 0.7cm 0cm 0cm,width=0.55\columnwidth]{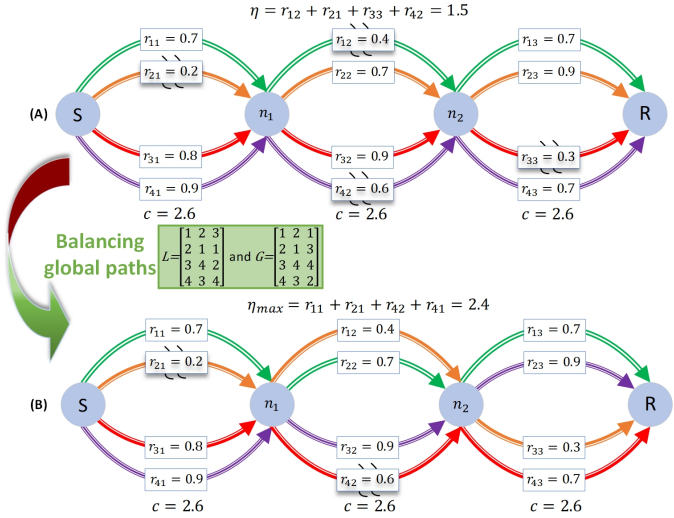}
     \caption{Balancing optimization and global paths example. The numbers correspond to the rates of the paths.}
    \label{ex_balancing}
\end{figure}
\fi

In section \ref{MHCodeConstruction}, the definition of the global paths is described as well as the full multi-hop algorithm, which is then analyzed theoretically in section \ref{Analytical_results_MH} and finally simulated in \Cref{MHsimulation}. \ifdouble\Cref{table : definition_MH} \else Table II \fi summarizes the symbol definitions we use in this section.


\subsection{Adaptive Coding Algorithm}\label{MHCodeConstruction}
Here we describe the suggested MP and MH solution.

\paragraph{Global paths - problem formulation}
In order for the $h$-th node to transmit packets over the paths maximizing the rate, it needs to know the local matching $L(p,h)$, such that, $L(p,h) = j$ implies that the $j$-th path of the $(h+1)$-th hop is matched with the $p$-th path of the $h$-th hop. The definition of the global paths can be done equivalently through a global matching $G(p,h)$, such that, $G(p,h) = j$ implies that the $j$-th path of the $h$-th hop belongs to the $p$-th global path\footnote{The global matching of the first hop is such that the $p$-th local path belongs to the $p$-th global path, i.e. $G(p,1)=p \, \forall p=1...P$.}. We point out that even if these two definitions are equivalent, the local matching is particularly convenient to express the global paths in a decentralized way. Moreover, it is important to note that, for $L$ and $G$ to be an admissible matching, each local path must be matched with exactly one other local path at each node. Hence, $\mathcal{L}$ and $\mathcal{G}$ are defined respectively as the set of admissible local and global matchings.
The values of $L$ and $G$, \rev{shown} in the example given in \Cref{ex_balancing}, are respectively equal to
\begin{equation*}
    \textstyle L = \begin{pmatrix}
    1&2&3 \\
    2&1&1 \\
    3&4&2 \\
    4&3&4 \\
    \end{pmatrix}  \textstyle \quad \text{and} \textstyle \quad G = \begin{pmatrix}
    1&2&1 \\
    2&1&3 \\
    3&4&4 \\
    4&3&2 \\
    \end{pmatrix}.
\end{equation*}
Once admissible global paths are determined, the maximum achievable throughput $\eta_{max}$ corresponds to the sum of the min-cut of each global path \cite{mincutmaxflowtheorem,LunMedKoeEff2008}. Defining $r_{G(p,h)h}$ as the rate of the $G(p,h)$-th path of the $h$-th hop, $\eta_{max}$ is thus equal to
\begin{equation*}
     \eta_{max}(G) =  \sum_{p=1}^P\min_{h=1...H}r_{G(p,h)h}.
\end{equation*}
\ar{Consequently, since $G$ and $L$ are equivalent, the global path problem can be expressed as
\begin{equation} \label{prob}
    L = \arg\max_{\Tilde{L}\in\mathcal{L}}\eta_{max}(\Tilde{L}).
\end{equation}}
Note that this problem admits in general several solutions, as one can see from \Cref{ex_balancing}. For instance, letting $r_{11}$ be matched with $r_{42}$ and $r_{31}$ with $r_{22}$, $\eta_{max}$ is unchanged. In the following, a decentralized solution of (\ref{prob}) is first presented and secondly, we show that this solution minimizes the bottlenecks of the network.

\paragraph{Global paths - decentralized solution}
\begin{theorem}[Optimal matching]
Considering paths are sorted in rate-decreasing order at each hop (i.e. $r_{1,h}\geq r_{2,h}\geq ... \geq r_{P,h} \, \forall h=1...H$), a matching solving (\ref{prob}) is such that $r_{p,h}$ is matched with $r_{p,h+1} \quad \forall p=1...P,\,h=1...H$. This matching is called the natural matching.
\label{natural_matching}
\end{theorem}
\begin{proof}
By contradiction, suppose an optimal matching $L^*$ resulting in a strictly higher $\eta_{max}$ than the natural matching $L$:
\begin{equation}
    \eta_{max}\left(L^*\right) > \eta_{max}\left(L\right).
    \label{to_contradict}
\end{equation}
In the following, first we prove that at the last hop, the natural matching has at least a rate equal to the optimal one. Then, in appendix, the proof is extended to the previous hops, hence contradicting \eqref{to_contradict}.

\textbf{Last hop : }
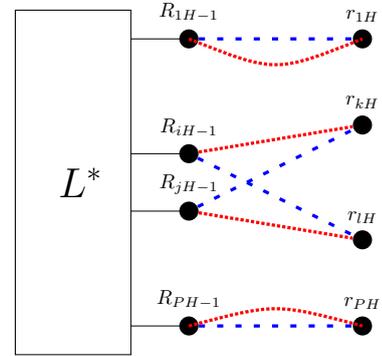
\begin{figure}[!h]
\centering
\resizebox{5cm}{!}{\input{Matching_last_hop.tikz}}
\caption{Matching of the last hop. The blue dashed lines highlight the non-natural matching while the red dotted lines correspond to the natural one.}\label{fig:matching_last_hop}
\end{figure}
Letting the achieved rate up to the $h$-th hop be $R_{ph} = min_{f=1...h}r_{G^*(p,f)f}$ and considering $R_{ph}$ is sorted in rate-decreasing order, suppose $L^*$ does not give a natural matching at the last hop, as represented in \Cref{fig:matching_last_hop} with blue dashed lines. Hence $\exists \, i>j,k>l$ such that $R_{iH-1}\, (resp.\, R_{jH-1})$ is matched with $R_{lH-1}\, (resp.\,R_{kH-1})$. We prove below that matching $R_{iH-1}\, (resp.\, R_{jH-1})$ with $r_{kH}\, (resp.\,r_{lH})$ does not decrease the sum-rate\footnote{Since unmodified links do not modify the rate, they don't need to be taken into account in the proof. By abuse of notation, we let $\eta_{max}$ be the rate restricted to the modified links.}
\begin{equation*}
   \eta_{max}\left(L^*\right) =  \min\left(R_{iH-1},r_{lH}\right) + \min\left(R_{jH-1},r_{kH}\right).
\end{equation*}
  Two cases are possible :
\\If $r_{lH} \leq R_{jH-1}$, since $R_{jH-1}<R_{iH-1}$,
\ifdouble
\begin{equation*}
   \eta_{max}\left(L^*\right) \leq \min\left(R_{jH-1},r_{lH}\right) + \min\left(R_{iH-1},r_{kH}\right).
\end{equation*}
\else
    $\eta_{max}\left(L^*\right) \leq \min\left(R_{jH-1},r_{lH}\right) + \min\left(R_{iH-1},r_{kH}\right)$.
\fi

\noindent If $R_{jH-1} < r_{lH}$, since $r_{lH}<r_{kH}$,
\ifdouble
\begin{equation*}
   \eta_{max}\left(L^*\right) \leq \min\left(R_{iH-1},r_{kH}\right) + \min\left(R_{jH-1},r_{lH}\right).
\end{equation*}
\else
    $\eta_{max}\left(L^*\right) \leq \min\left(R_{iH-1},r_{kH}\right) + \min\left(R_{jH-1},r_{lH}\right)$.
\fi

\noindent In both cases, $\eta_{max}\left( L^*\right)$ is upper bounded by the rate achieved by the natural matching, highlighted by red dotted lines in \Cref{fig:matching_last_hop}.
Applying this reasoning recursively until no $i,j,k,l$ can be found, it appears $L^*$ can be replaced by the natural matching at the last hop without decreasing $\eta_{max}$.

~~\\\jrevm{\textbf{Previous hop:}} The proof for the previous hops uses the same ideas as for the last one recursively, and is presented in Appendix \ref{appendix:appendixMatching}.
\end{proof}
The resulting matching procedure is presented in \ifdouble \Cref{greedy_balancing}\else Algorithm 2\fi, where the paths are greedily matched from the first to the last hop. Note that this procedure is almost fully decentralized, as sorting the hop rates and matching them can be done locally, at each node. The only needed communication arise\jrevm{s} when two links have the same rate, since it that case the node needs to know which of the links have been considered as \rev{``}best" at the previous node. Hence, by default, the order of the paths is forwarded to the next node in \ifdouble\Cref{greedy_balancing}\else Algorithm 2\fi. We also stress out that this algorithm is much \jrev{simpler} than the one previously suggested in \cite{cohen2019MPMH}, as it does not require to apply the Hungarian algorithm \cite{assignment_problem} at each node.
\ifdouble
\begin{algorithm}
\caption{Greedy balancing protocol}
\label{greedy_balancing}
\begin{algorithmic}[1]
\For{$h=1...H-1$}
    \State Match the paths using the natural matching
    \State Send the order of the output paths to the next node.
\EndFor
\end{algorithmic}
\end{algorithm}
\fi

\paragraph{Global paths - bottleneck effect}
\ifdouble\Cref{greedy_balancing} \else Algorithm 2 \fi gives a very efficient procedure to determine the global paths. In the following, we show that this solution minimizes the bottleneck effect that arises when two links with different rates are matched.
\off{Looking at each global path independently, it appears that rate is lost each time two segments of the global paths are unbalanced (i.e. having  different rates). Indeed, if the first segment has a greater rate than the second one, more packets arrive at the intermediate node than what the second link is capable of transmitting. On the contrary, in the opposite setting, the second link can send packets at a high rate but incoming packets arrive at a low rate. Therefore, the second link is limited to the rate of the first segment. Hence, in order to match the paths in a balanced way, we thus suggest the following optimization problem.}
\begin{prop}[Balancing Optimization]
Given incoming paths with rate $\textbf{r}_{in}$ and outgoing paths with rate $\textbf{r}_{out}$, the following problems are equivalent.
\begin{eqnarray}
    \textstyle \textbf{l}_1 &=& \rev{\underset{\Tilde{\textbf{l}}\in Perm(P)}{\arg\max}} \sum_{p=1}^P \min\left(\textbf{r}_{in}\left(p\right),\textbf{r}_{out}(\tilde{\textbf{l}}(p))\right),\label{P1}\\
    \textstyle \textbf{l}_2 &=& \rev{\underset{\Tilde{\textbf{l}}\in Perm(P)}{\arg\min}} \sum_{p=1}^P \left|\textbf{r}_{in}\left(p\right)-\textbf{r}_{out}(\tilde{\textbf{l}}(p))\right| \label{P3},
\end{eqnarray}
with $\textbf{r}_{out}(\tilde{\textbf{l}}(p))$ the rate of the $\tilde{\textbf{l}}(p)$-th outgoing path and $Perm(P)$ the set of all permutations\footnote{This restriction prevents non-admissible matchings.} of the vector $[1,2,...,P]$.
\label{prob_v2}
\end{prop}

\off{\textstyle \textbf{l}_2 &=& \rev{\underset{\Tilde{\textbf{l}}\in Perm(P)}{\arg\min}} \sum_{p\in\mathcal{P}^+(\tilde{\textbf{l}})
, $\mathcal{P}^+(\textbf{l}) := \{p|\textbf{r}_{in}(p)-\textbf{r}_{out}(\textbf{l}(p))>0\}$}
\left[\textbf{r}_{in}\left(p\right)-\textbf{r}_{out}(\tilde{\textbf{l}}(p))\right],\label{P2}\\}

In the above proposition, one can recognize on one hand in (\ref{P1}) the problem solved by the natural matching, i.e. the maximization of the sum-rate. On the other hand, (\ref{P3}) minimizes the absolute rate differences (namely, the bottlenecks).

\begin{proof}
Letting $\mathcal{P}^+(\textbf{l}) := \{p|\textbf{r}_{in}(p)-\textbf{r}_{out}(\textbf{l}(p))>0\}$, and $\mathcal{P}^-(\textbf{l})$ be the complementary set, the objective function of (\ref{P3}) is rewritten as
\ifdouble
\begin{equation}
\sum_{p=1}^P \left|\textbf{r}_{in}\left(p\right)-\textbf{r}_{out}(\tilde{\textbf{l}}(p))\right|
 \label{eq:firstSplit}
 \end{equation}
\begin{equation*}
 = \sum_{p\in\mathcal{P}^+(\tilde{\textbf{l}})} \left[\textbf{r}_{in}(p)-\textbf{r}_{out}(\tilde{\textbf{l}}(p))\right]-\sum_{p\in\mathcal{P}^-(\tilde{\textbf{l}})} \left[\textbf{r}_{in}(p)-\textbf{r}_{out}(\tilde{\textbf{l}}(p))\right].
\end{equation*}
\else
\begin{equation}\label{eq:firstSplit}
\sum_{p=1}^P \left|\textbf{r}_{in}\left(p\right)-\textbf{r}_{out}(\tilde{\textbf{l}}(p))\right| = \sum_{p\in\mathcal{P}^+(\tilde{\textbf{l}})} \left[\textbf{r}_{in}(p)-\textbf{r}_{out}(\tilde{\textbf{l}}(p))\right]-\sum_{p\in\mathcal{P}^-(\tilde{\textbf{l}})} \left[\textbf{r}_{in}(p)-\textbf{r}_{out}(\tilde{\textbf{l}}(p))\right].
\end{equation}
\fi
Moreover, since
\begin{equation*}
    \sum_{p=1}^P \left[\textbf{r}_{in}\left(p\right)-\textbf{r}_{out}(\tilde{\textbf{\textbf{l}}}(p))\right]=\sum_{p=1}^P r_{in}\left(p\right)-\sum_{p=1}^P r_{out}(\tilde{\textbf{l}}(p)) = K,
\end{equation*}
with $K$ a constant independent from $\textbf{l}$, the following holds:
\ifdouble
\begin{equation*}
\begin{split}
  \sum_{p\in\mathcal{P}^-(\tilde{\textbf{l}})}\left[ \textbf{r}_{in}\left(p\right)-\textbf{r}_{out}(\tilde{\textbf{l}}(p))\right] \quad \quad \quad \quad \quad \quad \quad \quad \quad  \\
  = K - \sum_{p\in\mathcal{P}^+(\tilde{\textbf{l}})} \left[\textbf{r}_{in}\left(p\right)-\textbf{r}_{out}(\tilde{\textbf{l}}(p))\right].
\end{split}
\end{equation*}
\else
\begin{equation*}
  \sum_{p\in\mathcal{P}^-(\tilde{\textbf{l}})}\left[ \textbf{r}_{in}\left(p\right)-\textbf{r}_{out}(\tilde{\textbf{l}}(p))\right] = K - \sum_{p\in\mathcal{P}^+(\tilde{\textbf{l}})} \left[\textbf{r}_{in}\left(p\right)-\textbf{r}_{out}(\tilde{\textbf{l}}(p))\right].
\end{equation*}
\fi
Hence, (\ref{eq:firstSplit}) is rewritten as
\ifdouble
\begin{equation}
\begin{split}
\sum_{p=1}^P \left|\textbf{r}_{in}\left(p\right)-\textbf{r}_{out}(\tilde{\textbf{l}}(p))\right| \quad \quad \quad \quad \quad \quad \quad \quad \quad \\
 = 2\sum_{p\in\mathcal{P}^+(\tilde{\textbf{l}})} \left[\textbf{r}_{in}(p)-\textbf{r}_{out}(\tilde{\textbf{l}}(p))\right]-K.
\end{split}
\label{equi1}
\end{equation}
\else
\begin{equation}\label{equi1}
\sum_{p=1}^P \left|\textbf{r}_{in}\left(p\right)-\textbf{r}_{out}(\tilde{\textbf{l}}(p))\right| = 2 \sum_{p\in\mathcal{P}^+(\tilde{\textbf{l}})} \left[\textbf{r}_{in}(p)-\textbf{r}_{out}(\tilde{\textbf{l}}(p))\right]-K.
\end{equation}
\fi
Finally, the following holds
\ifdouble
\begin{equation}
\begin{split}
    \sum_{p=1}^P\textbf{r}_{in}(p)-\sum_{p=1}^P \min\left(\textbf{r}_{in}\left(p\right),\textbf{r}_{out}(\tilde{\textbf{l}}(p))\right)\quad \quad \quad \quad \quad \\
 = \sum_{p=1}^P \max\left(0,\textbf{r}_{in}\left(p\right)-\textbf{r}_{out}(\tilde{\textbf{l}}(p))\right)\\
 = \sum_{p\in\mathcal{P}^+(\tilde{\textbf{l}})} \left[\textbf{r}_{in}\left(p\right)-\textbf{r}_{out}(\tilde{\textbf{l}}(p))\right].
\end{split}
\label{equi2}
\end{equation}
\else
\begin{multline}\label{equi2}
\sum_{p=1}^P\textbf{r}_{in}(p)-\sum_{p=1}^P \min\left(\textbf{r}_{in}\left(p\right),\textbf{r}_{out}(\tilde{\textbf{l}}(p))\right)\\ = \sum_{p=1}^P \max\left(0,\textbf{r}_{in}\left(p\right)-\textbf{r}_{out}(\tilde{\textbf{l}}(p))\right) = \sum_{p\in\mathcal{P}^+(\tilde{\textbf{l}})} \left[\textbf{r}_{in}\left(p\right)-\textbf{r}_{out}(\tilde{\textbf{l}}(p))\right].
\end{multline}
\fi
Combining (\ref{equi1}) and (\ref{equi2}), one can observe that maximizing the objective function of (\ref{P1}) is equivalent to minimizing the one of (\ref{P3}), which completes the proof.
\end{proof}
\jrev{The above propositions are valid for scenarios with the same number of links per hop, where each incoming path is matched with exactly one outgoing path. To relax this constraint, the matching procedure could be defined to match subsets of paths with other subsets of paths of the following hop, possibly with a different number of paths per subset. One possible solution would be to consider at each hop the minimum number of paths (denoted by $n_{\min}$) between the incoming and outgoing links. Then, the matching procedure could be defined between the $n_{\min}$ paths on the one hand, and all possible path partitions of size $n_{\min}$ on the other hand. \jrevm{For instance, in a two hop network with two paths followed by three paths, two global paths could be defined. One of them would be singlepath while the second one would be made of one path for the first hop, and two paths for the second.  To handle the varying number of paths of the global paths, each intermediate node would receive and send linear combinations along the global paths in a traditional RLNC way. A rigorous definition of such a matching procedure}, as well as the analysis of other matching schemes is an interesting direction for future work.}

\paragraph{Selective mixing}
The AC-RLNC MP protocol can be applied on the balanced global paths. Yet, mixing packets between some of the paths can improve the algorithm. \rev{From} \Cref{ex_balancing}, it can be seen on one hand that defining global paths reduces the maximum throughput. On the other hand, mixing all packets between all the paths suppresses the usefulness of FEC and FB-FEC transmissions. But if at intermediate nodes, new packets are mixed together on one hand and the FEC's and FB-FEC's on the other hand, the throughput will be increased without increasing the delay.
\subsection{Analytical Results for Delay and Throughput}\label{Analytical_results_MH}
In the following, the in-order delivery delay and throughput are analyzed in the case of the MP-MH network. The bounds derived in \Cref{Analytical_results_MP} are generalized by noticing that once the matching between the paths is defined, from the sender's point of view, the network is equivalent to a MP network with rates\off{\footnote{This result is valid only in the asymptotic regime, that is if large window sizes are considered. Since the bounds are derived considering the transmission of $\bar{o} = 2RTT$ packets, that hypothesis hold for moderate to high $RTT$'s.}}
\begin{equation*}
     r_{G_p} = \min_{h=1...H}r_{G(p,h)h}, \quad \forall p=\{1;\ldots;P\},
\end{equation*}
where $r_{G_p}$ denotes the rate of the $p$-th global path.

~~\\\subsubsection{An Upper Bound for the Throughput}
Once the rates of the global paths are defined, the minimum Bhattacharyya distance given in \Cref{BhattacharyyaD} can be used to bound the achievable rate, and hence directly leads to \Cref{theoremMH} and \Cref{BEC_MH}.

\ifdouble
\begin{figure}
\centering
\includegraphics[trim=0cm 0.7cm 0cm 0cm,width=1\columnwidth]{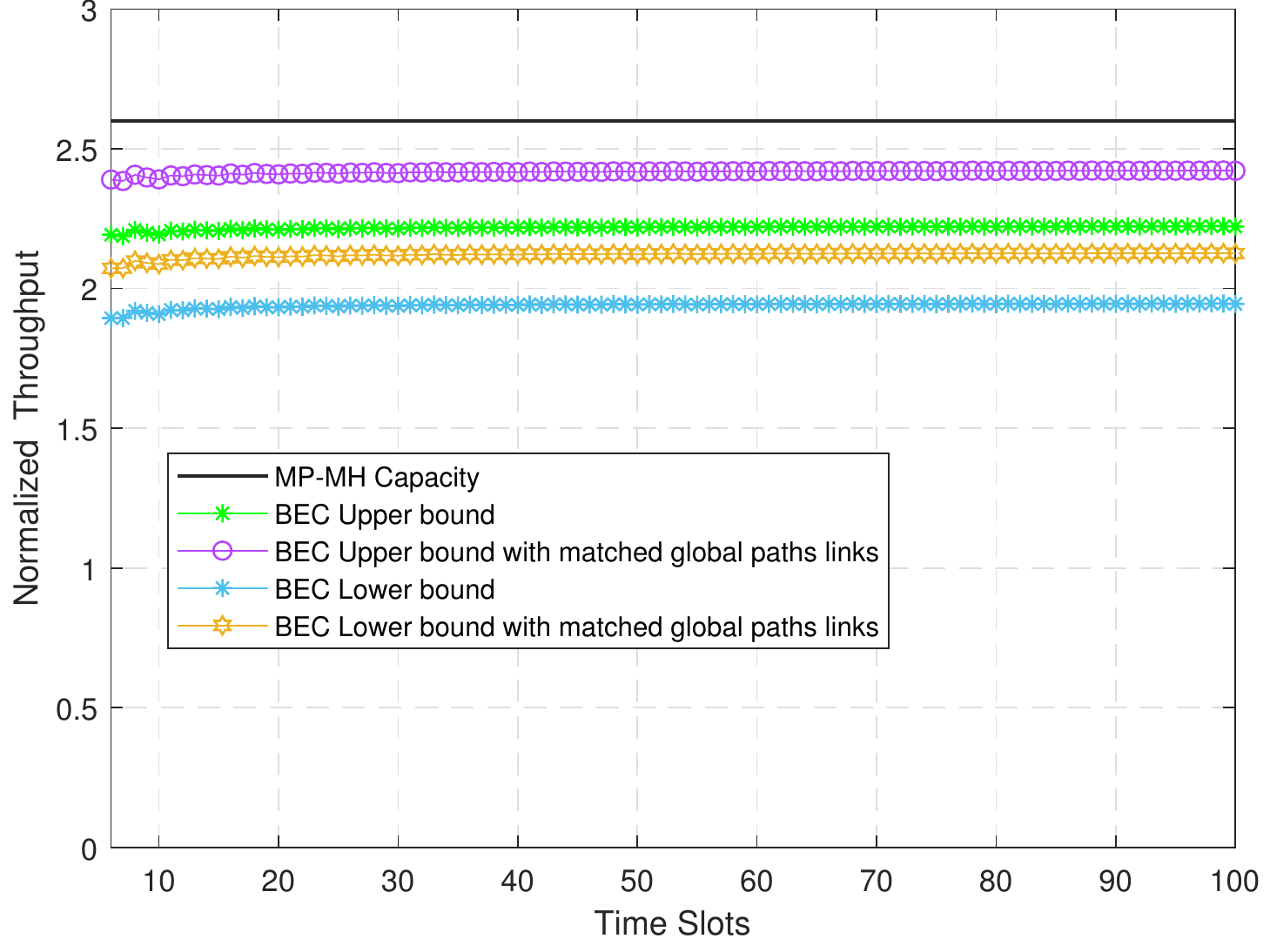}
	\caption{\jrev{Throughput upper and lower bounds in MP-MH network with $H=3$ and $P=4$ for BEC's with erasure probabilities of \Cref{ex_balancing}. Note that the range of the abscissa is from $4$ (the theoretical minimum of the $RTT$ delay\protect\footnotemark[17]) to $100$. Moreover, the throughput is not degraded by increasing the $RTT$. In the asymptotic regime the MP-MH AC-RLNC code may attain the capacity.
	In purple, the throughput upper bound is shown assuming all the paths are matched perfectly, without bottlenecks. That is, when the following rates of \Cref{ex_balancing} are changed to  $r_{21}=0.4$, $r_{31}=0.6$, $r_{33}=0.4$ and $r_{43}=0.6$.}}
	\label{fig:ThroughputErrorBecMPMH}
\end{figure}
\else
\begin{figure}
\centering
\subfigure[]{\includegraphics[trim=0cm 0.7cm 0cm 0cm,width=0.4\columnwidth]{MPMH_BhattacharyyaDistance2_lower_1}}
\subfigure[]{\includegraphics[trim=0cm 0.7cm 0cm 0cm,width=0.5\columnwidth]{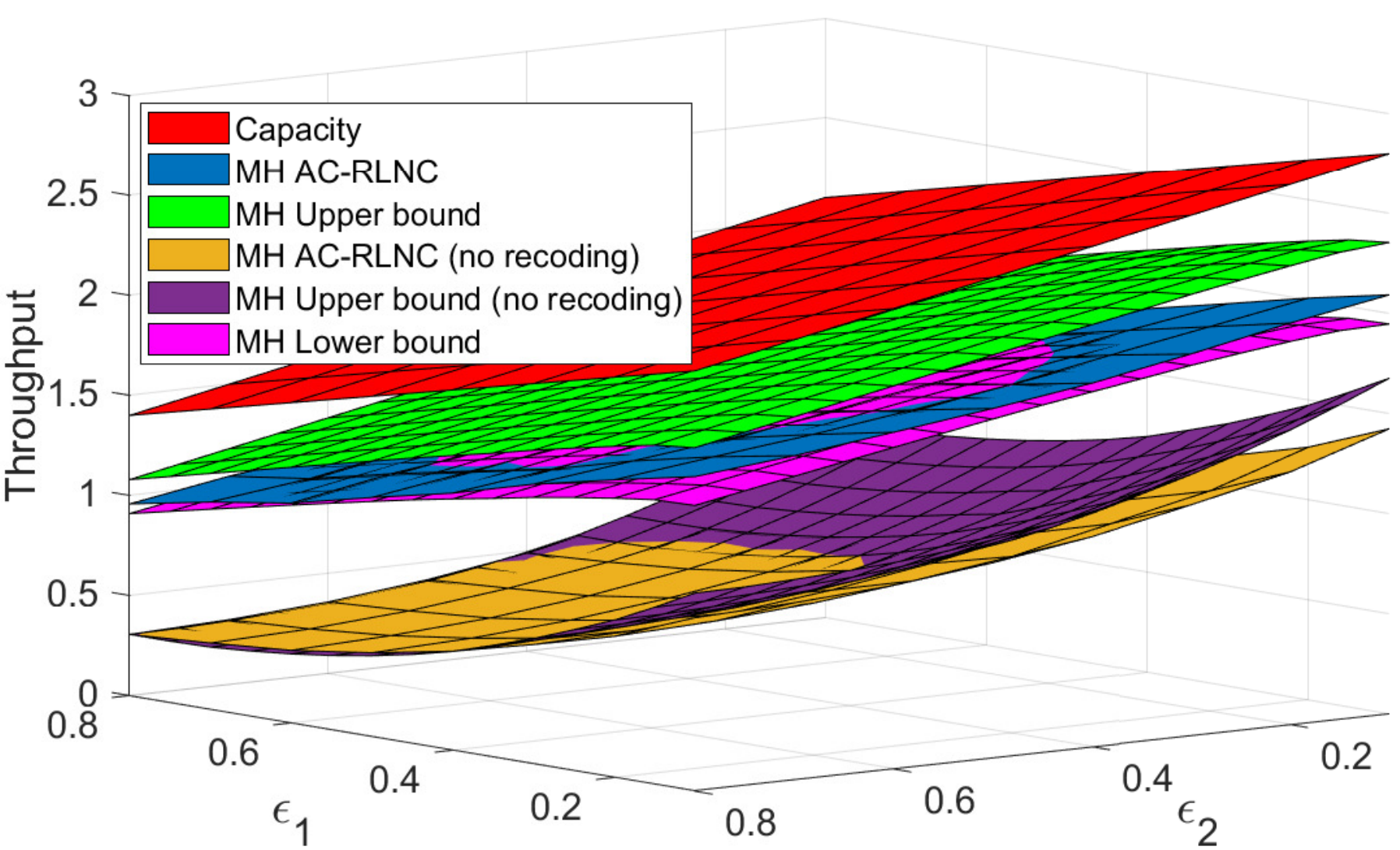}}
\subfigure[]{\includegraphics[trim=0cm 0.7cm 0cm 0cm,width=0.5\columnwidth]{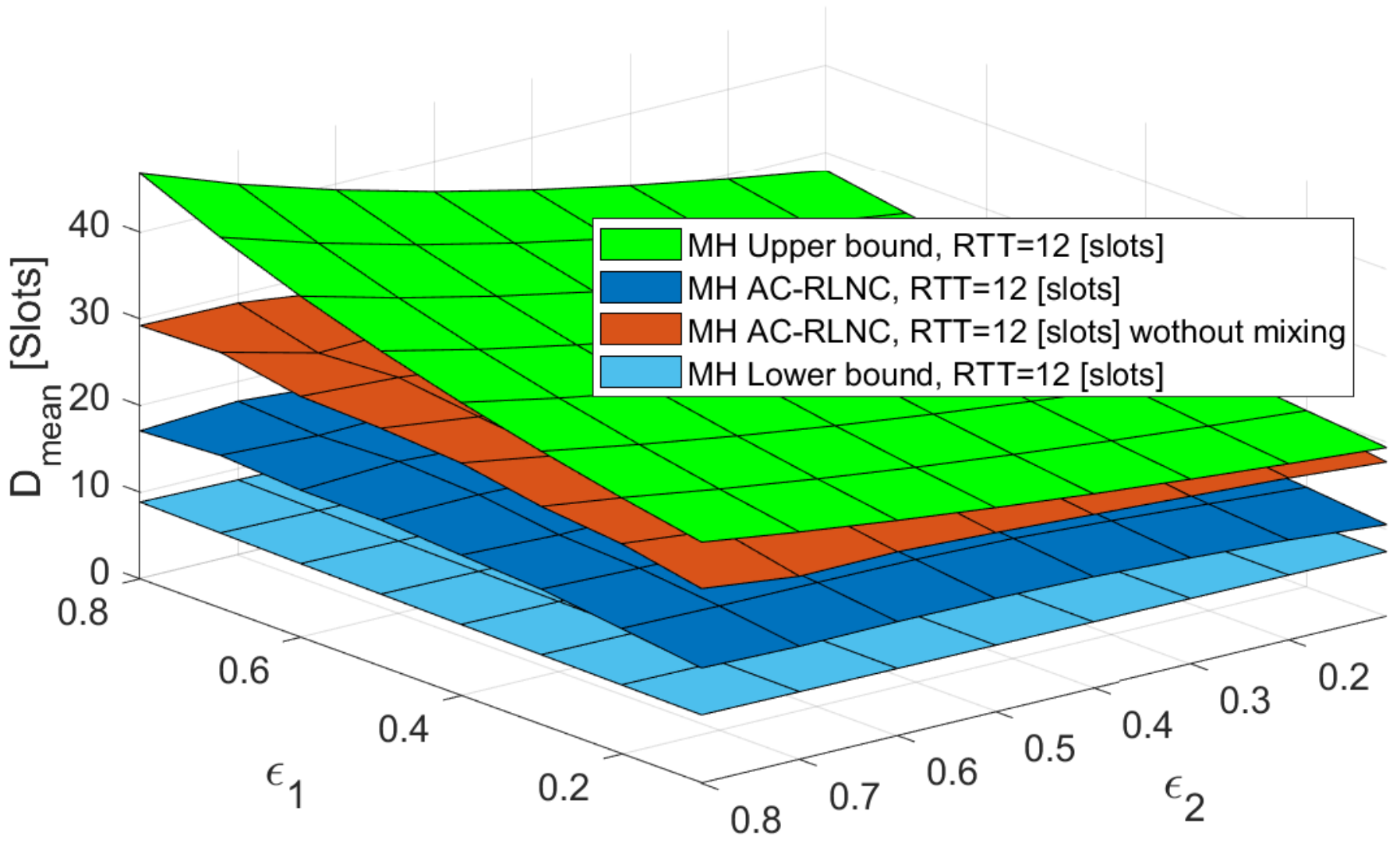}}
	\caption{(a) \jrev{Throughput upper and lower bounds for BEC in MP-MH network with $H=3$, $P=4$, $RTT=12$, and with
	$[\epsilon_{11}=\epsilon_{1}$ $\epsilon_{12}=0.6$ $\epsilon_{13}=0.3$; $\epsilon_{21}=0.8$ $\epsilon_{22}=\epsilon_{1}$ $\epsilon_{23}=\epsilon_{1}$; $\epsilon_{31}=0.2$ $\epsilon_{32}=\epsilon_{2}$ $\epsilon_{33}=0.7$; $\epsilon_{41}=\epsilon_{2}$ $\epsilon_{42}=0.4$ $\epsilon_{43}=\epsilon_{2}]$, while the erasure probabilities of $\epsilon_{1}$ and $\epsilon_{2}$ vary in the range of $[0.1 \;  0.8]$.} (b) {\jrev{Throughput upper and lower bounds for BEC in MP-MH network with $H=3$, $P=4$, $RTT=12$, and with
	$[\epsilon_{11}=\epsilon_{1}$ $\epsilon_{12}=0.6$ $\epsilon_{13}=0.3$; $\epsilon_{21}=0.8$ $\epsilon_{22}=\epsilon_{1}$ $\epsilon_{23}=\epsilon_{1}$; $\epsilon_{31}=0.2$ $\epsilon_{32}=\epsilon_{2}$ $\epsilon_{33}=0.7$; $\epsilon_{41}=\epsilon_{2}$ $\epsilon_{42}=0.4$ $\epsilon_{43}=\epsilon_{2}]$, while the erasure probabilities of $\epsilon_{1}$ and $\epsilon_{2}$ vary in the range of $[0.1 \;  0.8]$.}} (c)	\jrev{Mean in-order delivery delay upper bounds for BEC in MP-MH network with $H=3$, $P=4$, $RTT=12$, and with
	$[\epsilon_{11}=\epsilon_{1}$ $\epsilon_{12}=0.6$ $\epsilon_{13}=0.3$; $\epsilon_{21}=0.8$ $\epsilon_{22}=\epsilon_{1}$ $\epsilon_{23}=\epsilon_{1}$; $\epsilon_{31}=0.2$ $\epsilon_{32}=\epsilon_{2}$ $\epsilon_{33}=0.7$; $\epsilon_{41}=\epsilon_{2}$ $\epsilon_{42}=0.4$ $\epsilon_{43}=\epsilon_{2}]$, while the erasure probabilities of $\epsilon_{1}$ and $\epsilon_{2}$ vary in the range of $[0.1 \;  0.8]$.}}
	\label{fig:ThroughputErrorBecMPMH}
\end{figure}
\fi
\footnotetext[17]{This theoretical minimum is given considering that the processing at the nodes is immediate and that the feedback message is of negligible size. Hence, in this case, it takes three slots for the direct message to reach the receiver and one slot for the feedback message.}

\begin{theorem}\jrev{An upper bound on the throughput of AC-RLNC in MP-MH network is} \label{theoremMH}
\begin{align*}
    \eta\leq \sum_{p=1}^{P} \jrevm{\left(r_{G_p}(t^-) - l(r_{G_p}(t),r_{G_p}(t^-))\right)},
\end{align*}
where $l(\cdot,\cdot)$ is the Bhattacharyya distance.
\end{theorem}
\begin{proof}
Given the rate of each global path, $r_{G_p}$ $\forall p\in\{1;\ldots;P\}$, the proof of the upper bound on the throughput in the MP-MH network is a direct consequence of \Cref{theoremMP}
\end{proof}
\begin{corollary}{\bf BEC.}\label{BEC_MH}
An upper bound on the throughput of AC-RLNC for BEC in the MP-MH network is
\begin{eqnarray*}
    \eta^{BEC}\leq \sum_{p=1}^{P} \jrevm{\left(r_{G_p}^{BEC}(t) - l(r_{G_p}^{BEC}(t),r_{G_p}^{BEC}(t^-))\right)}.
\end{eqnarray*}
\end{corollary}

\ifdouble
\begin{figure}
\centering
\includegraphics[trim=0cm 0.4cm 0cm 0cm,width=1\columnwidth]{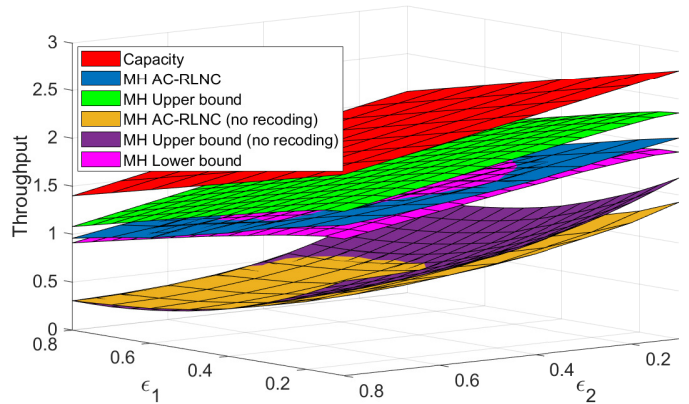}
	\caption{\jrev{Throughput upper and lower bounds for BEC in MP-MH network with $H=3$, $P=4$, $RTT=12$, and with
	$[\epsilon_{11}=\epsilon_{1}$ $\epsilon_{12}=0.6$ $\epsilon_{13}=0.3$; $\epsilon_{21}=0.8$ $\epsilon_{22}=\epsilon_{1}$ $\epsilon_{23}=\epsilon_{1}$; $\epsilon_{31}=0.2$ $\epsilon_{32}=\epsilon_{2}$ $\epsilon_{33}=0.7$; $\epsilon_{41}=\epsilon_{2}$ $\epsilon_{42}=0.4$ $\epsilon_{43}=\epsilon_{2}]$, while the erasure probabilities of $\epsilon_{1}$ and $\epsilon_{2}$ vary in the range of $[0.1 \;  0.8]$.}}
	\label{fig:ThroughputMPMH}
\end{figure}
\fi

\begin{figure}
\centering
\ifdouble
\includegraphics[trim=0cm 0.7cm 0cm 0cm,width=1\columnwidth]{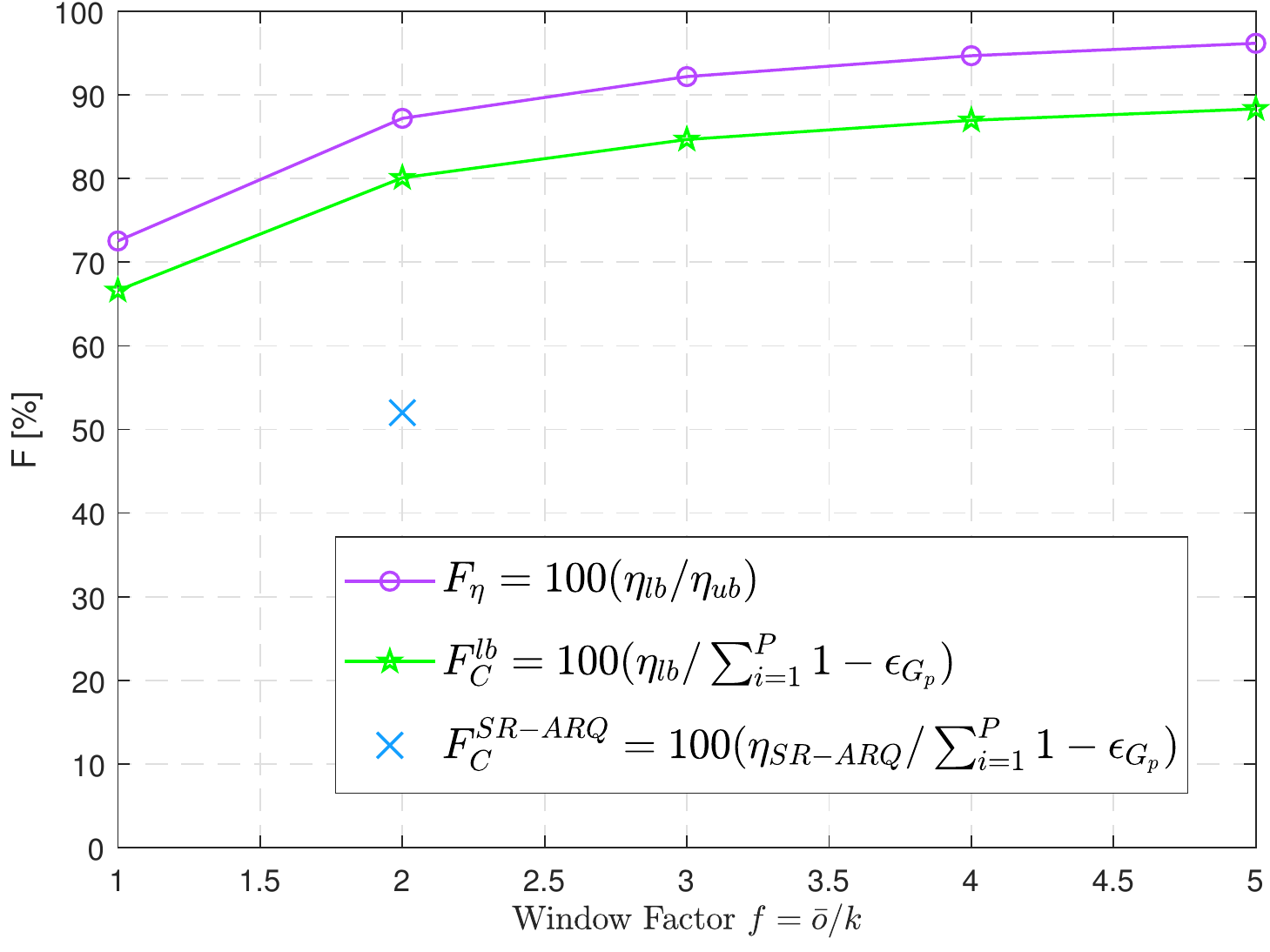}
\else
\includegraphics[trim=0cm 0.7cm 0cm 0cm,width=0.45\columnwidth]{MPMH_factor_rate2}
\fi
	\caption{\jrev{Throughput factor in MP-MH network with $H=3$, $P=4$ and $RTT=12$ for BEC channels with erasure probabilities of \Cref{ex_balancing}. Note that the performance ratios get closer to $100\%$ when the window size limit increases, hence showing that the upper and lower bounds get tighter for larger windows. In this figure, the solid lines correspond to the bounds while the SR-ARQ point comes from the numerical validation of the protocol.}}
	\label{fig:ThroughpuFactorMPMH}
\end{figure}

In \ifdouble \Cref{fig:ThroughputErrorBecMPMH}\else \Cref{fig:ThroughputErrorBecMPMH}-(a)\fi, the throughput upper bound of AC-RLNC protocols in a MP-MH network with $H=3$ and $P=4$ is shown as function of $RTT$ for BEC with erasure probabilities corresponding to \Cref{ex_balancing}. We can note that for this specific network, the MP-MH suggested solution could achieve around $85\%$ of the BEC capacity. With regards to the MP network, the MP-MH bound is worse as rate is spoilt due to the bottleneck effect. Indeed, even considering an optimal matching, residual bottlenecks remain, thus decreasing the achievable rate. Modifying the rates of \Cref{ex_balancing} to $r_{21}=0.4$, $r_{31}=0.6$, $r_{33}=0.4$ and $r_{43}=0.6$, bottlenecks are entirely removed, leading to  the purple result in \ifdouble \Cref{fig:ThroughputErrorBecMPMH}\else \Cref{fig:ThroughputErrorBecMPMH}-(a)\fi.
\\The results of \ifdouble\Cref{fig:ThroughputMPMH} \else \Cref{fig:ThroughputErrorBecMPMH}-(b) \fi allow to compare the upper bound to the actual achieved throughput. Comparing the green curve to the blue one, one can see that the achieved throughput is close to the upper bound. In yellow and purple, the results are shown when no recoding is performed at the intermediate nodes. Namely, packets are forwarded from hop to hop. The achieved results are much worse, both for the upper bound and the simulations results. This can be explained by the fact that, when packets are forwarded, the rate of each global path $\tilde{r}_{G_p}$ is not longer the minimum rate but the product of the rate of each links, i.e.,  $\tilde{r}_{G_p} = \prod_{h=1...H}r_{G(p,h)h}$, $\forall p=\{1;\ldots;P\}$, as in this setting, a packet is received only if all the transmissions of the global path are successful.

\subsubsection{\jrev{A Lower Bound for the Throughput}}\label{SecThroughputLowMH}
\jrev{Once the rates of the global paths are defined, the bounds derived in \Cref{SecThroughputLowMP} can be directly generalized to obtain a lower bound for the throughput in the MH-MP network.}

\jrev{\begin{theorem}{A lower bound on the throughput of AC-RLNC in MP-MH network is} \label{theoremMHlb}
\begin{align*}
    \eta_{lb} \geq \sum_{p=1}^{P}  \jrevm{\left(r_{G_p,up} - \frac{ n^{\text{EW}}_{G_p}}{n^{w}_{G_p}}\right)}.
\end{align*}
\end{theorem}}
\begin{proof}
\jrev{Given the rate of each global path, $r_{G_p}$ $\forall p\in\{1,\ldots,P\}$, the proof of the upper bound on the throughput in the MP-MH network is a direct consequence of \Cref{theoremMPlb}.}
\end{proof}
\jrev{Similarly to the MP network, \ifdouble \Cref{fig:ThroughputErrorBecMPMH} \else \Cref{fig:ThroughputErrorBecMPMH}-(a) \fi compares the capacity to the upper and lower bounds. One can observe that increasing the RTT does not decrease the performances, and that matching the paths allows to obtain results closer to the capacity. In \ifdouble\Cref{fig:ThroughputMPMH}\else \Cref{fig:ThroughputErrorBecMPMH}-(b)\fi, these bounds are validated numerically for a 3-hop 4-path network with a RTT of $12$ time slots. One can clearly observe the gain obtained with path matching. Finally, in \Cref{fig:ThroughpuFactorMPMH}, the performance factors $F_{\eta}^{\text{AC-RLNC}} = 100 \cdot \frac{\eta_{lb}}{\eta_{ub}}$ and $F_{\text{capacity}}^{\text{AC-RLNC}} = 100 \cdot \frac{\eta_{lb}}{\sum_{i=1}^{P} 1-\epsilon_p}$ compare the lower bound to the upper bound and the capacity respectively. From this figure, one can observe that for larger window size limit, the capacity may be achieved while the upper and lower bounds get tighter\jrevm{\footref{foot:numerical_evaluation}}. The numerical validation of the SR-ARQ protocol, for this network, gives results that are at $26\%$ of the capacity when applied to one unique path and at $52\%$ when applied to all global paths (as one can check from \Cref{MH_perf}).}

\ifdouble
\begin{figure}
\centering
\includegraphics[trim=0cm 0.7cm 0cm 0cm,width=1\columnwidth]{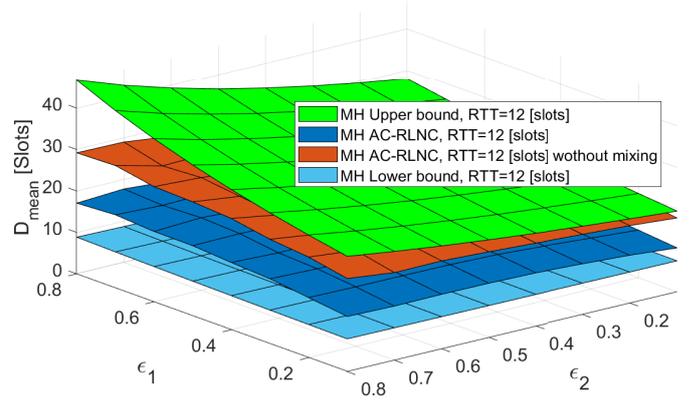}
	\caption{\jrev{Mean in-order delivery delay upper bounds for BEC in MP-MH network with $H=3$, $P=4$, $RTT=12$, and with
	$[\epsilon_{11}=\epsilon_{1}$ $\epsilon_{12}=0.6$ $\epsilon_{13}=0.3$; $\epsilon_{21}=0.8$ $\epsilon_{22}=\epsilon_{1}$ $\epsilon_{23}=\epsilon_{1}$; $\epsilon_{31}=0.2$ $\epsilon_{32}=\epsilon_{2}$ $\epsilon_{33}=0.7$; $\epsilon_{41}=\epsilon_{2}$ $\epsilon_{42}=0.4$ $\epsilon_{43}=\epsilon_{2}]$, while the erasure probabilities of $\epsilon_{1}$ and $\epsilon_{2}$ vary in the range of $[0.1 \;  0.8]$.}}
	\label{fig:meam_delay_MPMH}
\end{figure}
\else\fi

\begin{figure}
\centering
\ifdouble
\includegraphics[trim=0cm 0.4cm 0cm 0cm,width=1\columnwidth]{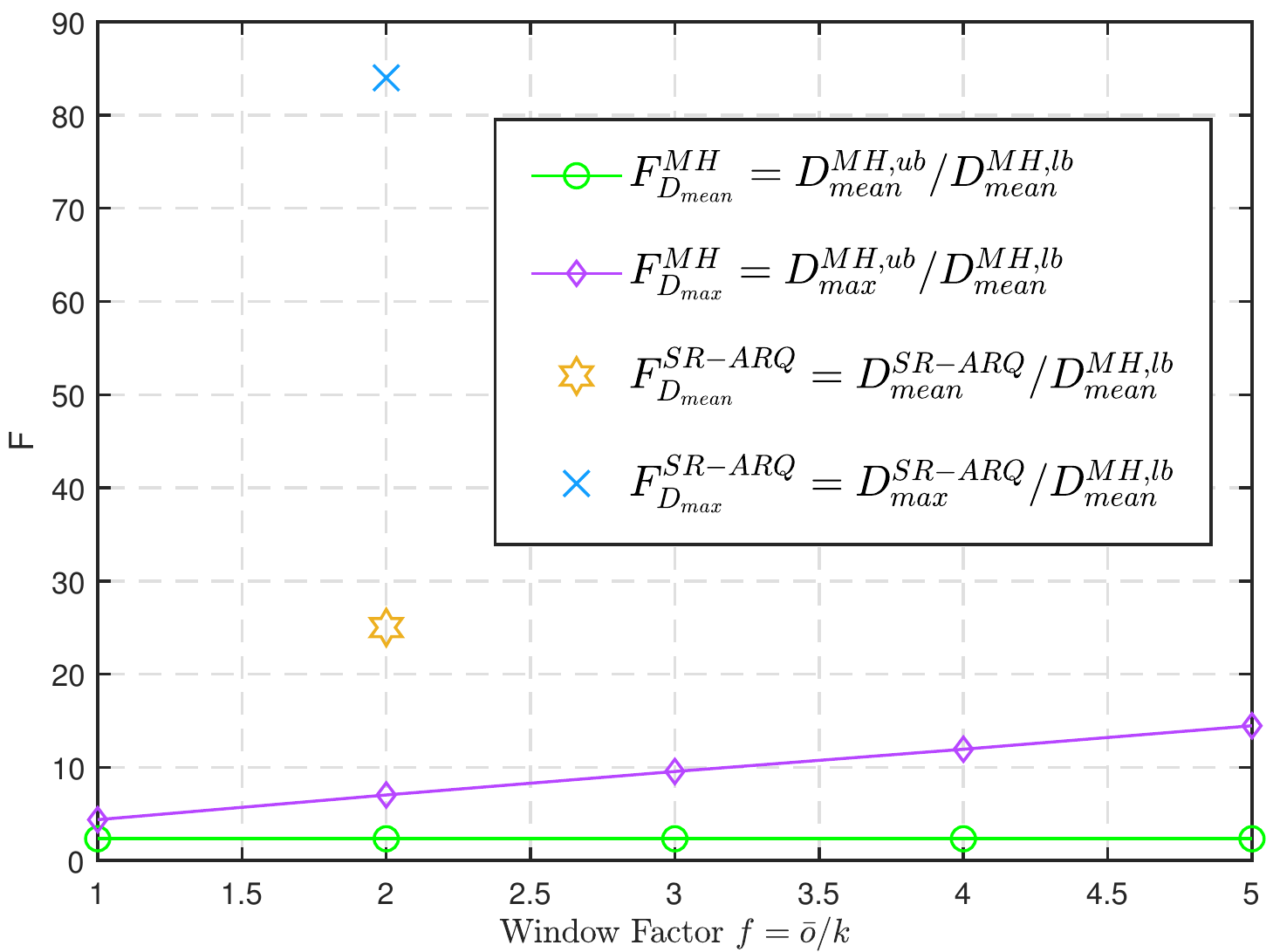}
\else
\includegraphics[trim=0cm 0.7cm 0cm 0cm,width=0.45\columnwidth]{MPMH_delay_factor3}
\fi
	\caption{\jrev{Delay factor in MP-MH network with $H=3$, $P=4$ and $RTT=12$ for BEC's with erasure probabilities of \Cref{ex_balancing}. Note that as for the multipath case, the mean delay bound does not depend on this window factor while the maximum one increases linearly with it. In this figure, the solid lines correspond to the bounds while the SR-ARQ points come from the numerical validation of the protocol.}}
	\label{fig:max_delay_MPMH_lb}
\end{figure}

\ifdouble
\begin{figure*}
\centering
\includegraphics[trim=0cm 0.0cm 0cm 0cm,width=1\textwidth]{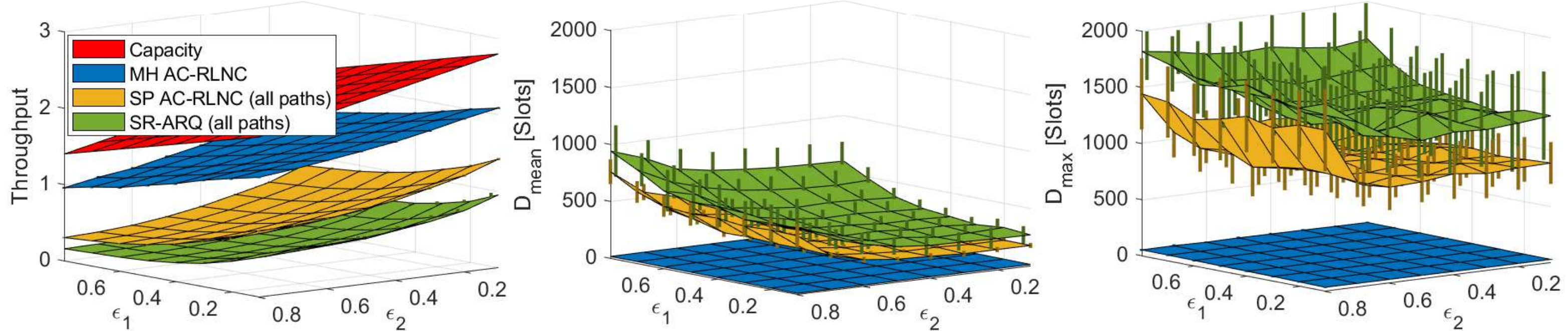}
\includegraphics[trim=0cm 0.0cm 0cm 0cm,width=1\textwidth]{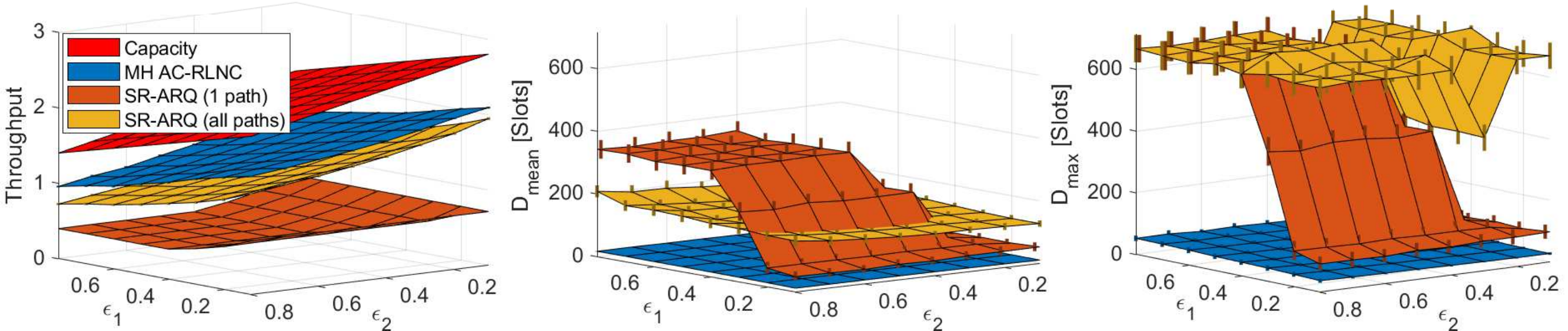}
    \caption{Performances of the multi-hop protocol with $RTT=12$ [slots], $th=0$ and $\bar{o}=2k$, on 4 paths, averaged on $150$ iterations. The top results are for end-to-end feedback (from the receiver to the sender), while the bottom results are for hop-by-hop feedback (from each intermediate node and receiver). The vertical bars correspond to the standard deviation of the simulated results.}
    \label{MH_perf}
\end{figure*}
\fi

~~\\\subsubsection{An Upper Bound for the Mean and Maximum In-Order Delivery Delay}
For both the mean and maximum in-order delivery delay, the analyses of \Cref{Analytical_results_MP} still hold when considering
$$\bar{\epsilon}_{MH} = \frac{1}{P} \sum_{p=1}^{P}\left(1- r_{G_p} \right) ,$$
instead of $\bar{\epsilon}$. Indeed, as discussed above, from the sender's point of view, the matched MP-MH network is equivalent to a MP network with $r_{G_p}$ as path's rate.
\jrevm{Thus, the upper bounds for the mean and max in-order delay in the MP-MH network, as given in the Theorems below, are a direct consequence from the proofs given in \Cref{Analytical_results_MP}.}
\jrevm{
\begin{theorem}
The mean in-order delivery delay $D_{\text{mean}}$ of AC-RLNC in MP-MH network is upper-bounded as:
\ifdouble
\begin{multline}
D_{mean} \leq \lambda D_{mean}^{no\mbox{ }feedback}+(1-\lambda )\\
(D_{mean}^{nack\mbox{ }feedback} +D_{mean}^{ack\mbox{ }feedback} ),
\end{multline}
\else
\begin{equation}
D_{mean} \leq \lambda D_{mean}^{no\mbox{ }feedback}+(1-\lambda )\\
(D_{mean}^{nack\mbox{ }feedback} +D_{mean}^{ack\mbox{ }feedback} ),
\end{equation}
\fi
with $\lambda$ denoting the fraction of time without feedback compared to the total time of transmission, and with the mean in-order delivery delay in case of the different feedback states being bounded by \eqref{no_FB_BEC}, \eqref{ACK_BEC}, \eqref{NACK_BEC} with $\bar{\epsilon}_{MH}$ instead of $\bar{\epsilon}$.
\end{theorem}}
\jrevm{\begin{theorem}
The maximum in-order delivery delay $D_{\max}$ of AC-RLNC in MP-MH network is upper-bounded with error probability $P_e$ as:
\begin{equation}
 D_{\max}\leq \left \lceil \frac{RTT}{2} \right \rceil +\left \lceil \frac{T_{\max}}{P}  \right \rceil,
\end{equation}
where $T_{\max}$ is computed through \eqref{eq:Tmax1} with $\bar{\epsilon}_{MH}$ instead of $\bar{\epsilon}$.
\end{theorem}
}
In \ifdouble \Cref{fig:meam_delay_MPMH}\else \Cref{fig:ThroughputErrorBecMPMH}-(c)\fi, the mean in-order delivery delay bound (in green) is compared with two results of the MP-MH AC-RLNC protocol. One (in blue) is for the case simulated in \Cref{MHsimulation}, in which the intermediate node can recode and mix the retransmission packets between the global paths, as described in the selective mixing paragraph in \Cref{MHCodeConstruction}. The second (in orange) is for the case where intermediate nodes are not mixing the retransmission packets between the global paths. Hence, each global path is independent as in the MP network. One can notice that by mixing the retransmission packets between the global paths at the intermediate nodes, we can obtain in practice a significantly lower mean in-order delay.
\ifdouble\else
\begin{figure*}
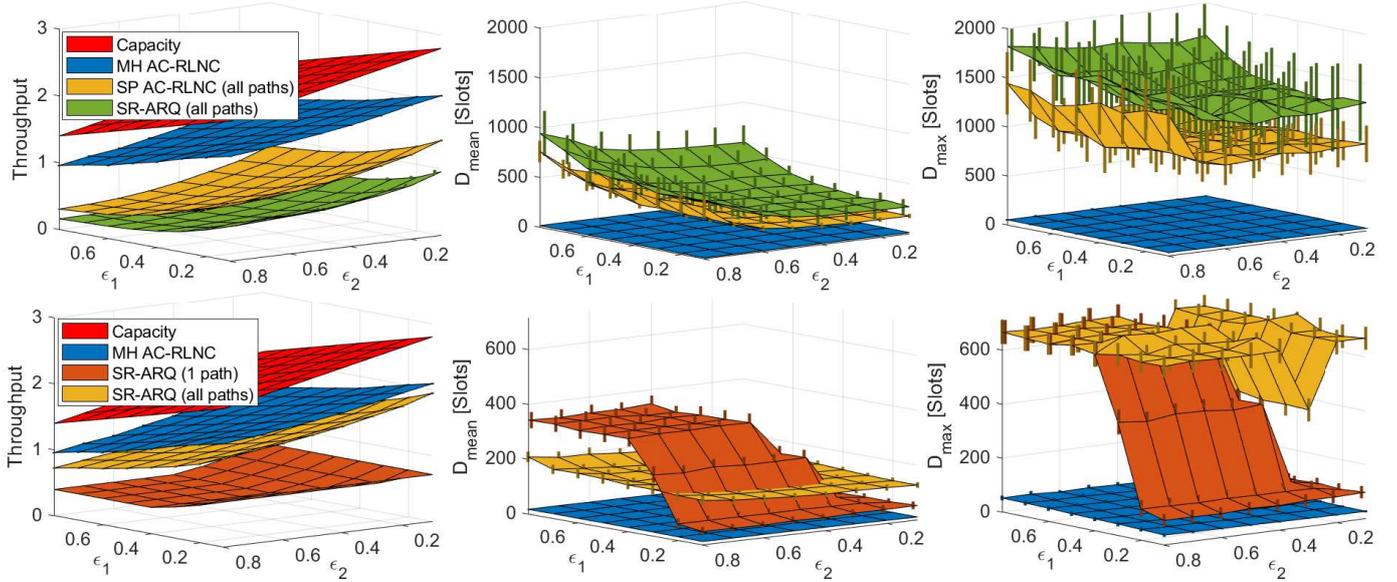

\centering
\includegraphics[trim=0cm 0.0cm 0cm 0cm,width=1\textwidth]{MH_figure_paper_End2End_2}
\includegraphics[trim=0cm 0.0cm 0cm 0cm,width=1\textwidth]{MH_figure_paper_v3}
    \caption{Performances of the multi-hop protocol with $RTT=12$ [slots], $th=0$ and $\bar{o}=2k$, on 4 paths, averaged on $150$ iterations. The top results are for end-to-end feedback (from the receiver to the sender), while the bottom results are for hop-by-hop feedback (from each intermediate node and receiver). The vertical bars correspond to the standard deviation of the simulated results.}
    \label{MH_perf}
\end{figure*}
\fi

\subsubsection{\jrev{Lower Bounds for the Mean and Max In-Order Delivery Delay}}\label{mean_delay_mh_lb}
\jrev{For the multihop network, lower bounds on the delay can be obtained as for the multipath one, by considering the rate of each global path, as done for the throughput. \jrevm{Hence, the results of \Cref{Analytical_results_MP} can be extended to the MP-MH network as follows.}
\jrevm{\begin{theorem}
The in-order delivery delay $D$ of AC-RLNC in MP-MH network is lower-bounded as
\begin{equation*}
    D^{\text{genie-aided}} \geq \frac{\text{RTT}}{2}+\frac{1}{1-\bar{\epsilon}_{MH}},
\end{equation*}
if different packets are sent on each of the paths.
\end{theorem}}
This leads to the results of \ifdouble \Cref{fig:meam_delay_MPMH} \else \Cref{fig:ThroughputErrorBecMPMH}-(c) \fi where the optimal genie-aided lower bound can be compared to the performance of our MP-MH AC-RLNC solution. In \Cref{fig:max_delay_MPMH_lb}, the mean and maximum delay upper bounds are compared to the optimal genie-aided lower bound. As for the multipath, the mean delay is not impacted by the window size factor while the maximum one increases with this size. The in-order delay of the SR-ARQ protocol in this network is within a $25$ factor for the mean one when applied on all global paths while the maximum one is close to $84$ times bigger than the bound.  }

\off{\jrev{In the same manner, once the rates of the global paths are defined, the bounds derived in \Cref{mean_delay_mp_lb} and \Cref{max_delay_mp_lb}, can be directed generalized to obtain lower bounds for the mean and max in-order delivery delay in the MH-MP network. Hence,
\begin{equation*}
    D_{mean}^{lb} \geq \frac{RTT/2}{1-\bar{\epsilon}_{MH}},
\end{equation*}
end
\begin{equation*}
    D_{max}^{lb}\left(t^{f}\right) \geq f \frac{RTT/2}{1-\bar{\epsilon}_{MH}}.
\end{equation*}
See \Cref{fig:meam_delay_MPMH} and \Cref{fig:max_delay_MPMH_lb}. Need to explain the results.... Comparison with approximation factor:  $F_{D_{mean}}^{\text{AC-RLNC}} = \frac{D_{mean}^{ub}}{D_{mean}^{lb}}$ (Note that this factor is equal to 2, probably since our solution is causal from the feedback acknowledgments. Hence we now have the effect of the feedback channel delay, which is in our analysis precisely the same as the delay in the forward channel, $RTT/2$) with $F_{D_{mean}}^{\text{SR-ARQ}} = \frac{D_{mean,\text{SR-ARQ}}}{D_{mean}^{lb}}$ and $F_{D_{max}}^{\text{AC-RLNC}} = \frac{D_{max}^{ub}}{D_{mean}^{lb}}$ with $F_{D_{max}}^{\text{SR-ARQ}} = \frac{D_{max,\text{SR-ARQ}}}{D_{mean}^{lb}}$ (for SR-ARQ we need to add example from the simulations).}}

\subsection{Simulation Results}\label{MHsimulation}
The performances of the MP-MH AC-RLNC protocol are compared with two other protocols, as presented in \Cref{MH_perf}.

\paragraph{Setting and protocols}
\ar{We consider the setting of \Cref{MH_MP_setting},} \ar{with H=3, P=4, with}
\ar{\begin{equation*}
    \textstyle \begin{pmatrix}
    \epsilon_{11} & \epsilon_{12} & \epsilon_{13}\\
    \epsilon_{21} & \epsilon_{22} & \epsilon_{23}\\
    \epsilon_{31} & \epsilon_{32} & \epsilon_{33}\\
    \epsilon_{41} & \epsilon_{42} & \epsilon_{43}
    \end{pmatrix}
    =
    \textstyle \begin{pmatrix}
    \epsilon_{1} & 0.6 & 0.3\\
    0.8 & \epsilon_{1} & \epsilon_{1}\\
    0.2 & \epsilon_{2} & 0.7\\
    \epsilon_{2} & 0.4 & \epsilon_{2}
    \end{pmatrix},
\end{equation*}
with $\epsilon_{1}$ and $\epsilon_{2}$ varying in the range of $[0.1 \;  0.8]$. The results are shown for a $RTT$ delay of $12$ time slots.}

The proposed protocol has been simulated with $th=0$ and $\bar{o}=2w$. First, we compare in the upper graph of \Cref{MH_perf} our solution with end to end protocols. Specifically, as for the multipath case, the SP AC-RLNC protocol and the SR-ARQ protocol are applied independently on the $P$ balanced global paths we get with \ifdouble\Cref{greedy_balancing}\else Algorithm 2\fi. However, no recoding is performed at the intermediate nodes. Secondly, in order \rev{to not} depend on node recoding, we investigate in the  lower graph of \Cref{MH_perf} the performances of the  hop by hop SR-ARQ protocol on two different settings. In one setting, it is applied to one single global path that is build from the best path of each hop. In the other, it is applied on the $P$ balanced paths.

The metrics have been averaged on 150 different channel realizations, where the filled curves correspond to the mean performances while the error bars represent the standard deviation, as for the MP results.

\paragraph{Results}
From \Cref{MH_perf}, one can see that the MP-MH AC-RLNC protocol performs dramatically better both with regards to the rate and in-order delay. As expected, in the upper graph, one can see that the rate is improved a lot with regards to end to end protocols. The rate is \rev{doubled} for good channel conditions ($\epsilon_1=\epsilon_2=0.1$) and multiplied by $3$ for bad \jrev{ channel conditions} ($\epsilon_1=\epsilon_2=0.8$),  for the SP AC-RLNC algorithm. Compared with the SP SR-ARQ protocols, performances are even better. Yet, the improvements are even more dramatic from the delay point of view. \rev{The} mean and max in-order delay are reduced by a factor $15$ for good channel conditions and up to a factor $40$ for bad ones, for both the AC-ARLNC and the SR-ARQ protocols.  From the lower part of \Cref{MH_perf}, it can be seen that the hop by hop SR-ARQ protocol is also much worse than our solution. \rev{The} improvement we get on the rate is small ($10\%$) for good channel conditions but it becomes significant ($35\%$) when channels are bad. From the delay point of view, the gain is \rev{significant} since both the mean and max in-order delay are reduced approximately by a factor 10. The hop by hop SR-ARQ protocol on only 1 path has obviously a much \rev{lower} rate. For the in-order delay, the improvements we get highly depend on the channel configuration (changing slightly the erasure probabilities gives very different results) but the MH AC-RLNC protocol is still better independently of the configuration of the channel.


\section{Conclusion}\label{Conc}
We proposed a MP-MH adaptive and causal coding algorithm for communications with delayed feedback. The MP algorithm, and especially the combination of the a priori and a posteriori mechanisms with efficient bit-filling packet allocation, outperforms significantly the existing protocols. By tuning the parameters of the algorithm, the desired throughput-\ifdouble\\\else\fi delay trade-off is obtained. The achieved trade-offs are in agreement with the theoretical analyses of the throughput and in-order delivery delay.
Splitting the MH protocol into the balancing optimization and the MP algorithm gives very promising results: on one hand, a decentralized matching procedure is proven to be optimal with regards to the achieved rate. On the other hand, theoretical analyses of the achieved throughput and delays are also in agreement with the simulations. Compared to other protocols, MP-MH AC-RLNC protocol has a higher throughput, and a much lower delay than SR-ARQ. Specifically, in the end-to-end setting without recoding in the intermediate nodes, it reaches a very high delay and a lower throughput. The hop-by-hop SR-ARQ is less impacted by the absence of recoding. Nevertheless, this comes at the price of a high sensitivity to the channel configuration and the order of the hops, while each node needs to be able to perform the full SR-ARQ protocol.\off{, as it can be seen from the abrupt increase in the delay, and from the non-symmetric shape of the max delay. For large RTT's, the insertion of FEC's at the end of each window of new transmissions might be sub-optimal. Therefore, optimizing the timing of FEC's insertions is an interesting direction for future improvements of the protocol.}

Future work \jrev{includes} the study of general mesh networks \jrev{and settings with multiple sources and receivers where fairness and resource allocation are needed. Furthermore, the model constraints, such as the same number of paths per hop, the reliable feedback, or the same propagation delay on each hop, are simplifications with regards to realistic networks. The relaxation of these constraints is also an interesting lead for future work, while ongoing research focuses on the implementation of this scheme in the framework of the QUIC protocol \cite{swett2017network,swett2018coding,8816838}.}
\section*{Acknowledgment}
We want to thank Maxence Vigneras for contributing in the MATLAB implementation of the suggested solution, and to Derya Malak and Arno Schneuwly for fruitful discussions.
\appendices
\section{Proof of \Cref{natural_matching}}
\label{appendix:appendixMatching}
The proof of \Cref{natural_matching} is continued here. We proved above that the natural matching is optimal at the last hop. Now, following the same ideas as the above proof, we show that the natural matching is also optimal for the previous hops, hence leading to the desired result.
\\\textbf{$(\textbf{H-1)}$-th hop : }
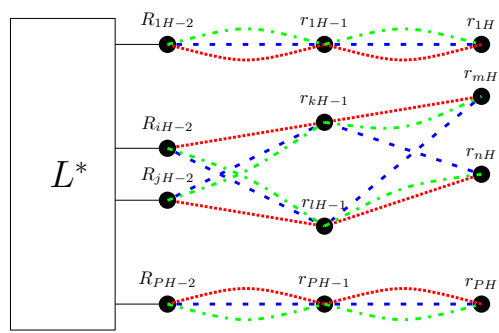
\begin{figure}[!h]
\centering
\resizebox{6.6cm}{!}{\input{Matching_hop.tikz}}
\caption{Matching of the $(H-1)$-th hop. The green dash-dotted lines show the non-natural matching when $r_{lH-1} \leq R_{jH-2}$, the blue dashed lines show the non-natural matching when $r_{lH-1} > R_{jH-2}$ and the red dotted lines highlight the natural one.}\label{fig:matching_hop}
\end{figure}
Suppose the matching of the previous hop (between $R_{H-2}$ and $r_{H-1}$) is not the natural one, as shown in \Cref{fig:matching_hop} with blue dashed and green dash-dotted lines.  Hence, one can find $i<j$, $k<l$ such that $R_{iH-2}$ (resp. $R_{jH-2}$) is matched with $r_{lH-1}$ (resp. $r_{kH-1}$). Moreover, let $n>m$ (implying $r_{mH}\geq r_{nH}$) corresponds to the corresponding rates of the last hop.
\\If $r_{lH-1} \leq R_{jH-2}$, then $R_{lH-1} \leq R_{kH-1}$. Thus, building on the above results for the last hop, the optimal matching of the last hop matches $r_{kH-1}$ (resp. $r_{lH-1}$) with $r_{mh}$ (resp. $r_{nH}$), as represented with green dash-dotted lines in \Cref{fig:matching_hop}.
In that case,
\begin{equation*}
\begin{split}
        \min \left(R_{jH-2},r_{kH-1},r_{mH}\right)+\min \left(R_{iH-2},r_{lH-1},r_{nH}\right)\\
        \leq \min \left(R_{iH-2},r_{kH-1},r_{mH}\right)+\min \left(R_{jH-2},r_{lH-1},r_{nH}\right).
\end{split}
\end{equation*}
\\If $ R_{jH-2}<r_{lH-1} $, then $R_{lH-1} \geq R_{kH-1}$. Thus, the optimal matching of the last hop matches $r_{kH-1}$ (resp. $r_{lH-1}$) with $r_{nh}$ (resp. $r_{mH}$), as represented with blue dashed lines in \Cref{fig:matching_hop}.
In that case,
\begin{equation*}
\begin{split}
        \min \left(R_{jH-2},r_{kH-1},r_{nH}\right)+\min \left(R_{iH-2},r_{lH-1},r_{mH}\right)\\
        \leq \min \left(R_{jH-2},r_{lH-1},r_{nH}\right)+\min \left(R_{iH-2},r_{kH-1},r_{mH}\right).
\end{split}
\end{equation*}
In both cases, the natural matching, represented with red dotted lines in \Cref{fig:matching_hop}, has a rate greater or equal to the one of $L^*$. Applying this reasoning till no $i,j,k,l$ can be found, it appears the last two hops can be matched naturally without decreasing the achieved rate.
\\\textbf{Previous hops : }
Since the proof can be applied recursively to each hop until the first one is reached, proving recursively that the rate do not decrease when using the natural matching instead of the optimal one, the following is obtained
\begin{equation*}
    \eta_{max}\left(L^*\right) \leq \eta_{max}\left(L\right),
\end{equation*}
contradicting \eqref{to_contradict}. Hence, the natural matching is proven to be optimal.

\bibliographystyle{IEEEtran}
\bibliography{references}
\end{document}

%% file: Matching_last_hop.tikz
\begin{tikzpicture}[transform shape]

\draw[draw=black] (0,0.5) rectangle ++(2,6);

\node[] at (1,3.5) {\begin{huge} $L^*$ \end{huge}};

\draw (2, 1) -- (3, 1);
\draw (2, 3) -- (3, 3);
\draw (2, 4) -- (3, 4);
\draw (2,6) -- (3, 6);

\node[circle,fill=black,label=above:{$R_{PH-1}$}] at (3, 1)   (n21) {};
\node[circle,fill=black,label=above:{$R_{jH-1}$}] at (3, 3)   (n22) {};
\node[circle,fill=black,label=above:{$R_{iH-1}$}] at (3, 4)   (n23) {};
\node[circle,fill=black,label=above:{$R_{1H-1}$}] at (3, 6)   (n24) {};

\node[circle,fill=black,label=above:{$r_{PH}$}] at (6,1) (n31) {};
\node[circle,fill=black,label=above:{$r_{lH}$}] at (6, 2.5)   (n32) {};
\node[circle,fill=black,label=above:{$r_{kH}$}] at (6, 4.5)   (n33) {};
\node[circle,fill=black,label=above:{$r_{1H}$}] at (6, 6)   (n34) {};

\draw[ultra thick, draw=blue,loosely dashed] (n21) -- (n31) ;
\draw[ultra thick, draw=red,densely dotted] (3, 1) .. controls (4.5,1.4) .. (6,1);
\draw[ultra thick, draw=red,densely dotted] (n22) -- (n32) ;
\draw[ultra thick, draw=blue,loosely dashed] (n22) -- (n33) ;
\draw[ultra thick, draw=blue,loosely dashed] (n23) -- (n32) ;
\draw[ultra thick, draw=red,densely dotted] (n23) -- (n33) ;
\draw[ultra thick, draw=blue,loosely dashed] (n24) -- (n34) ;
\draw[ultra thick, draw=red,densely dotted] (3, 6) .. controls (4.5,5.4) .. (6,6);

\end{tikzpicture}

%% file: Matching_hop.tikz
\begin{tikzpicture}[transform shape]

\draw[draw=black] (0,0.5) rectangle ++(2,6);

\node[] at (1,3.5) {\begin{huge} $L^*$ \end{huge}};

\draw (2, 1) -- (3, 1);
\draw (2, 3) -- (3, 3);
\draw (2, 4) -- (3, 4);
\draw (2,6) -- (3, 6);

\node[circle,fill=black,label=above:{$R_{PH-2}$}] at (3, 1)   (n21) {};
\node[circle,fill=black,label=above:{$R_{jH-2}$}] at (3, 3)   (n22) {};
\node[circle,fill=black,label=above:{$R_{iH-2}$}] at (3, 4)   (n23) {};
\node[circle,fill=black,label=above:{$R_{1H-2}$}] at (3, 6)   (n24) {};

\node[circle,fill=black,label=above:{$r_{PH-1}$}] at (6,1) (n31) {};
\node[circle,fill=black,label=above:{$r_{lH-1}$}] at (6, 2.5)   (n32) {};
\node[circle,fill=black,label=above:{$r_{kH-1}$}] at (6, 4.5)   (n33) {};
\node[circle,fill=black,label=above:{$r_{1H-1}$}] at (6, 6)   (n34) {};

\draw[ultra thick, draw=blue,loosely dashed] (n21) -- (n31) ;
\draw[ultra thick, draw=red,densely dotted] (3, 1) .. controls (4.5,1.4) .. (6,1);
\draw[ultra thick, draw=green,loosely dashdotted] (3,1) .. controls (4.5,0.6) .. (6,1);
\draw[ultra thick, draw=red,densely dotted] (n22) -- (n32) ;
\draw[ultra thick, draw=blue,loosely dashed] (n22) -- (n33) ;
\draw[ultra thick, draw=green,loosely dashdotted] (3,3) .. controls (4.5,3.5) .. (6,4.5);
\draw[ultra thick, draw=blue,loosely dashed] (n23) -- (n32) ;
\draw[ultra thick, draw=red,densely dotted] (n23) -- (n33) ;
\draw[ultra thick, draw=green,loosely dashdotted] (3,4) .. controls (4.5,3.5) .. (6,2.5);
\draw[ultra thick, draw=blue,loosely dashed] (n24) -- (n34) ;
\draw[ultra thick, draw=red,densely dotted] (3, 6) .. controls (4.5,5.6) .. (6,6);
\draw[ultra thick, draw=green,loosely dashdotted] (3,6) .. controls (4.5,6.4) .. (6,6);

\node[circle,fill=black,label=above:{$r_{PH}$}] at (9,1) (n41) {};
\node[circle,fill=black,label=above:{$r_{nH}$}] at (9, 3.5)   (n42) {};
\node[circle,fill=black,label=above:{$r_{mH}$}] at (9, 5)   (n43) {};
\node[circle,fill=black,label=above:{$r_{1H}$}] at (9, 6)   (n44) {};

\draw[ultra thick, draw=blue,loosely dashed] (n41) -- (n31) ;
\draw[ultra thick, draw=red,densely dotted] (6, 1) .. controls (7.5,1.4) .. (9,1);
\draw[ultra thick, draw=green,loosely dashdotted] (6,1) .. controls (7.5,0.6) .. (9,1);
\draw[ultra thick, draw=red,densely dotted] (n42) -- (n32) ;
\draw[ultra thick, draw=blue,loosely dashed] (n42) -- (n33) ;
\draw[ultra thick, draw=green,loosely dashdotted] (6,2.5) .. controls (7.5,3.3) .. (9,3.5);
\draw[ultra thick, draw=blue,loosely dashed] (n43) -- (n32) ;
\draw[ultra thick, draw=red,densely dotted] (n43) -- (n33) ;
\draw[ultra thick, draw=green,loosely dashdotted] (6,4.5) .. controls (7.5,4.3) .. (9,5);
\draw[ultra thick, draw=blue,loosely dashed] (n44) -- (n34) ;
\draw[ultra thick, draw=red,densely dotted] (6, 6) .. controls (7.5,5.6) .. (9,6);
\draw[ultra thick, draw=green,loosely dashdotted] (6,6) .. controls (7.5,6.4) .. (9,6);
\end{tikzpicture}